\documentclass{scrartcl}

\usepackage[T1]{fontenc}
\usepackage[utf8]{inputenc}
\usepackage{amsmath, amsthm, amssymb}
\usepackage{dsfont}
\usepackage{graphicx}
\usepackage[format=plain]{caption}
\usepackage{capt-of}
\usepackage{subcaption}
\usepackage{bm}
\usepackage[alphabetic]{amsrefs}

\newcommand{\norm}[1]{\left\| #1 \right\|}  
\newcommand{\scprd}[1]{\left\langle #1 \right\rangle}  
\newcommand{\ddt}[1]{\frac{d} {dt} #1}   
\newcommand{\N}{\mathbb{N}}  
\newcommand{\Z}{\mathbb{Z}}
\newcommand{\Q}{\mathbb{Q}}
\newcommand{\R}{\mathbb{R}}
\newcommand{\C}{\mathbb{C}}

\newcommand{\ran}{\operatorname{ran}}
\newcommand{\rk}{\operatorname{rk}}

\newcommand{\spn}{\operatorname{span}} 
\newcommand{\dist}{\operatorname{dist}}
\newcommand{\dom}{\operatorname{dom}}

\DeclareMathOperator*{\esssup}{ess-sup}  
\newcommand{\eps}{\varepsilon}
\renewcommand{\phi}{\varphi}
\newcommand{\ul}{\underline}
\newcommand{\ol}{\overline}
\newcommand{\tr}{\operatorname{tr}}

\renewcommand{\Re}{\operatorname{Re}}

\numberwithin{equation}{section}

\newtheorem{thm}{Theorem}[section]
\newtheorem{cor}[thm]{Corollary}
\newtheorem{prop}[thm]{Proposition}
\newtheorem{lm}[thm]{Lemma}
\newtheorem{cond}[thm]{Condition}

\theoremstyle{definition} \newtheorem{ex}[thm]{Example}
\theoremstyle{definition}

\title{Adiabatic theorems for general linear operators with time-independent domains}

\author{Jochen Schmid\\  
\small Institut f\"ur Mathematik, Universit\"at W\"urzburg, 97074 W\"urzburg, Germany\\
\small jochen.schmid@mathematik.uni-wuerzburg.de}   

\date{}

\begin{document}

\maketitle

\begin{abstract}
\small{ \noindent 
We establish adiabatic theorems with and without spectral gap condition for general -- typically dissipative -- 
linear operators $A(t): D(A(t)) \subset X \to X$ with time-independent domains $D(A(t)) = D$ in some Banach space $X$. Compared to the previously known adiabatic theorems -- especially those without spectral gap condition -- we do not require the considered spectral values $\lambda(t)$ of $A(t)$ to be (weakly) semisimple. We also impose only fairly weak regularity conditions. 
Applications are given to slowly time-varying open quantum systems and to adiabatic switching processes. 
}
\end{abstract}

{ \small \noindent \emph{Subject classification (2010) and key words:} 34E15, 34G10, 35Q41, 47D06, 81Q12, 81S22 
\\
Adiabatic theorems for general linear operators, dissipative operators, time-independent domains, non-semisimple spectral values, spectral gap, open quantum systems, adiabatic switching
}

\section{Introduction} \label{sect: intro}

%



Adiabatic theory -- or, more precisely, time-adiabatic theory for linear operators with time-independent domains -- 
is concerned with slowly time-varying systems described by 
evolution equations
\begin{align} \label{eq: awp, adtheorie 0}
x' = A(\eps s) x \quad (s \in [s_0,1/\eps]) \quad \text{and} \quad x(s_0) = y, 
\end{align}
where $A(t): D(A(t)) \subset X \to X$ for $t \in [0,1]$ is a densely defined closed linear operator with time-independent domain $D(A(t)) = D$ in a Banach space $X$ and where $\eps \in (0,\infty)$ is some (small) slowness parameter.
Smaller and smaller values of $\eps$ mean that $A(\eps s)$ depends more and more slowly on time $s$ or, in other words, that the typical time 
where $A(\eps \,.\,)$ varies 
appreciably gets larger and larger. 
%
Such slowly time-varying systems arise, for instance, when an electric or magnetic potential is slowly switched on or in approximate molecular dynamics (in the context of the Born--Oppenheimer approximation).
%
It is common and convenient in adiabatic theory to rescale time as $t = \eps s$ and to consider the equivalent rescaled evolution equations 
\begin{align} \label{eq: awp, adtheorie}
x' = \frac 1 \eps A(t) x \quad (t \in [t_0,1]) \quad \text{and} \quad x(t_0) = y
\end{align}
with initial times $t_0 \in (0,1]$ and initial values $y \in D$. 
It is further assumed 
that these evolution equations are well-posed, that is, for every initial time $t_0 \in (0,1]$ and every initial value $y \in D$ the initial value problem~\eqref{eq: awp, adtheorie} has a unique classical solution $x_{\eps}(\,.\,,t_0,y)$ and $x_{\eps}(\,.\,,t_0,y)$ 
continuously depends on $t_0$ and $y$. A bit more concisely and conveniently, the well-posedness of~\eqref{eq: awp, adtheorie} can be characterized by the existence of a unique so-called evolution system $U_{\eps}$ for $\frac 1 \eps A$ on $D$, that is, a two-parameter family of bounded solution operators $U_{\eps}(t,t_0)$ in $X$ determined by $U_{\eps}(t,t_0)y = x_{\eps}(t,t_0,y)$ for $y \in D$ and $t_0 \le t$. 
\smallskip


Adiabatic theory is further concerned with curves of spectral values $\lambda(t) \in \sigma(A(t))$, mostly eigenvalues, 
of the operators $A(t)$. In the classical special case of skew-adjoint operators $A(t)$ (that is, operators of the form $1/i$ times a self-adjoint operator $A_0(t)$), these spectral values $\lambda(t) = 1/i \, \lambda_0(t)$ could correspond to the ground-state energy $\lambda_0(t)$ of $A_0(t)$, for instance.
If $\lambda(t)$ is isolated in the spectrum $\sigma(A(t))$ of $A(t)$ for every $t \in [0,1]$, one speaks of a spectral gap. And such a spectral gap, in turn, is called uniform or non-uniform depending on whether or not 
\begin{align} \label{eq: ausglage und -frage, glm sl}
\inf_{t \in [0,1]} \dist \big( \lambda(t), \sigma(A(t))\setminus \{\lambda(t)\}  \big) > 0.
\end{align}
Some typical spectral situations are illustrated below for the special case of skew-adjoint operators $A(t)$: the spectrum $\sigma(A(t))$ is plotted on the vertical axis $i \R$ against the horizontal $t$-axis and the red line represents the considered spectral values $\lambda(t)$. In the first two figures, we have a spectral gap which is uniform in the first and non-uniform in the second picture. And the third figure depicts a situation without spectral gap.

\vfill
\begin{figure}[htbp]%
\centering
\begin{subfigure}[b]{0.3\textwidth}
\centering
\includegraphics[width=\columnwidth]{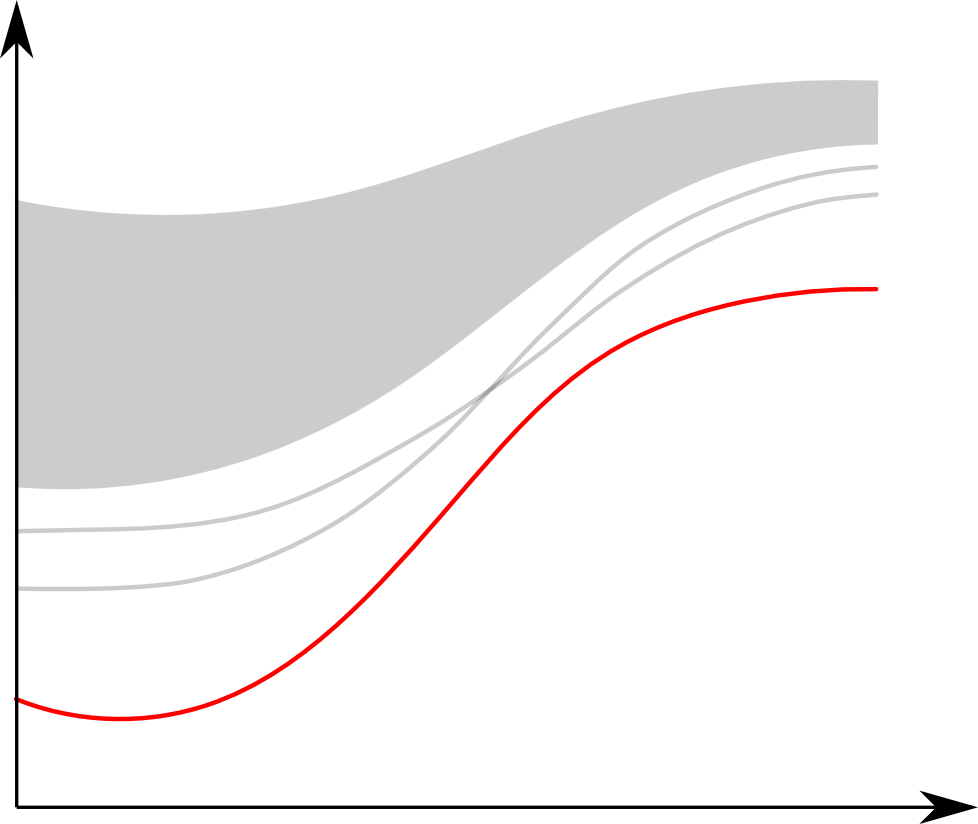}
\end{subfigure}
\hfill
\begin{subfigure}[b]{0.3\textwidth}
\centering
\includegraphics[width=\columnwidth]{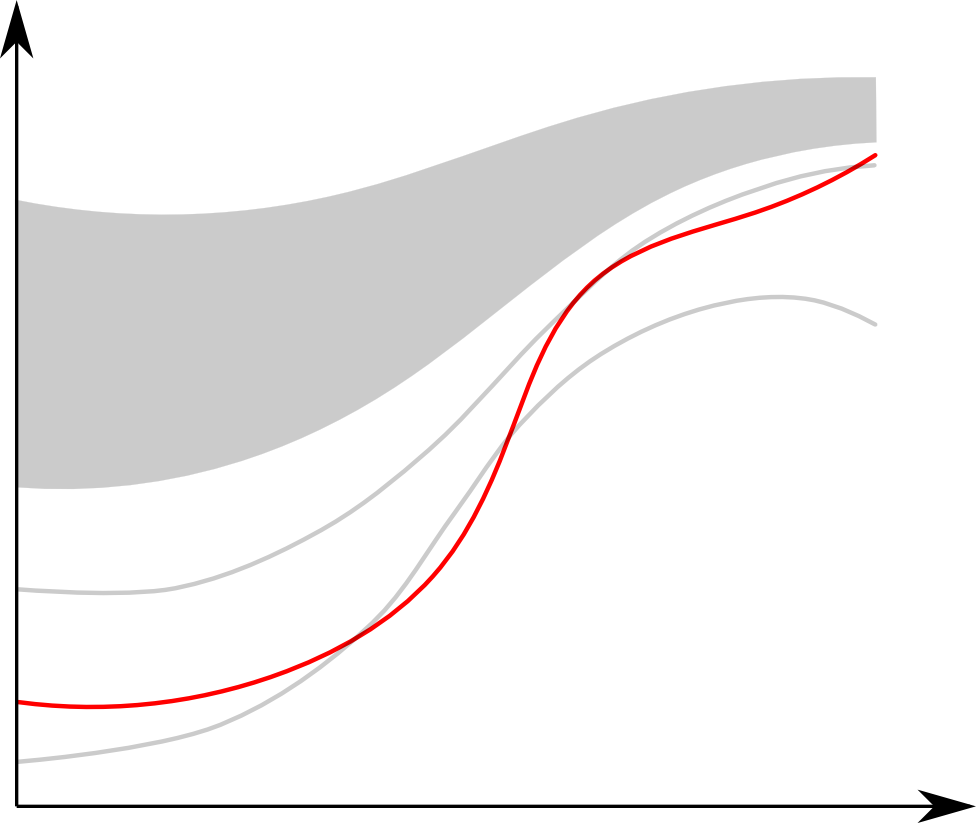}
\end{subfigure}
\hfill
\begin{subfigure}[b]{0.3\textwidth}
\centering
\includegraphics[width=\columnwidth]{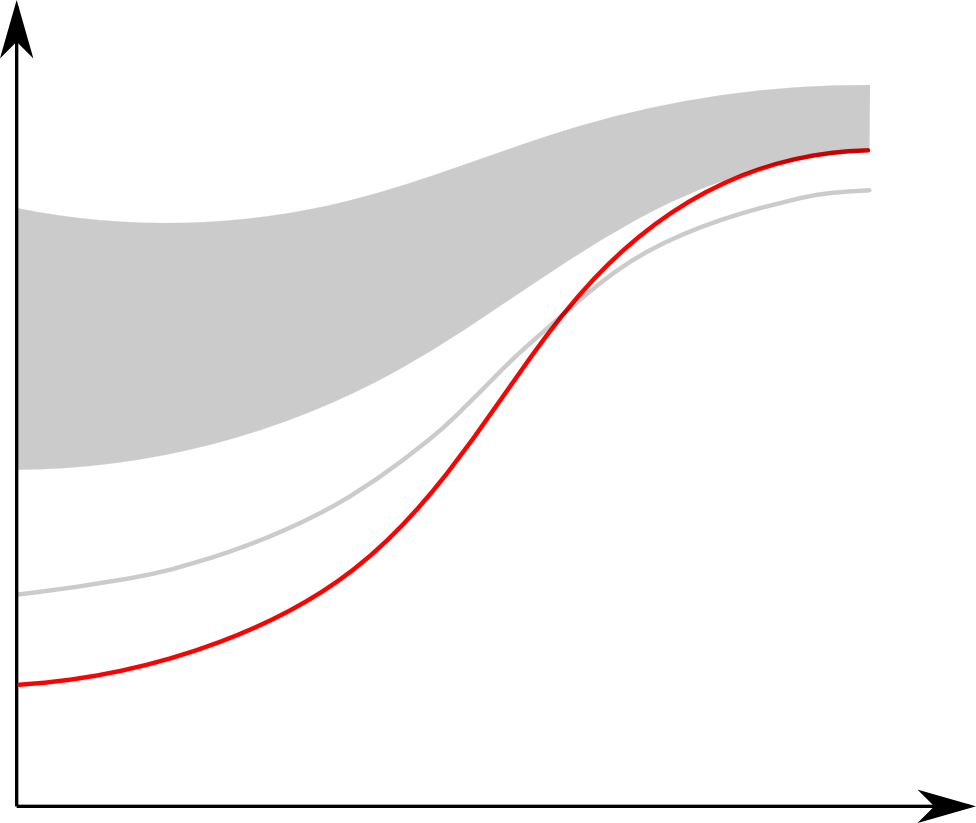}
\end{subfigure}
\end{figure}

%


What adiabatic theory is interested in is how certain distinguished solutions to~\eqref{eq: awp, adtheorie} behave in the singular limit where the slowness parameter $\eps$ tends to $0$. In more specific terms, the basic goal of adiabatic theory can be described -- for skew-adjoint and then for general operators -- as follows. 
In the special case of skew-adjoint operators $A(t)$, one wants to show that for small values of $\eps$ and every $t$ the solution operator $U_{\eps}(t,0)$ 
takes eigenvectors of $A(0)$ corresponding to $\lambda(0)$ into eigenvectors of $A(t)$ corresponding to $\lambda(t)$ -- up to small errors in $\eps$. Shorter and more precisely, one wants to show that 
\begin{align}  \label{eq: aussage des adsatzes, schiefsa}
(1-P(t)) U_{\eps}(t,0) P(0) \longrightarrow 0 \qquad (\eps \searrow 0)
\end{align}
for all $t \in [0,1]$, where $P(t)$ for (almost) every $t$ is the canonical spectral projection of $A(t)$ corresponding to $\lambda(t)$. 
It is defined via the spectral measure $P^{A(t)}$ of $A(t)$, namely $P(t) = P^{A(t)}(\{\lambda(t)\})$, and it is the orthogonal projection yielding the decomposition of $X$ into $P(t)X = \ker(A(t)-\lambda(t))$ and $(1-P(t))X = \ol{\ran}(A(t)-\lambda(t))$.
In the case of general operators $A(t)$, one again wants to show that
\begin{align}  \label{eq: aussage des adsatzes}
(1-P(t)) U_{\eps}(t,0) P(0) \longrightarrow 0 \qquad (\eps \searrow 0)
\end{align}
for all $t \in [0,1]$, where now $P(t)$ for (almost) every $t$ is a suitable general 
spectral projection of $A(t)$ corresponding to $\lambda(t)$. 
In the case with spectral gap, suitable 
spectral projections are 
the so-called associated projections, which yield the decomposition
\begin{align} \label{eq: zerl, sl}
P(t)X = \ker(A(t)-\lambda(t))^{m(t)} \quad \text{and} \quad (1-P(t))X = \ran(A(t)-\lambda(t))^{m(t)}
\end{align}
for some $m(t) \in \N$ provided $\lambda(t)$ is a pole of $(\,.\,-A(t))^{-1}$.
In the case without spectral gap, suitable 
spectral projections are 
the so-called weakly associated projections, which yield the decomposition
\begin{align}  \label{eq: zerl, ohne sl}
P(t)X = \ker(A(t)-\lambda(t))^{m(t)} \quad \text{and} \quad (1-P(t))X = \ol{\ran}(A(t)-\lambda(t))^{m(t)}
\end{align}
for some $m(t) \in \N$.
An adiabatic theorem is now simply a theorem that gives conditions on $A(t)$, $\lambda(t)$, $P(t)$ under which the convergence~\eqref{eq: aussage des adsatzes} holds true. A bit more precisely, 
such a theorem should be termed a linear time-adiabatic theorem in contradistinction to the various space-adiabatic theorems and nonlinear adiabatic theorems from the literature (see~\cite{Teufel03} and \cite{Sparber16}, \cite{GangGrech17}, \cite{FrankGang17}, 
for instance). Yet, space-adiabatic theory and nonlinear adiabatic theory will not be touched upon in this paper at all 
and so there is no danger of confusion in our slightly imprecise terminology. Also, adiabatic theory for resonances~\cite{AbouSalemFroehlich07}, \cite{ElgartHagedorn11}, \cite{KelerTeufel12} will not be treated here. 
Sometimes, we will distinguish quantitative and qualitative adiabatic theorems depending on whether they 
give information on the rate of convergence in~\eqref{eq: aussage des adsatzes} or not. 
\smallskip


Adiabatic theory has a long history going back to the first days of quantum theory and many authors have contributed to it since then. 
In the first decades after 1928, all adiabatic theorems were exclusively concerned with skew-adjoint operators $A(t)$ and until 1998 they all required a spectral gap condition. See, for instance, \cite{BornFock28}, \cite{Kato50}, \cite{Lenard59}, \cite{Garrido64}, \cite{Sancho66}, \cite{Nenciu80}, \cite{AvronSeilerYaffe87}, \cite{JoyePfister91}, \cite{JoyePfister93}, \cite{Nenciu93} for the case with spectral gap and \cite{AvronElgart99}, \cite{Bornemann98}, \cite{Teufel01} \cite{FishmanSoffer16}, for instance, for the case without spectral gap.
In the last decade, various adiabatic theorems for more general operators $A(t)$ have been established and again, just like in the special case of skew-adjoint operators, the case with spectral gap has been treated first. A major motivation for these general adiabatic theorems 
comes from applications to open quantum systems which, unlike closed quantum systems, cannot be described by skew-adjoint operators anymore. 
See, for instance, \cite{NenciuRasche92}, \cite{Joye07}, \cite{AbouSalem07}, \cite{HansonJoyePautratRaquepas17} for the case with spectral gap and \cite{AvronGraf12}, \cite{dipl}, \cite{JaksicPillet14}, for instance, for the case without spectral gap.
A detailed historical overview can be found in~\cite{diss}.
%
So far, however, almost all adiabatic theorems with spectral gap condition, except those from~\cite{NenciuRasche92} and \cite{Joye07}, and all adiabatic theorems without spectral gap condition require the considered spectral values $\lambda(t)$ to be semisimple (case with spectral gap) or weakly semisimple (case without spectral gap), that is, the decomposition~\eqref{eq: zerl, sl} or~\eqref{eq: zerl, ohne sl} holds with $m(t) = 1$. 
It is clear that the spectral values of a general linear operator -- as opposed to a skew-adjoint operator -- will generally fail to be (weakly) semisimple. 
\smallskip

In this paper, we therefore extend and develop further the existing adiabatic theory accordingly: 
we establish adiabatic theorems -- with and especially without spectral gap condition -- for general linear operators $A(t): D \subset X \to X$ with time-independent domain $D(A(t)) = D$ and with spectral values $\lambda(t)$ that are no longer required to be (weakly) semisimple. 
Additionally, the required regularity conditions on $A(t)$, $\lambda(t)$, $P(t)$ from our adiabatic theorems are fairly mild. 
We will apply our adiabatic theorems without spectral gap to 
slowly time-varying open quantum systems described by weakly dephasing generators $A(t)$ and to adiabatic switching processes described by skew-adjoint operators $A(t) = A_0 + \kappa(t)V$ with a switching function $\kappa$. In particular, we generalize the classic 
Gell-Mann and Low theorem to not necessarily isolated eigenvalues. 
In more detail, the contents and contributions of this paper can be described as follows.
\smallskip


In Section~\ref{sect: vorber} we provide the most important preliminaries needed for our adiabatic theorems. 
Sections~\ref{sect: sobolev reg} and~\ref{sect: wohlg und zeitentw} provide the preliminaries related to our regularity assumptions and to well-posedness. At first reading one may well confine oneself to Section~\ref{sect: wohlg und zeitentw} where the concept of well-posedness of non-autonomous linear evolution equations is defined by way of evolution systems and where a fundamental criterion for well-posedness due to Kato is recalled. Section~\ref{sect: sobolev reg} can be skipped at first reading because the less common notions of $W^{m,1}_*$-regularity and $(M,0)$-stability of operator-valued functions introduced there can, at any occurrence, be replaced by the simpler notions of $m$ times strong continuous differentiability and contraction semigroup generators, respectively. 
Section~\ref{sect: spectral op} collects some basic facts about spectral operators and their spectral theory for the convenience of the reader.
\smallskip

In Section~\ref{sect: spectral relatedness}, in turn, we introduce suitable spectral projections for general linear operators, namely the associated and the weakly associated projections, and discuss their central properties. In particular, we discuss the decompositions~\eqref{eq: zerl, sl} and~\eqref{eq: zerl, ohne sl} 
as well as existence and uniqueness issues. While in the case with spectral gap existence and uniqueness of associated projections is for granted, existence of weakly associated projections is unfortunately 
not for granted in the case without spectral gap (but, at least, existence of such a projection  already implies uniqueness). We therefore present criteria for the existence of weakly associated projections, particularly in the case of spectral operators.
\smallskip

Section~\ref{sect: def sl} properly defines uniform and non-uniform spectral gaps 
and introduces the closely related intuitive notion of a set-valued map $\sigma(\,.\,)$ falling into $\sigma(A(\,.\,))\setminus\sigma(\,.\,)$. In addition, continuity of set-valued maps is explained.
In Section~\ref{sect: ad zeitentw} we introduce the basic concept of adiabatic evolution systems, that is, evolution systems $V$ that for a given family of projections $P(t)$ exactly follow the subspaces $P(t)X$ and $(1-P(t))X$ in the sense that 
\begin{align} \label{eq: def ad zeitentw, einl}
V(t,t_0)P(t_0) = P(t)V(t,t_0)
\end{align}
for all $t_0 \le t$. We also identify circumstances under which an adiabatic theorem holds true already on trivial grounds. 
And finally, in Section~\ref{sect: vorber q.d.s.} we provide the preliminaries on generators -- especially (weakly) dephasing generators -- of quantum dynamical semigroups needed for our application to open quantum systems.
\smallskip

In Section~\ref{sect: adsatz mit glm sl} and~\ref{sect: adsatz mit nichtglm sl} we prove our adiabatic theorems with uniform and non-uniform spectral gap condition 
which generalize in a quite simple way the adiabatic theorem of Abou Salem from~\cite{AbouSalem07}. 
In simplified form, our theorems (combined) 
can be formulated as follows (with $I := [0,1]$). See~\cite{Schmid15qmath}. 
If $A(t): D \subset X \to X$ for every $t \in I$ generates a contraction semigroup, 
if $\lambda(t)$ for every $t \in I$ is a spectral value of $A(t)$ and $\lambda(\,.\,)$ falls into $\sigma(A(\,.\,)) \setminus \{\lambda(\,.\,)\}$ at only countably many points which, in turn, accumulate at only finitely many points, 
and if $P(t)$ for every $t \in I \setminus N$ is associated with $A(t)$ and $\lambda(t)$, 
where $N$ denotes the set of those points where $\lambda(\,.\,)$ falls into $\sigma(A(\,.\,)) \setminus \{\lambda(\,.\,)\}$, then -- under suitable regularity assumptions -- one has:
\begin{align} \label{eq: adsatz mit sl, einl}
\sup_{t \in I} \norm{U_{\eps}(t,0)-V_{\eps}(t,0)} = O(\eps) \quad \text{or} \quad 
\sup_{t \in I} \norm{U_{\eps}(t,0)-V_{\eps}(t,0)} = o(1) 
\end{align}
as $\eps \searrow 0$, depending on whether $N = \emptyset$ (uniform spectral gap) or $N \ne \emptyset$ (non-uniform spectral gap). In the above relation, $U_{\eps}$ and $V_{\eps}$ denote the evolution system for $\frac 1 \eps A$ and $\frac 1 \eps A + [P',P]$, respectively. Since $V_{\eps}$ is adiabatic w.r.t.~$P$ in the sense of~\eqref{eq: def ad zeitentw, einl} for every $\eps$, one in particular has  the convergence~\eqref{eq: aussage des adsatzes}. 
Actually, we prove a slightly more general version of the above theorems where at any occurrence the singleton $\{\lambda(t)\}$ is replaced by a general compact subset $\sigma(t)$ of $\sigma(A(t))$. 
In Section~\ref{sect: bsp, adsaetze mit sl} we discuss, among other things, the special case of the above theorem 
where the spectral values $\lambda(t)$ are poles of $(\,.\,-A(t))^{-1}$. It turns out that this special case is particularly enlightening with regard to the proof of our adiabatic theorems without spectral gap condition. We also present an example showing that the contraction semigroup generator assumption on $A(t)$ cannot be essentially weakened.  
\smallskip

In Section~\ref{sect: qual adsatz ohne sl} and~\ref{sect: quant adsatz ohne sl} we establish our (qualitative and quantitative) adiabatic theorems without spectral gap condition. With these theorems, 
we generalize the respective adiabatic theorems of Avron, Fraas, Graf, Grech from~\cite{AvronGraf12} and of Schmid from~\cite{dipl}, which cover the case of weakly semisimple eigenvalues. 
Section~\ref{sect: qual adsatz ohne sl} contains a qualitative adiabatic theorem which, in simplified form, can be formulated 
as follows (with $I := [0,1]$). See~\cite{Schmid15qmath}. 
If $A(t): D \subset X \to X$ for every $t \in I$ generates a contraction semigroup, 
if $\lambda(t)$ for every $t \in I$ is an eigenvalue of $A(t)$ such that $\lambda(t) + \delta e^{i \vartheta(t)} \in \rho(A(t))$ for every $\delta \in (0,\delta_0]$, 
and if $P(t)$ 
is weakly associated with $A(t)$ and $\lambda(t)$ for almost every $t \in I$ and of finite rank and the reduced resolvent estimate
\begin{align} \label{eq: resolvabsch, einl}
\norm{ \big( \lambda(t)+\delta e^{i \vartheta(t)} - A(t) \big)^{-1} (1-P(t)) } \le \frac{M_0}{\delta} \qquad (\delta \in (0,\delta_0]),
\end{align}
is satisfied, 
then -- under suitable regularity assumptions -- one has the convergence $\sup_{t \in I} \norm{(1-P(t)) U_{\eps}(t,0) P(0)} \longrightarrow 0$ as $\eps \searrow 0$. If, in addition, $X$ is reflexive, then one even has 
\begin{align} \label{eq: adsatz ohne sl, einl}
\sup_{t \in I} \norm{U_{\eps}(t,0)-V_{\eps}(t,0)} \longrightarrow 0 \qquad (\eps \searrow 0),
\end{align}
where $U_{\eps}$ and $V_{\eps}$ as before denote the evolution system for $\frac 1 \eps A$ and $\frac 1 \eps A + [P',P]$, respectively.
An important step in the proof of this theorem is 
to find bounded operators $B(t)$ that approximately solve the commutator equation 
\begin{align} \label{eq: comm eq, einl}
B(t)A(t)-A(t)B(t) \subset [P'(t),P(t)]
\end{align}
up to a suitable controllable error. In the case with spectral gap, this commutator equation has an exact solution (which is used in Section~\ref{sect: adsaetze mit sl}) and, by recasting this exact solution appropriately, we can guess an at least approximate solution to~\eqref{eq: comm eq, einl} in the case without spectral gap.  
%
%
As has already been pointed out above, the existence of a 
projection $P(t)$ weakly associated with $A(t)$ and $\lambda(t)$ is not for granted in the situation of the above theorem without spectral gap. We therefore identify a relatively large class of spectral operators $A(t)$ and corresponding eigenvalues $\lambda(t)$ for which weakly associated projections do exist and for which, moreover, the reduced resolvent estimate~\eqref{eq: resolvabsch, einl} holds true.
Additionally, we extend the above adiabatic theorem 
to the case of several eigenvalue curves $\lambda_1, \dots, \lambda_r$. It seems that this extension is new even in the special case of skew-adjoint operators $A(t)$. 
Section~\ref{sect: quant adsatz ohne sl} contains some quantitative refinements of the qualitative adiabatic theorem above. In particular, it contains a quantitative adiabatic theorem for scalar-type spectral operators $A(t)$ whose spectral measures $P^{A(t)}$ are H\"older continuous in $t$ around $\lambda(t)$ in some sense, and our bound on the rate of convergence in~\eqref{eq: adsatz ohne sl, einl} improves 
the respective bound from~\cite{AvronElgart99} and~\cite{Teufel01}. 
%
In Section~\ref{sect: bsp, adsaetze ohne sl} we present some examples illustrating the generality of our theorems and the necessity of some of their regularity assumptions. In particular, we show that adiabatic theory is typically 
uninteresting for multiplication operators $A(t) = M_{f_t}$.  
\smallskip

In Section~\ref{sect: anwendung q.d.s.} we apply our qualitative adiabatic theorem without spectral gap condition to slowly time-varying open quantum systems described by weakly dephasing generators $A(t)$ of quantum dynamical semigroups on the Schatten class $X = S^p(\mathfrak{h})$ on a Hilbert space $\mathfrak{h}$ with $p \in (1,\infty)$. Such operators are of the form
\begin{align} \label{eq:Lindblad form}
A(t) \rho := Z_0(t)(\rho) + \sum_{j \in J} B_j(t) \rho B_j(t)^* - 1/2 \{ B_j(t)^*B_j(t), \rho \} 
\qquad (\rho \in D(Z_0(t))) 
\end{align}
with $Z_0(t)$ being the generator of the semigroup on $S^p(\mathfrak{h})$ defined by $e^{Z_0(t)\tau}(\rho) :=  e^{-i H(t) \tau} \rho \, e^{iH(t) \tau}$, where $H(t): D(H(t)) \subset \mathfrak{h} \to \mathfrak{h}$ is a self-adjoint operator and $B_j(t)$ for every $j$ in the arbitrary index set $J$ is a bounded opertor in $\mathfrak{h}$ such that  
\begin{align} \label{eq:weak dephasingness, def}
\sum_{j\in J} B_j(t) B_j(t)^* = \sum_{j\in J} B_j(t)^* B_j(t) < \infty
\end{align}
for every $t \in [0,1]$. We thereby generalize a result from~\cite{AvronGraf12} where the case of 
dephasing -- instead of weakly dephasing -- generators $A(t)$ with bounded operators $H(t)$ is considered. A dephasing generator is an operator of the form~\eqref{eq:Lindblad form} such that the finiteness condition from~\eqref{eq:weak dephasingness, def} is satisfied and such that $B_j(t), B_j(t)^*$ belong to the double commutant of 
\begin{align}
\mathcal{A} := \big\{ f(H(t)): f \text{ bounded measurable function } \sigma(H(t)) \to \C \big\}
\end{align}
for every $t \in [0,1]$ and $j \in J$.
In Section~\ref{sec: ad switching} we apply our qualitative adiabatic theorem without spectral gap condition -- in the version for several eigenvalues -- to adiabatic switching processes described by skew-adjoint operators of the form $A(t) = A_0 + \kappa(t)V$ with a switching function $\kappa$. In doing so, 
we generalize the Gell-Mann and Low theorems from~\cite{Gell-MannLow51}, \cite{NenciuRasche89}, \cite{Panati10} to the case of not necessarily isolated eigenvalues. In particular, we obtain Gell-Mann and Low formulas of the following two types:
\begin{itemize}
\item a formula 
that relates the eigenstates of the perturbed system described by $A_0+V$ to the eigenstates of the unperturbed system described by $A_0$
\item a formula 
that expresses the change of energy during the switching process in terms of the evolution system $U_{\eps}$ for $\frac{1}{\eps}A$.
\end{itemize}
\smallskip


In the entire paper, we will use the following notational conventions. $X$, $Y$, $Z$ will denote Banach spaces over $\C$, $\mathfrak{h}$ will denote 
a Hilbert space over $\C$, 
and $\norm{\,.\,}_{X,Y}$ will stand for the operator norm on $L(X,Y)$, the space of bounded linear operators from $X$ to $Y$. If $X=Y$, we will usually simply write $\norm{\,.\,}$ for $\norm{\,.\,}_{X,X}$.
Also, we abbreviate 
\begin{align*}
I := [0,1] \qquad \text{and} \qquad \Delta := \{(s,t) \in I^2: s \le t\}
\end{align*}
and for evolution systems $U$ defined on $\Delta$ we will write $U(t) := U(t,0)$ for brevity. 
Whenever a family of linear operators $A(t): D \subset X \to X$ with time-independent domain $D$ is given, then $U_{\eps}$ will denote 
the evolution system for $\frac{1}{\eps} A$ on $D$ provided it exists.
And finally, in our examples $I_d := \{1,\dots,d\}$ for $d \in \N$ and $I_{\infty}:= \N$. 

\section{Some preliminaries} \label{sect: vorber}

\subsection{Sobolev-regularity of operator-valued functions and stability}  \label{sect: sobolev reg}

In this section we introduce a notion of Sobolev regularity and a notion of stability for operator valued-functions that will be used in all our adiabatic theorems below.
\smallskip

%
%
%
We begin by defining 
the notion of $W^{m,p}_*$-\emph{regularity} for $m \in \N$ and $p \in [1,\infty) \cup \{\infty\}$ which is inspired by the introduction of Kato's work~\cite{Kato85}. An operator-valued function $J \ni t \mapsto A(t) \in L(X,Y)$ on a 
compact interval $J$ is said to belong to $W^{0,p}_*(J, L(X,Y)) = L^{p}_*(J,L(X,Y))$ if and only if $t \mapsto A(t)$ is strongly measurable and $t \mapsto \norm{A(t)}$ has a $p$-integrable majorant.
And $t \mapsto A(t)$ is said to belong to $W^{1,p}_*(J, L(X,Y))$ if and only if there is a $B \in L^{p}_*(J,L(X,Y))$ (called a $W^{1,p}_*$-\emph{derivative} of $A$) such that for some (and hence every) $t_0 \in J$ 
\begin{align} \label{eq: def W1,p-ableitung}
A(t)x = A(t_0)x + \int_{t_0}^t B(\tau)x \,d\tau \text{\, for all } t \in J \text{ and } x \in X.
\end{align}
$W^{m,p}_*(J,L(X,Y))$ for arbitrary $m \in \N$ is defined recursively, of course.
\smallskip

In contrast to the usual vector-valued Sobolev spaces $W^{m,p}(J,Z)$, the $W^{m,p}_*(J,L(X,Y))$-spaces by definition, consist of functions (of operators) rather than equivalence classes of such functions.  
It is obvious from the characterization of $W^{1,p}(J,Y)$ by way of indefinite integrals 
that, if $t \mapsto A(t)$ is in $W^{1,p}_*(J, L(X,Y))$, then $t \mapsto A(t)x$ is (the continuous representative of an element) in $W^{1,p}(J,Y)$. 
It is also obvious that 
\begin{align}
W^{1,\infty}_*(J,L(X,Y)) \subset W^{1,p}_*(J,L(X,Y)) \subset W^{1,1}_*(J,L(X,Y))
\end{align}
and that $W^{1,1}_*$- and $W^{1,\infty}_*$-regularity imply absolute continuity or Lipschitz continuity w.r.t.~the norm topology, respectively. It should be noticed however that the converse implication is not true: for example, $t \mapsto A(t)$ with
\begin{align*}
A(t)g := f(t) g \quad (g \in C(I,\C)) \qquad (f(t) := (t-\,.\,) \chi_{[0,t]}(\,.\,) \in C(I,\C))
\end{align*} 
is Lipschitz continuous from $I$ to $L(X,Y)$ ($X = Y := C(I,\C)$), but not $W^{1,\infty}_*$-regular because $t \mapsto A(t) g$ is non-differentiable at every $t \in (0,1)$ for $g := 1$ (Example~1.2.8 of~\cite{ArendtBatty}). 
%
%
A simple and important criterion for $W^{1,\infty}_*$-regularity is furnished 
by the following proposition.

\begin{prop} \label{prop: WOT-stet db impl W^{1,infty}-reg}
Suppose $J \ni t \mapsto A(t) \in L(X,Y)$ is continuously differentiable w.r.t.~the strong or weak operator topology, where $J$ is a compact interval. 
Then $t \mapsto A(t)$ is in $W^{1,\infty}_*(J, L(X,Y))$.
\end{prop}

\begin{proof}
 It is well-known that a weakly continuous map $J \to Y$ is almost separably valued, whence $t \mapsto A'(t)x$ is measurable 
for every $x \in X$ by Pettis' characterization of measurability (Theorem~1.1.1 of~\cite{ArendtBatty}). With 
the Hahn--Banach theorem the conclusion readily follows. 
\end{proof}

It follows from Lebesgue's differentiation theorem that $W^{1,p}_*$-derivatives are essentially unique, more precisely: if $t \mapsto A(t)$ is in $W^{1,p}_*(J,L(X,Y))$ for a $p \in [1,\infty) \cup \{\infty\}$ and $B_1$, $B_2$ are two $W^{1,p}_*$-derivatives of $A$, then one has for every $x \in X$ that $B_1(t)x = B_2(t)x$ for almost every $t \in J$. It should be emphasized that this last condition does \emph{not} imply that $B_1(t) = B_2(t)$ for almost every $t \in J$. (Indeed, take $J:= [0,1]$, $X := \ell^2(J)$ and define 
\begin{align*}
A(t):= 0 \text{\,  as well as \,} B_1(t)x := \scprd{ e_t, x} e_t \text{\, and \,} B_2(t)x := 0 
\end{align*}
for $t \in J$ and $x \in X$, where $e_t(s) := \delta_{s \, t}$.
Then, for every $x \in X$, $B_1(t)x$ is different from $0$ for at most countably many $t \in J$, and it follows that $B_1$ and $B_2$ both are $W^{1,\infty}_*$-derivatives of $A$, but $B_1(t) \ne B_2(t)$ for every $t \in J$.) 
\smallskip

A very important property of the $W^{1,p}_*$-spaces is that $W^{1,p}_*$-regularity carries over to products and inverses. 
It is used implicitly in~\cite{Dorroh75} for $p = 1$ and noted explicitly in the introduction of~\cite{Kato85} for $p= \infty$ and for separable spaces. A proof for general exponents $p$ and spaces can be found in~\cite{diss} (Lemma~2.1.2).

\begin{lm} \label{lm: prod- und inversenregel} 
Suppose that $J = [a,b]$ is compact 
and $p \in [1,\infty) \cup \{\infty\}$.
\begin{itemize}
\item[(i)] If $t \mapsto A(t)$ is in $W^{1,p}_*(J,L(X,Y))$ and $t \mapsto B(t)$ is in $W^{1,p}_*(J,L(Y,Z))$, then $t \mapsto B(t)A(t)$ is in $W^{1,p}_*(J,L(X,Z))$ and $t \mapsto B'(t)A(t) + B(t)A'(t)$ is a $W^{1,p}_*$-derivative of $B A$ for every $W^{1,p}_*$-derivative $A'$, $B'$ of $A$ or $B$, respectively. 
\item[(ii)] If $t \mapsto A(t)$ is in $W^{1,p}_*(J,L(X,Y))$ and $A(t)$ is bijective onto $Y$ 
for every $t \in J$, 
then $t \mapsto A(t)^{-1}$ is in $W^{1,p}_*(J,L(Y,X))$ and $t \mapsto - A(t)^{-1} A'(t) A(t)^{-1}$ is a $W^{1,p}_*$-derivative of $A^{-1}$ for every $W^{1,p}_*$-derivative $A'$ of $A$.
\end{itemize}
\end{lm}

We now move on to briefly recall from~\cite{Kato70} or~\cite{Pazy} the concept of stable families of operators. A family $A$ of linear operators $A(t): D(A(t)) \subset X \to X$ (where $t \in J$) is called \emph{$(M,\omega)$-stable} (for some $M \in [1,\infty)$ and $\omega \in \R$) if and only if $A(t)$ generates a strongly continuous semigroup on $X$ for every $t \in J$ and 
\begin{align}
\norm{   e^{A(t_n) s_n}  \, \dotsm \, e^{A(t_1) s_1}   } \le M e^{\omega (s_1 + \, \dotsb \, + s_n) } 
\end{align} 
for all $s_1, \dots, s_n \in [0,\infty)$ and $t_1, \dots, t_n \in J$ satisfying $t_1 \le \dotsb \le t_n$ with arbitrary~$n \in \N$. 
Alternatively, $(M,\omega)$-stability could be defined via 
the resolvents of the $A(t)$ (Proposition~3.3 of~\cite{Kato70}) or certain monotonic families of norms 
(Proposition~1.3 of~\cite{Nickel00}). 
\smallskip

Clearly, a family $A$ of linear operators in $X$ is $(1,0)$-stable if and only if each member $A(t)$ of the family generates a contraction semigroup on $X$. 
In the particular case of operators $A(t)$ having the simple form $\lambda(t) + \alpha(t) N$ in $X = \ell^p(I_d)$, stability of $A$ 
can be easily characterized in terms of the following condition.

\begin{cond} \label{cond: baustein mit nicht-halbeinfachem ew}
$N \ne 0$ is a nilpotent operator in $X := \ell^p(I_d)$ (with $p \in [1,\infty)$ and $d \in \N$), $\lambda(t) \in \C$ and $\alpha(t) \in [0,\infty)$ for all $t \in I$, and there is an $r_0 > 0$ such that 
\begin{align*}
-\Re \lambda(t) = |\Re \lambda(t) | \ge r_0 \alpha(t) \qquad (t \in I).
\end{align*}
\end{cond}

\begin{lm} \label{lm: char (M,0)-stab für einfaches A}
Suppose that $N \ne 0$ is a nilpotent operator in $X := \ell^p(I_d)$ with $p \in [1,\infty)$ and $d \in \N$ and that $A(t) = \lambda(t) + \alpha(t) N$ for every $t \in I$, where $\lambda(t) \in \C$ and $\alpha(t) \in [0,\infty)$. Then $A$ is $(M,0)$-stable for some $M \in [1,\infty)$ if and only if 
Condition~\ref{cond: baustein mit nicht-halbeinfachem ew} is satisfied.
\end{lm}

\begin{proof}
Suppose first that $A$ is $(M, 0)$-stable for some $M \in [1,\infty)$ and assume, without loss of generality, that $N = \operatorname{diag}(J_1, \dots, J_m)$ is in Jordan normal form with decreasingly ordered Jordan block matrices $J_1, \dots, J_m$ 
We then show that $-\Re \lambda(t) = | \Re \lambda(t) | \ge \frac{1}{4 M} \, \alpha(t)$ for every $t \in I$. 
It is clear by the $(M,0)$-stability of $A$ that $\lambda(t) \in \sigma(A(t)) \subset \{ \Re z \le 0 \}$ for every $t \in I$ and that the family $\tilde{A}$ with $\tilde{A}(t) := \Re \lambda(t) + \alpha(t) N$ is $(M, 0)$-stable as well. 
If $\alpha(t) = 0$ for some $t$, then the desired inequality 
is trivial. If $\alpha(t) \ne 0$ for some $t$, then $\Re \lambda(t) < 0$ by the $(M,0)$-stability of $A$ and therefore we get from
\begin{align*}
(\lambda - \tilde{A}(t))^{-1} e_2 = (\frac{\alpha(t)}{(\lambda-\Re \lambda(t))^2}, \frac{1}{\lambda-\Re \lambda(t)}, 0, 0, \dots) \qquad (\lambda \in (0,\infty))
\end{align*}
with the particular choice $\lambda := |\Re \lambda(t)|$ and from the $(M,0)$-stability of $\tilde{A}$ that
\begin{align*}
\frac{\alpha(t)}{4 \, | \Re \lambda(t) |}    \le  \norm{ |\Re \lambda(t)| \, \big( |\Re \lambda(t)| - \tilde{A}(t) \big)^{-1} \, e_2 } \le M,
\end{align*}
as desired.
Suppose conversely that there is an $r_0 > 0$ such that $-\Re \lambda(t) = |\Re \lambda(t)| \ge r_0 \alpha(t)$ for every $t \in I$. Then, for $M = M_{r_0} \in [1,\infty)$ chosen such that $\norm{ e^{N s} } \le M e^{r_0 \, s}$ for all $s \in [0,\infty)$, 
we obtain
\begin{align*}
\norm{ e^{A(t_n)s_n} \dotsb e^{A(t_1)s_1} } 
= e^{\Re \lambda(t_n)s_n} \dotsb e^{\Re \lambda(t_1)s_1} \, \norm{ e^{N ( \alpha(t_n)s_n + \dotsb + \alpha(t_1)s_1 )} }
\le M
\end{align*}
for all $s_1, \dots, s_n \in [0,\infty)$ and all $t_1, \dots, t_n \in I$ satisfying $t_1 \le \dotsb \le t_n$ (with arbitrary~$n \in \N$), as desired.
\end{proof}

With this lemma, it is simple to produce examples of $(M,0)$-stable families that fail to be $(1,0)$-stable. 
Choose, for instance, $A(t) := -t/3 +t^2 N$ in $X := \ell^p(I_d)$ with $p \in [1,\infty)$ and $d \ge 2$ and with $N$ being the standard $d$ by $d$ Jordan block (with ones on the upper diagonal and zeros everywhere else). 
\smallskip

When it comes to estimating perturbed evolution systems in Section~\ref{sect: adsaetze mit sl} and~\ref{sect: adsaetze ohne sl}, 
the following well-known criterion for stability (Proposition~3.5 of~\cite{Kato70}) 
will always be used. 

\begin{lm} \label{lm:stoersatz (M,omega)-stab}
If $A$ is an $(M,\omega)$-stable family of linear operators $A(t): D(A(t)) \subset X \to X$ for $t \in J$, $B(t)$ is a bounded operator in $X$ for $t \in J$ and $b := \sup_{t \in J} \norm{ B(t) }$ is finite, then $A + B$ is $(M, \omega + M b)$-stable
\end{lm}

%
%
In our examples, the following lemma will be important. 

\begin{lm} \label{lm: (M,w)-stabilität und ähnl.trf.}
Suppose $A_0$ is an $(M_0,\omega_0)$-stable family of 
operators $A_0(t): D(A_0(t)) \subset X \to X$ for $t \in J$ and $R(t): X \to X$ for every $t \in J$ is a bijective bounded operator such that $t \mapsto R(t)$ is in $W^{1,\infty}_*(J, L(X))$. Then the family $A$ with $A(t) := R(t)^{-1} A_0(t) R(t)$ is $(M,\omega)$-stable for some $M \in [1,\infty)$ and $\omega = \omega_0$.
\end{lm}

\begin{proof}
We may assume that $\omega_0 = 0$, since $(\tilde{M},\tilde{\omega})$-stability of a family $\tilde{A}$ is equivalent to the $(\tilde{M},0)$-stability of $\tilde{A}-\tilde{\omega}$.
Set $\norm{x}_t := d \, e^{-M_0 c t} \, \norm{ R(t) x }_{0 \, t}$ for $x \in X$ and $t \in J$, where 
\begin{align*}
c := \esssup_{t \in J} \norm{ R'(t) R(t)^{-1} } \quad \text{and} \quad d := \sup_{t \in J}  e^{M_0 c t} \norm{ R(t)^{-1} }
\end{align*}  
and the $\norm{\,.\,}_{0 \, t}$ are norms on $X$ associated with $A_0$ according to Proposition~1.3 of~\cite{Nickel00}. It then easily follows -- in a similar way as in the proof of Theorem~4.2 of~\cite{Kisynski63} -- 
that the norms $\norm{\,.\,}_t$ satisfy the conditions (a), (b), (c) of Proposition~1.3 in~\cite{Nickel00} for the family $A$ with a certain $M \in [1,\infty)$ 
and therefore $A$ is $(M,0)$-stable, as desired.
\end{proof}

\subsection{Well-posedness and evolution systems}  \label{sect: wohlg und zeitentw}

In this section, we recall from~\cite{EngelNagel} the concepts of well-posedness and (solving) evolution systems for non-autonomous linear evolution equations 
\begin{align} \label{eq: ivp, vorber}
x' = A(t)x  \quad (t \in [s,b]) \quad \text{and} \quad  x(s) = y
\end{align}
with densely defined linear operators $A(t): D \subset X \to X$  ($t \in [a,b]$) and initial values $y \in D$ at initial times $s \in [a,b)$. We also recall a fundamental criterion for well-posedness due to Kato which is constantly used in this paper.
\smallskip

Well-posedness of evolution equations~\eqref{eq: ivp, vorber} means, of course, something like unique (classical) solvability with continuous dependence of the initial data. 
In precise terms, the initial value problems~\eqref{eq: ivp, vorber} for $A$ are called \emph{well-posed on (the space) $D$} if and only if there exists a \emph{(solving) evolution system for $A$ on (the space) $D$}. 
Such an evolution system for $A$ on $D$ is, by definition, 
a family $U$ of bounded operators $U(t,s)$ in $X$ for $(s,t) \in \Delta_J := \{ (s,t) \in J^2: s \le t \}$ such that 
\begin{itemize}
\item [(i)] for every $s \in [a,b)$ and $y \in D$, the map $[s,b] \ni t \mapsto U(t,s)y$ is a continuously differentiable solution to the initial value problem~\eqref{eq: ivp, vorber}, that is,  
a continuously differentiable map $x: [s,b] \to X$ such that $x(t) \in D$ and $x'(t) = A(t) x(t)$ for all $t \in [s,b]$ and $x(s) = y$,
\item[(ii)] $U(t,s) U(s,r) = U(t,r)$ for all $(r,s), (s,t) \in \Delta_J$ 
and $\Delta_J \ni (s,t) \mapsto U(t,s)x$ is continuous for all $x \in X$.
\end{itemize}
%
\smallskip

If, for a given family $A$ of densely defined operators $A(t): D \subset X \to X$, there exists any solving evolution system, then it is already unique. In order to see this we need the following simple lemma, which will always be used when the difference of two evolution systems has to be dealt with. 

\begin{lm} \label{lm: zeitentw rechtsseit db}
Suppose $A(t): D \subset X \to X$ for every $t \in J$ is a densely defined linear operator such that $t \mapsto A(t)x$ is continuous for $x \in D$ and suppose further that $U$ is an evolution system for $A$ on $D$. 
Then, for every $x \in D$, the map $[a,t] \ni s \mapsto U(t,s)x$ is continuously differentiable with derivative $s \mapsto -U(t,s) A(s) x$. 
\end{lm}


\begin{proof}
Since $U(t,s)U(s,r) = U(t,r)$ for $(r,s), (s,t) \in \Delta_J$ and since $\Delta_J \ni (s,t) \mapsto U(t,s)$ is strongly continuous, we obtain for every $s_0 \in [a,t)$ and $x \in D$ that
\begin{align*}
\frac{ U(t,s_0+h)x - U(t,s_0)x }{h} &= -U(t,s_0+h) \frac{ U(s_0+h,s_0)x - x }{h} \\
&\longrightarrow -U(t,s_0)A(s_0)x 
\end{align*}
as $h \searrow 0$. So, the map $[a,t] \ni s \mapsto U(t,s)x$ is right differentiable with right derivative $s \mapsto -U(t,s) A(s) x$. Since this right derivative is continuous, the asserted continuous differentiability of $[a,t] \ni s \mapsto U(t,s)x$ for $x \in D$ follows by Corollary~2.1.2 of~\cite{Pazy}.
\end{proof}

\begin{cor} \label{cor: zeitentw eind} 
Suppose $A(t): D \subset X \to X$ for every $t \in J$ is a densely defined linear operator. If $U$ and $V$ are two evolution systems for $A$ on $D$, then $U = V$.
\end{cor}

\begin{proof}
If $U$ and $V$ are two evolution systems for $A$ on the space $D$, then for every $(s,t) \in \Delta_J$ with $s < t$ and $y \in D$ the map $[s,t] \ni \tau \mapsto U(t,\tau)V(\tau,s)y$ is continuous and right differentiable with vanishing right derivative by virtue of Lemma~\ref{lm: zeitentw rechtsseit db}. With the help of Corollary~2.1.2 of~\cite{Pazy} 
it then follows that
\begin{align*}
V(t,s)y - U(t,s)y = U(t,\tau)V(\tau,s)y \big|_{\tau=s}^{\tau=t} = 0,
\end{align*}
which by the density of $D$ in $X$ implies $U(\,.\,,s) = V(\,.\,,s)$. Since $s$ was arbitrary in $[a,b)$ we obtain $U = V$, as desired.
\end{proof}


\begin{cond} \label{cond: reg 1}
$A(t): D \subset X \to X$ for every $t \in I$ is a densely defined closed 
linear operator such that $A$ is $(M,\omega)$-stable for some $M \in [1, \infty)$ and $\omega \in \R$ and such that $t \mapsto A(t)$ is in $W^{1,1}_*(I,L(Y,X))$, where $Y$ is the space $D$ endowed with the graph norm of $A(0)$. 
\end{cond}

%
It follows from a classic theorem of Kato (Theorem~1 of~\cite{Kato73}) that Condition~\ref{cond: reg 1} guarantees well-posedness of~\eqref{eq: ivp, vorber} on $D$ as well as the bound
\begin{align*}
\norm{U(t,s)} \le M e^{\omega (t-s)} \qquad ((s,t) \in \Delta)
\end{align*}
for the evolution system $U$ for $A$ on $D$. 
%
%
%
%
Also, Condition~\ref{cond: reg 1} is essentially everything we have to require of $A$ in our adiabatic theorems: 
indeed, we have only to add the requirement that $\omega = 0$ 
to arrive at the assumptions on $A$ of these theorems. 
In most adiabatic theorems in the literature -- for example those of ~\cite{AvronSeilerYaffe87}, \cite{AvronElgart99}, \cite{Teufel01}, \cite{Teufel03}, \cite{AbouSalemFroehlich05}, 
\cite{AbouSalem07} or~\cite{AvronGraf12} -- by contrast, the assumptions on $A$ rest upon Yosida's theorem (Theorem~XIV.4.1 of~\cite{Yosida}): 
in these 
theorems it is required of $A$ that each $A(t)$ generate a contraction semigroup on $X$ and that an appropriate translate $A-z_0$ of $A$ satisfy the rather 
involved hypotheses of Yosida's theorem (or -- for example in the case of~\cite{AvronSeilerYaffe87} or~\cite{AvronGraf12} -- more convenient strengthenings thereof). 
It is shown in \cite{SchmidGriesemer14} that this is the case 
if and only if $A(t)-z_0$, for every $t \in I$, is a boundedly invertible generator of a contraction semigroup on $X$ and 
\begin{align*}
t \mapsto A(t)x \text{ is continuously differentiable for all } x \in D.
\end{align*}
In particular, it follows (Proposition~\ref{prop: WOT-stet db impl W^{1,infty}-reg}) 
that the regularity conditions on $A$ of the adiabatic theorems presented here are more general than the respective assumptions of the previously known adiabatic theorems. 


\subsection{Spectral operators}  \label{sect: spectral op}


In this section we recall from~\cite{DunfordSchwartz} some basic facts about spectral operators and their spectral theory that will be needed in the sequel. 
\smallskip
 
We begin with the definition of spectral measures. A \emph{spectral measure $P$ on $(\C, \mathcal{B}_{\C},X)$} is a map from $\mathcal{B}_{\C}$ to the set of bounded projections on $X$ such that
\begin{itemize}
\item[(i)] $P(\emptyset) = 0$ and $P(\C) = 1$,
\item[(ii)] $P(E\cap F) = P(E)P(F)$ for all $E,F \in \mathcal{B}_{\C}$,
\item[(iii)] $P(\cup_{n=1}^{\infty}E_n)x = \sum_{n=1}^{\infty} P(E_n)x$ for all $x \in X$ and all pairwise disjoint sets $E_n \in \mathcal{B}_{\C}$.
\end{itemize}
If, in addition, $X = H$ is Hilbert space and 
$P(E)$ is an orthogonal projection for every $E \in \mathcal{B}_{\C}$, then we call $P$ an \emph{orthogonal spectral measure on $(\C, \mathcal{B}_{\C},X)$}. 
Sometimes, we will also use the alternative notation $P_{E} := P(E)$. 
\smallskip

A densely defined closed operator $A: D(A) \subset X \to X$ is called a \emph{spectral operator} if and only if there exists a spectral measure $P$ on $(\C, \mathcal{B}_{\C}, X)$ such that 
\begin{align*}
P(E)A \subset AP(E) \quad \text{and} \quad \sigma(A|_{P(E)D(A)}) \subset \ol{E}
\end{align*}
for every $E \in \mathcal{B}_{\C}$ and such that $P(E)D(A) = P(E)X$ for every bounded $E \in \mathcal{B}_{\C}$. 
Such a spectral measure $P$ is called a \emph{spectral measure for $A$} or 
a \emph{resolution of the identity for $A$}. 
It can be shown (Corollary~XV.3.8 and Theorem~XVIII.1.5 of~\cite{DunfordSchwartz}) that for a given spectral operator $A$ there exists only one spectral measure (called the \emph{spectral measure of $A$} and often denoted by $P^{A}$).
\smallskip


A simple consequence of the definition is that, 
for every $E \in \mathcal{B}_{\C}$, the restriction $A|_{P^{A}(E)D(A)}$ of a spectral opertor $A$ is a spectral operator as well with spectral measure 
given by
\begin{align}
P^{A|_{P^{A}(E)D(A)}}(F) = P^{A}(F)|_{P^{A}(E)X} = P^{A}(F \cap E)|_{P^{A}(E)X} \quad (F \in \mathcal{B}_{\C}).
\end{align}
In particular, if the set $E$ is bounded, then the operator $A|_{P^{A}(E)D(A)} = A|_{P^{A}(E)X}$ is bounded.
It is also easy to see 
that $P^{A}(E) = 0$ for every $E \in \mathcal{B}_{\C}$ with $E \subset \C \setminus \sigma(A)$. In particular, $P^{A}(\sigma(A)) = 1$, and 
if $\sigma(A)$ is bounded, then the operator $A = A P^{A}(\sigma(A))$ is bounded as well.
And finally, if $E \in \mathcal{B}_{\C}$ is an isolated subset of $\sigma(A)$, then 
\begin{align} \label{eq: bem 3 zur def spektralop}
\sigma(A|_{P^{A}(E)D(A)}) = E \quad \text{and} \quad \sigma(A|_{(1-P^{A}(E))D(A)}) = \sigma(A) \setminus E. 
\end{align}

Important special classes of spectral operators are given by the spectral operators of scalar or finite type, 
respectively. An operator $A: D(A) \subset X \to X$ 
is called 
\begin{itemize}
\item[(i)] \emph{spectral operator of scalar type} if and only if $A = \int z \,dP(z)$ for some spectral measure $P$ on $(\C, \mathcal{B}_{\C}, X)$,
\item[(ii)] \emph{spectral operator of finite type} if and only if $A = S + N$ for some bounded spectral operator $S$ of scalar type and some nilpotent operator $N$ with $SN = NS$. 
\end{itemize}

See, for instance, Chapter~XVIII.1 of \cite{DunfordSchwartz} for the definition and central properties of spectral integrals $\int f(z) \,dP(z)$ w.r.t.~a general -- not necessarily orthogonal -- spectral measure $P$. 
Simple examples of spectral operators of scalar type are, of course, the normal operators on a Hilbert space. 
In fact, every spectral operator $A$ of scalar type on a Hilbert space $X = H$ is essentially (up to similarity transformation) a normal operator (by Theorem~1 of~\cite{Wermer54}). 
Simple examples of spectral operators of finite type are the operators on finite-dimensional spaces (Jordan normal form theorem!).
See Chapter~XV.11 and XV.12 and Chapter~XIX and XX of~\cite{DunfordSchwartz} for more interesting -- differential operator -- examples of spectral operators. See also~\cite{GesztesyTkachenko09} where it is shown that the generic one-dimensional periodic Schrödinger operator is spectral of scalar type (Remark~8.7).  
\smallskip

It can be shown 
that spectral operators of scalar or finite type really are spectral operators: 
for every spectral measure $P$ on $(\C, \mathcal{B}_{\C}, X)$, the operator $\int z \,dP(z)$ is spectral with spectral measure $P$ (Lemma~XVIII.2.13 of~\cite{DunfordSchwartz}); 
and for every bounded spectral operator $S$ of scalar type and every nilpotent operator $N$ with $SN = NS$, the operator $S+N$ is bounded spectral with spectral measure $P^{S}$. 
%
In fact, one has the following sufficient condition for an operator to be spectral (Corollary~XVIII.1.4 and Theorem~XVIII.2.28 of~\cite{DunfordSchwartz}), which is also necessary in the case of bounded operators (Theorem~XV.4.5). 

\begin{thm}  \label{thm: char beschr spektralop}
\begin{itemize}
\item[(i)] If $A = S+N$ for a spectral operator $S$ of scalar type and some quasinilpotent operator $N$ with $SN \supset NS$, then $A$ is a spectral operator with spectral measure $P^S$.
\item[(ii)] If $A$ 
is a bounded spectral operator, then $A = S + N$ for some bounded spectral operator $S$ of scalar type and some quasinilpotent operator $N$ with $SN = NS$. 
Additionally, $S$ and $N$ with the above properties are uniquely determined by $A$, namely $S = \int z \,dP^{A}(z)$ and $N = A-S$.
\end{itemize}
\end{thm}




At last, some facts from the spectral theory of bounded spectral operators (Theorem~XV.8.2, Theorem~XV.8.3 and Theorem~XV.8.6 of~\cite{DunfordSchwartz}). See~\cite{diss} (Proposition~3.1.4) for a simple proof.

\begin{prop}  \label{prop: spektrth beschr spektralop}
Suppose $A$ is a bounded spectral operator on $X$ (with spectral measure $P^{A}$) and $\lambda \in \sigma(A)$. 
\begin{itemize}
\item[(i)] If $\lambda \in \sigma_p(A)$, then $P^{A}(\{\lambda\}) \ne 0$.
\item[(ii)] If $P^{A}(\{\lambda\}) = 0$, then $\lambda \in \sigma_c(A)$.
\end{itemize} 
If, in particular, $A$ is of finite type, then $\sigma_r(A) = \emptyset$ and for every $\lambda \in \sigma(A)$ one has:
$\lambda \in \sigma_p(A)$ iff $P^{A}(\{\lambda\}) \ne 0$ and $\lambda \in \sigma_c(A)$ iff $P^{A}(\{\lambda\}) = 0$.
\end{prop}

\subsection{Spectral projections for general linear operators}  \label{sect: spectral relatedness}

In this section we introduce suitable 
notions of spectral projections for general linear operators, namely the notion of associated projections (which is completely canonical) 
and the notion of weakly associated projections (which --~for non-normal, or at least, non-spectral operators -- 
is not 
canonical). 
%
%
%
Suppose $A: D(A) \subset X \to X$ is a densely defined closed linear operator with $\rho(A) \ne \emptyset$, $\sigma \ne \emptyset$ is a compact isolated subset of $\sigma(A)$, $\lambda$ a not necessarily isolated spectral value of $A$, and $P$ a bounded projection in $X$.
We then say, following~\cite{TaylorLay80}, that 
\emph{$P$~is associated with $A$ and $\sigma$} if and only if 
$P$ commutes with $A$, 
$P D(A) = P X$ and
\begin{align*}
\sigma(A|_{PD(A)}) = \sigma  \text{ \, whereas \, }  \sigma(A|_{(1-P)D(A)}) = \sigma(A) \setminus \sigma.
\end{align*}
We say that \emph{$P$~is weakly associated with $A$ and $\lambda$} if and only if 
$P$ commutes with $A$, $P D(A) = P X$ and 
\begin{gather*}
A|_{PD(A)} - \lambda \text{\: is nilpotent whereas \:} A|_{(1-P)D(A)} - \lambda \text{\: is injective and} \\
\text{has dense range in } (1-P)X.
\end{gather*}
If above the order of nilpotence is at most $m$, we will often, 
more precisely, speak of $P$ as being \emph{weakly associated with $A$ and $\lambda$ of order $m$}. (It should be noticed that the above definition allows weakly associated projections to be zero, which however will be not relevant in our adiabatic theorems below.) 
Also, we call $\lambda$ a \emph{weakly semisimple eigenvalue of $A$} if and only if $\lambda$ is an eigenvalue and there is a projection $P$ weakly associated with $A$ and $\lambda$ of order $1$. In this context, 
recall that $\lambda$ is called a \emph{semisimple eigenvalue of $A$} if and only if it is a pole of the resolvent map $(\,.\,-A)^{-1}$ of order $1$ (which is then automatically an eigenvalue by~\eqref{eq: zerl von X, lambda isoliert} below). Also, a semisimple eigenvalue is called \emph{simple} if and only if its geometric 
multiplicity is $1$.

\subsubsection{Central facts about associatedness and weak associatedness}

We now state 
some central facts about 
associatedness and weak associatedness, concerning the question of existence and uniqueness of (weakly) associated projections (for given operators $A$ and spectral values $\lambda$) and the question of describing (in terms of $A$ and $\lambda$) the 
subspaces into which a (weakly) associated projection decomposes the base space $X$. 
We will 
use these facts again and again and they play an important role in our adiabatic theorems. 
It should be pointed out that these facts are completely well-known in the case of associatedness, but seem to be new in the case of weak associatedness. 

\begin{thm} \label{thm: central properties associatedness}
Suppose $A: D(A) \subset X \to X$ is a densely defined closed linear operator with $\rho(A) \ne \emptyset$ 
and $\emptyset \ne \sigma \subset \sigma(A)$ is compact. 
If $\sigma$ is 
isolated in $\sigma(A)$, then there exists a unique projection $P$ associated with $A$ and $\sigma$, 
namely 
\begin{align*}  
P := \frac{1}{2 \pi i} \int_{\gamma} (z-A)^{-1} \, dz,
\end{align*}
where $\gamma$ is a cycle in $\rho(A)$ with indices $\operatorname{n}(\gamma, \sigma) = 1$ and $\operatorname{n}(\gamma, \sigma(A) \setminus \sigma) = 0$.
If $P$ is associated with $A$ and $\sigma = \{ \lambda \}$ and $\lambda$ is a pole of $(\,.\,-A)^{-1}$ 
of order $m$, then
\begin{align}  \label{eq: zerl von X, lambda isoliert}
PX = \ker(A-\lambda)^k \quad \text{and} \quad (1-P)X = \ran(A-\lambda)^k 
\end{align}
for all $k \in \N$ with $k \ge m$.
\end{thm}

\begin{proof}
See, for instance, \cite{dipl} (Theorem~2.14 and Proposition~2.15) or~\cite{GohbergGoldbergKaashoek} for detailed proofs of the existence and uniqueness statement and Theorem~5.8-A of~\cite{Taylor58} for a proof of~\eqref{eq: zerl von X, lambda isoliert}.
\end{proof}

\begin{thm} \label{thm: typ mögl für PX und (1-P)X}
Suppose $A: D(A) \subset X \to X$ is a densely defined closed linear operator with $\rho(A) \ne \emptyset$ 
and $\lambda \in \sigma(A)$. 
If $\lambda$ is non-isolated in $\sigma(A)$, then in general there exists no projection $P$ weakly associated with $A$ and $\lambda$, but if such a projection exists it is already unique. 
If $P$ is weakly associated with $A$ and $\lambda$ of order $m$, then
\begin{align}  \label{eq: zerl von X, lambda nicht isoliert}
PX = \ker(A-\lambda)^k \quad \text{and} \quad (1-P)X = \overline{ \ran}(A-\lambda)^k 
\end{align}
for all $k \in \N$ with $k \ge m$.
\end{thm}


\begin{proof}
We first show that a projection $P$ weakly associated with $A$ and $\lambda$ decomposes the space $X$ according to~\eqref{eq: zerl von X, lambda nicht isoliert}. 
%
So, let $P$ be weakly associated with $A$ and $\lambda$. We may clearly assume that $\lambda = 0$ 
because $P$, being weakly associated with $A$ and $\lambda$, is also weakly associated with $A-\lambda$ and $0$.
Set $M := PX$ and $N := (1-P)X$. 
We first show that $M = \ker A^k$ for all $k \ge m$. Since $A|_{PX} = A|_{PD}$ is nilpotent of order $m$, $A^k|_{PX} = ( A|_{PX} )^k = 0$ and hence $M = PX \subset \ker A^k$ for all $k \ge m$. And since $A|_{(1-P)D(A)}$ is injective, 
\begin{align*}
A^k|_{(1-P)D(A^k)} = ( A|_{(1-P)D(A)} )^k
\end{align*}
is injective as well and hence $\ker A^k \subset PX = M$ for all $k \in \N$.
We now show that $N = \overline{\ran}\, A^k$ for all $k \ge m$. As $PX = \ker A^k$ for $k \ge m$, we have
\begin{align*}
\ran A^k = A^k PD(A^k) + A^k (1-P)D(A^k) = (1-P) A^k D(A^k) \subset (1-P)X = N
\end{align*} 
and therefore $\overline{\ran}\, A^k \subset N$ for all $k \ge m$. It remains to show that the reverse inclusion $N \subset \overline{\ran}\, A^k$ holds true for all $k \in \N$ and this will be done by induction over $k$. Since $A|_{(1-P) D(A)}$ has dense range in $(1-P)X = N$, the desired inclusion is clearly satisfied for $k = 1$. Suppose now that $N \subset \overline{\ran}\, A^k$ is satisfied for some arbitrary $k \in \N$. Since 
\begin{align*}
\ran A|_{(1-P)D(A)} = A (1-P)D(A) = A (z_0 - A)^{-1} N
\end{align*}
and since $A (z_0-A)^{-1}$ is a bounded operator for every $z_0 \in \rho(A)$, it then follows by the induction hypothesis that $A (z_0-A)^{-1} N \subset \overline{\ran}\, A^{k+1}$ and hence 
\begin{align*}
N = \overline{\ran}\, A|_{(1-P)D(A)}  \subset \overline{\ran}\, A^{k+1},
\end{align*}
which concludes the induction and hence the proof of~\eqref{eq: zerl von X, lambda nicht isoliert}.
\smallskip

With~\eqref{eq: zerl von X, lambda nicht isoliert} at hand, we can now easily show the uniqueness and existence statements. 
%
Indeed, if $P$ and $Q$ are two projections weakly associated with $A$ and $\lambda$ of order $m$ and~$n$ respectively, then 
\begin{gather*}
PX = \ker(A-\lambda)^m = \ker(A-\lambda)^n = QX, \\ (1-P)X = \overline{\ran}(A-\lambda)^m = \overline{\ran}(A-\lambda)^n = (1-Q)X
\end{gather*}
by virtue of~\eqref{eq: zerl von X, lambda nicht isoliert} and therefore $P = Q$.
%
In order to see the existence statement, choose $A := S_-$ on $X:=\ell^2(\N)$ and $\lambda := 0$ ($S_-$ the left shift operator on $\ell^2(\N)$) or alternatively $A := \operatorname{diag}(0,S_+)$ on $X := \ell^2(\N) \times \ell^2(\N)$ and $\lambda := 0$ ($S_+$ the right shift operator on $\ell^2(\N)$). It is then elementary to check that 
\begin{align*}
\ker(A-\lambda)^k \subsetneq \ker(A-\lambda)^{k+1}  \quad \text{or} \quad \ol{\ran}(A-\lambda)^k \supsetneq \ol{\ran}(A-\lambda)^{k+1}
\end{align*}
for all $k \in \N$, respectively. In other words: the subspaces $\ker(A-\lambda)^k$ and $\ol{\ran}(A-\lambda)^k$ do not stop growing or shrinking, respectively. 
So, by virtue of~\eqref{eq: zerl von X, lambda nicht isoliert}, there cannot exist a projection weakly associated with $A$ and $\lambda$. 
(See also~\eqref{eq: finite type vor wesentl im krit für schw assoz, A spektral} for an example where $A$ is a spectral operator. 
Another 
class of examples for the possible non-existence of weakly associated projections can be found at the beginning of Section~\ref{sect: anwendung q.d.s.}).
\end{proof}

We make some remarks which discuss certain converses of the above two theorems as well as the relation 
of associatedness and weak associatedness (and of semisimplicity and weak semisimplicity) in the case of an isolated spectral value. 
\smallskip

1. It has been shown in the theorems above that associated and weakly associated projections $P$ of a densely defined operator $A: D(A)\subset X \to X$ and certain 
spectral values $\lambda \in \sigma(A)$ yield decompositions of the space $X$ into the closed subspaces given in~\eqref{eq: zerl von X, lambda isoliert} and~\eqref{eq: zerl von X, lambda nicht isoliert}. 
Conversely, such decompositions of $X$ also yield associated and weakly associated projections: 
let $A: D(A)\subset X \to X$ be a densely defined operator with $\rho(A) \ne \emptyset$ and $\lambda \in \sigma(A)$. 
\begin{itemize}
\item[(i)] If $P$ is a bounded projection such that
\begin{align}
PX = \ker (A-\lambda)^m \quad \text{and} \quad (1-P)X = \ran(A-\lambda)^m
\end{align} 
for some $m \in \N$, then $\lambda$ is isolated in $\sigma(A)$ 
and $P$ is 
associated with $A$ and $\lambda$, and furthermore, $\lambda$ is a pole of $(\,.\,-A)^{-1}$ of order less than or equal to $m$.
\item[(ii)] If $P$ is a bounded projection such that $PA \subset AP$ and
\begin{align} \label{eq: zerl impliziert schwach assoz}
PX = \ker (A-\lambda)^m \quad \text{and} \quad (1-P)X = \ol{\ran}(A-\lambda)^m
\end{align} 
for some $m \in \N$, 
then $P$ is weakly associated with $A$ and $\lambda$ of order less than or equal to $m$. 
\end{itemize}
(See, for instance, Theorem~5.8-D of~\cite{Taylor58} for the proof of~(i) -- the proof of~(ii) is not difficult. In case $m=1$ in~\eqref{eq: zerl impliziert schwach assoz}, the assumption $PA \subset AP$ is automatically satisfied.)
\smallskip

%
2. In the case of isolated spectral values $\lambda$ of operators $A$ as above, we have two notions of generalized spectral projections (associated and weakly associated projections) 
and so the question arises how these two notions are related. 
If $\lambda$ is a pole of $(\,.\,-A)^{-1}$, then associatedness and weak associatedness -- as well as semisimplicity and weak semisimplicity -- coincide: a projection $P$ is then associated with $A$ and $\lambda$ if and only if it is weakly associated with $A$ and $\lambda$. (Combine the preceding remark with the above theorems to see this equivalence.) 
If, however, $\lambda$ is 
an essential singularity of $(\,.\,-A)^{-1}$, then associatedness and weak associatedness have nothing to do with each other: 
a projection $P$ associated with $A$ and $\lambda$ can then not possibly be weakly associated with $A$ and $\lambda$, and vice versa. (Indeed, if a projection $P$ is both associated and weakly associated with $A$ and $\lambda$ of order $m$, then 
\begin{align*}
z \mapsto (z-A)^{-1} &= (z-A)^{-1}P + (z-A)^{-1}(1-P) \\
&= \sum_{k=0}^{m-1} \frac{ (A|_{PD(A)}-\lambda)^k }{ (z-\lambda)^{k+1} } \, P 
+ \big( z-A|_{(1-P)D(A)} \big)^{-1} (1-P) 
\end{align*}   
has a pole of order $m$ at $\lambda$.) 
A specific example of an operator $A$ (on $X = L^2(I) \times L^2(I)$), where $\lambda=0$ is an essential singularity of the resolvent and not only an associated projection $P_1$ but also a weakly associated projection $P_2$ exists, is given by 
\begin{align} \label{eq: bsp assoz und schw assoz proj ex, obwohl A nicht vom endl typ}
A := \operatorname{diag}(0,V) \quad \text{with} \quad (V f)(t) := \int_0^t f(s)\,ds \quad (f \in L^2(I)).
\end{align} 
%
%

3. If $A$ is an operator as above with distinct spectral values $\lambda \ne \mu$ and if $P$ is weakly associated with $A$ and $\lambda$ and $Q$ is weakly associated with $A$ and $\mu$, then 
\begin{align} \label{eq: PQ=0=QP für schw ass proj}
PQ = 0 = QP.
\end{align} 
An analogous statement for associated projections is well-known and easy to see, but we will not need that in the sequel. 
%
(In order to see~\eqref{eq: PQ=0=QP für schw ass proj}, 
notice that 
\begin{align} \label{eq: PQ=0=QP für schw ass proj, bew}
\sigma(A|_{PD(A)}) \subset \{\lambda\} \quad \text{and} \quad QX = \ker(A-\mu)^m
\end{align} 
by the definition of weak associatedness and the above theorem. 
If now $x \in QX$, then 
\begin{align*}
(A|_{PD(A)} - \mu)^m Px = P (A-\mu)^m x = P (A-\mu)^m Qx = 0
\end{align*}
by virtue of~(\ref{eq: PQ=0=QP für schw ass proj, bew}.b) and therefore $Px = 0$ by virtue of~(\ref{eq: PQ=0=QP für schw ass proj, bew}.a) and $\mu \ne \lambda$. We have thus shown $PQ = 0$ and the other equality follows by symmetry.) 
\smallskip

\subsubsection{Criteria for the existence of weakly associated projections} \label{sect: criteria for the existence of weakly associated projections}

We have seen in the theorem above that for given operators $A$ and spectral values $\lambda$, there will in general exist no projection weakly associated with $A$ and $\lambda$. It is therefore important 
to have criteria for the existence of weakly associated projections. 
\smallskip

%
%
%
In the case of spectral operators $A$ one has the following convenient criterion for the existence of weakly associated projections. In particular, this criterion applies 
if $A$ is a bounded spectral operator of finite type or if $A$ is an unbounded spectral operator of scalar type (for example, skew-adjoint or, more generally, normal).


\begin{prop} \label{prop: krit ex schw assoz proj, A spektral}
Suppose that $A: D(A) \subset X \to X$ is a spectral operator with spectral measure $P^{A}$ and $\lambda \in \sigma(A)$ such that for some bounded neighborhood $\sigma$ of $\lambda$ the bounded spectral operator $A|_{P^{A}(\sigma)X}$ is of finite type. Then there exists a (unique) projection $P$ weakly associated with $A$ and $\lambda$ and it is given by $P = P^{A}(\{\lambda\})$. 
%
\end{prop}

\begin{proof}
We 
often abbreviate $A_E := A|_{P^{A}(E)D(A)}$ 
for $E \in \mathcal{B}_{\C}$. 
It is clear from the definition of spectral operators that $P^{A}(\{\lambda\})$ commutes with $A$ and that $P^{A}(\{\lambda\})D(A) = P^{A}(\{\lambda\})X$, so that we have only to establish the nilpotence, injectivity, and dense range condition from the definition of weak associatedness.
\smallskip
 
As a first step we show that $A|_{P^{A}(\{\lambda\})X} - \lambda = A_{\{\lambda\}}-\lambda$ is nilpotent. 
Since $A_{\sigma}$ is a bounded spectral operator of finite type, we have
$A_{\sigma} = S + N$
with $S = \int z \,dP^{A_{\sigma}}(z)$ and a nilpotent operator $N$ (Theorem~\ref{thm: char beschr spektralop}). So, 
\begin{align*}
A_{\{\lambda\}} = S|_{P^{A}_{\{\lambda\}}X} + N|_{P^{A}_{\{\lambda\}}X} = \lambda + N|_{P^{A}_{\{\lambda\}}X}
\end{align*}
and therefore $A_{\{\lambda\}}-\lambda$ 
is nilpotent, as desired.
\smallskip

As a second step we show that $A|_{(1-P^{A}(\{\lambda\}))D(A)} - \lambda = A_{\sigma(A) \setminus \{\lambda\}}-\lambda$ is injective with dense range in $(1-P^{A}(\{\lambda\}))X = P^{A}(\sigma(A)\setminus \{\lambda\})X$.
In order to do so, we have to treat the case where $\lambda$ is isolated in $\sigma(A)$ and the case where $\lambda$ is non-isolated in $\sigma(A)$ separately. 
Suppose first that $\lambda$ is isolated in $\sigma(A)$. Then 
\begin{align*}
\sigma(A_{\sigma(A)\setminus \{\lambda\}}) \subset \ol{ \sigma(A)\setminus \{\lambda\} } = \sigma(A)\setminus \{\lambda\}
\end{align*} 
(because $\lambda$ is isolated in $\sigma(A)$)
and therefore $A_{\sigma(A) \setminus \{\lambda\}}-\lambda: D(A_{\sigma(A) \setminus \{\lambda\}}) \subset P^{A}(\sigma(A)\setminus \{\lambda\})X \to P^{A}(\sigma(A)\setminus \{\lambda\})X$ is bijective. 
In particular, it is injective with dense range in $P^{A}(\sigma(A)\setminus \{\lambda\})X$, as desired.
Suppose now that $\lambda$ is non-isolated in $\sigma(A)$. Then $A_{\sigma \setminus \{\lambda\}}$ is a bounded spectral operator with $\lambda \in \sigma(A_{\sigma \setminus \{\lambda\}})$ (because $\lambda$ is non-isolated in $\sigma(A_{\sigma})$)
and with $P^{A_{\sigma \setminus \{\lambda\}}}(\{\lambda\}) = P^{A}(\{\lambda\})|_{P^{A}( \sigma \setminus \{\lambda\} )X} = 0$. 
So, we have $\lambda \in \sigma_c(A_{\sigma \setminus \{\lambda\}})$ (Proposition~\ref{prop: spektrth beschr spektralop}) or, in other words, 
\begin{align} \label{eq: 1, krit ex schw assoz proj, A spektral}
A_{\sigma \setminus \{\lambda\}}-\lambda \text{ is injective with dense range in } P^{A}(\sigma \setminus \{\lambda\})X.
\end{align}
We also have $\sigma(A_{\sigma(A)\setminus \sigma}) \subset \ol{ \sigma(A)\setminus \sigma } \subset \C \setminus \{\lambda\}$ (because $\sigma$ is a neighborhood of $\lambda$) and therefore $\lambda \in \rho(A_{\sigma(A)\setminus \sigma})$ or, in other words,
\begin{align} \label{eq: 2, krit ex schw assoz proj, A spektral}
A_{\sigma(A) \setminus \sigma}-\lambda: D(A_{\sigma(A) \setminus \sigma}) \subset P^{A}(\sigma(A)\setminus \sigma)X \to P^{A}(\sigma(A)\setminus \sigma)X \text{ is bijective.}
\end{align}
Combining now~\eqref{eq: 1, krit ex schw assoz proj, A spektral} and~\eqref{eq: 2, krit ex schw assoz proj, A spektral} and 
using 
that the direct sum decomposition $P^{A}(\sigma(A)\setminus \{\lambda\})X = P^{A}(\sigma(A)\setminus \sigma)X \oplus P^{A}(\sigma \setminus \{\lambda\})X$
yields 
a corresponding decomposition of the operator $A_{\sigma(A) \setminus \{\lambda\}}$, we easily conclude that $A_{\sigma(A) \setminus \{\lambda\}}-\lambda$ is injective with dense range in $P^{A}(\sigma(A)\setminus \{\lambda\})X$, as desired.
\end{proof}

We point out that the finite-type assumption of the above proposition is essential. Indeed, the operator $A$ on $X := C(I) \times C(I)$ defined by
\begin{align}  \label{eq: finite type vor wesentl im krit für schw assoz, A spektral}
A := \operatorname{diag}(0,V) \quad \text{with} \quad (V f)(t) := \int_0^t f(s)\,ds \quad (f \in C(I))
\end{align}
is quasinilpotent and hence bounded spectral (Theorem~\ref{thm: char beschr spektralop}), but there exists no projection weakly associated with $A$ and $\lambda = 0$. 
(In order to see this, notice that $0 \in \sigma_r(V)$. So, if a weakly associated projection $P$ existed, we would have 
\begin{align*}
PX 
= \ker \operatorname{diag}(0,V^m) = C(I) \times 0 
\quad \text{and} \quad 
(1-P)X 
\subset \ol{\ran}\, \operatorname{diag}(0,V) \subsetneq 0 \times C(I)
\end{align*} 
for some $m \in \N$ by virtue of Theorem~\ref{thm: typ mögl für PX und (1-P)X} and so $PX + (1-P)X \subsetneq C(I) \times C(I) = X$. Contradiction!)
Compare with the operator $A$ from~\eqref{eq: bsp assoz und schw assoz proj ex, obwohl A nicht vom endl typ}, which violates the finite type assumption as well, but nonetheless does have 
a weakly associated projection.
\smallskip


In the case of generators $A$ of bounded semigroups and spectral value $\lambda \in i \R$, one has another criterion for the existence of weakly associated projections, which is due to Avron, Fraas, Graf, Grech (Lemma~14 of~\cite{AvronGraf12}) and will be used in 
Section~\ref{sect: anwendung q.d.s.}. 

\begin{prop} \label{prop: 2nd criterion ex of w ass proj}
Suppose $A: D(A) \subset X \to X$ is the generator of a bounded semigroup on a reflexive space $X$ and $\lambda \in \sigma(A) \cap i \R$ such that the subspace
\begin{align*}
\ker (A-\lambda) + \ol{\ran} (A-\lambda) 
\end{align*}
is closed in $X$. Then there exists a (unique) projection weakly associated with $A$ and $\lambda$.
\end{prop}

We point out that the assumption that $X$ be reflexive is essential here. (See Example~5 or~6 of~\cite{AvronGraf12} or the example at the beginning of Section~\ref{sect: anwendung q.d.s.}.)


\subsubsection{Weak associatedness carries over to the dual operators}

We close this section on spectral projections by noting that in reflexive spaces weak associatedness carries over to the dual operators -- provided that some core condition is satisfied, which is the case for semigroup generators, for instance (Proposition~II.1.8 of~\cite{EngelNagel}).
Associatedness carries over to dual operators as well 
(Section~III.6.6 of~\cite{KatoPerturbation80}) 
-- but this will not be needed in the sequel.

\begin{prop} \label{prop: schwache assoziiertheit, dual}
Suppose $A: D(A) \subset X \to X$ is a densely defined closed linear operator in the reflexive space $X$ such that $\rho(A) \ne \emptyset$ and $D(A^k)$ is a core for $A$ for all $k \in \N$. If $P$ is weakly associated with $A$ and $\lambda \in \sigma(A)$ of order $m$, then $P^*$ is weakly associated with $A^*$ and $\lambda$ of order $m$. 
\end{prop}

\begin{proof}
We begin by showing -- by induction over $k \in \N$ -- the preparatory statement that
\begin{align} \label{eq: indbeh 1}
(A^k)^* = (A^*)^k 
\end{align}
for all $k \in \N$, which might also be of independent interest (notice that $D(A^k)$ 
being a core for $A$ is dense in $X$, so that $(A^k)^*$ is really well-defined).
Clearly, \eqref{eq: indbeh 1} is true for $k = 1$ and, assuming that it is true for some arbitrary $k \in \N$, we now show that $(A^{k+1})^* = (A^*)^{k+1}$ holds true as well.
It is easy to see that $(A^*)^{k+1} \subset (A^{k+1})^*$ and it remains to see that $D((A^{k+1})^*) \subset D((A^*)^{k+1})$. So let $x^* \in D((A^{k+1})^*)$. We show that
\begin{align} \label{eq: zwbeh für k-indschritt}
x^* \in D((A^k)^*) \quad \text{and} \quad (A^k)^* x^* \in D(A^*),
\end{align}
from which it then follows -- by the induction hypothesis -- that $x^* \in D((A^*)^{k+1})$ as desired.
In order to prove that $x^* \in D((A^k)^*)$ we show that 
\begin{align*} 
x^* \in D((A^l)^*)
\end{align*}
for all $l \in \{ 1, \dots, k \}$ -- by induction over $l \in \{ 1, \dots, k \}$ and by working with suitable powers of $(A^*-z_0)^{-1} = ((A-z_0)^{-1})^*$, 
where $z_0$ is an arbitrary point of $\rho(A^*) = \rho(A) \ne \emptyset$ (Theorem~III.5.30 of~\cite{KatoPerturbation80}). 
In the base step of the 
induction, notice that for all $y \in D(A)$
\begin{align*}
&\big\langle (A^*-z_0)^{-k} (A^{k+1})^* x^*, y \big\rangle = \big \langle x^*, A^{k+1} (A-z_0)^{-k} y \big \rangle \\
&\qquad \qquad \qquad = \big \langle x^*, (A-z_0) y \big \rangle + \sum_{i = 0}^k \binom{k+1}{i} z_0^{k+1-i} \big \langle (A^*-z_0)^{-k+i} x^*, y \big \rangle, 
\end{align*}
from which it follows that $x^* \in D((A-z_0)^*) = D(A^*)$. 
In the inductive step, 
assume that $x^* \in D(A^*), \dots, D((A^l)^*)$ for some arbitrary $l \in \{1, \dots, k-1 \}$. Since for all $y \in D(A^{l+1})$
\begin{align*}
&\big \langle (A^*-z_0)^{-(k-l)} (A^{k+1})^* x^*, y \big \rangle = \big \langle x^*, A^{k+1} (A-z_0)^{-(k-l)} y \big \rangle \\
&\qquad \qquad = \big \langle x^*, (A-z_0)^{l+1} y \big \rangle + \sum_{i=k-l+1}^{k} \binom{k+1}{i} z_0^{k+1-i} \big \langle x^*, (A-z_0)^{-(k-l)+i} y \big
\rangle \\ 
&\qquad \qquad \quad + \sum_{i=0}^{k-l} \binom{k+1}{i} z_0^{k+1-i} \big \langle (A^*-z_0)^{-(k-l)+i} x^*,  y \big \rangle,
\end{align*}
it follows by the induction hypothesis of the $l$-induction and by applying the binomial formula to $(A-z_0)^{-(k-l)+i} y$ for $i \in \{ k-l+1, \dots, k+1 \}$ that $x^* \in D((A^{l+1})^*)$.
So the $l$-induction is finished and it remains to show that $(A^k)^* x^* \in D(A^*)$. Since $D(A^{k+1})$ by assumption is a core for $A$, there is for every $y \in D(A)$ a sequence $(y_n)$ in $D(A^{k+1})$ such that 
\begin{align*}
\big \langle (A^k)^* x^*, A y \big \rangle = \lim_{n \to \infty} \big \langle (A^k)^* x^*, A y_n \big \rangle = \lim_{n \to \infty} \big \langle x^*, A^{k+1} y_n \big \rangle 
= \big \langle (A^{k+1})^* x^*, y \big \rangle.
\end{align*}
It follows that $(A^k)^* x^* \in D(A^*)$ and this yields -- together with the induction hypothesis of the $k$-induction -- that $x^* \in D((A^*)^{k+1})$, which finally ends the proof the preparatory statement~\eqref{eq: indbeh 1}.
\smallskip

After this preparation we can now move on to the main part of the proof where we assume, without loss of generality, that $\lambda = 0$ and exploit 
the first remark after Theorem~\ref{thm: typ mögl für PX und (1-P)X} to show that $P^*$ is weakly $m$-associated with $A^*$ and $\lambda = 0$.
$A^*$ is densely defined (due to the reflexivity of $X$ (Theorem~III.5.29 of~\cite{KatoPerturbation80})) 
with $\rho(A^*) = \rho(A) \ne \emptyset$ (Theorem~III.5.30 of~\cite{KatoPerturbation80}) and 
\begin{align*}
P^* A^* \subset (AP)^* \subset (PA)^* = A^* P^*
\end{align*}
because $AP \supset PA$.
Since $(A^m)^* = (A^*)^m$ by~\eqref{eq: indbeh 1} and since $PX = \ker A^m$ and $(1-P)X = \overline{ \ran}\, A^m $ (by Theorem~\ref{thm: typ mögl für PX und (1-P)X}), we further have 
\begin{gather*}
P^*X^* = \ker(1-P)^* = ((1-P)X)^{\perp} = (\overline{\ran}\, A^m)^{\perp} = \ker(A^m)^* = \ker (A^*)^m \\
\text{and} \\
(1-P^*)X^* = \ker P^* = (P X)^{\perp} = (\ker A^m)^{\perp} = ( \ker (A^m)^{**} )_{\perp} = \overline{\ran} (A^m)^*  = \overline{\ran} (A^*)^m ,
\end{gather*}
where in the fourth equality of the second line the closedness of $A^m$ (following from $\rho(A) \ne \emptyset$) and the 
reflexivity of $X$ have been used. (In the above relations, we denote by $U^{\perp} := \{ x^* \in Z^*: \scprd{x^*, U^*} = 0 \}$ and $V_{\perp} := \{ x \in Z: \scprd{V, x} = 0 \}$ the annihilators of subsets $U$ and $V$ of a normed space $Z$ and its dual $Z^*$,  respectively.)
It is now clear from the first remark after Theorem~\ref{thm: typ mögl für PX und (1-P)X} that $P^*$ is weakly $m$-associated with $A^*$ and $\lambda = 0$ 
and we are done.
\end{proof}

\subsection{Spectral gaps and continuity of set-valued maps}  \label{sect: def sl}

We continue by properly defining what exactly we mean by uniform and non-uniform spectral gaps. 
Suppose that $A(t): D(A(t)) \subset X \to X$, for every $t$ in some compact interval $J$, is a densely defined closed linear operator and that $\sigma(t)$ is a compact subset of $\sigma(A(t))$ for every $t \in J$. We then speak of a \emph{spectral gap for $A$ and $\sigma$} if and only if $\sigma(t)$ is isolated in $\sigma(A(t))$ for every $t \in J$. Such a spectral gap for $A$ and $\sigma$ is called \emph{uniform} if and only if $\sigma(\,.\,)$ is even uniformly isolated in $\sigma(A(\,.\,))$ in the sense that 
\begin{align*}
\inf_{t\in J} \dist(\sigma(t), \sigma(A(t)) \setminus \sigma(t)) > 0.
\end{align*}
Also, we say that \emph{$\sigma(\,.\,)$ falls into $\sigma(A(\,.\,)) \setminus \sigma(\,.\,)$ at the point $t_0 \in J$} if and only if there is a sequence $(t_n)$ in $J$ 
converging to $t_0$ such that
\begin{align*}
\dist(\sigma(t_n), \sigma(A(t_n)) \setminus \sigma(t_n)) \longrightarrow 0 \quad (n \to \infty).
\end{align*} 
It is clear that the set of points at which $\sigma(\,.\,)$ falls into $\sigma(A(\,.\,)) \setminus \sigma(\,.\,)$ is closed. 
And by the compactness of $J$ it follows that 
a spectral gap for $A$ and $\sigma$ is uniform if and only if $\sigma(\,.\,)$ at no point falls into $\sigma(A(\,.\,)) \setminus \sigma(\,.\,)$. 
%
And finally, the set-valued map $t \mapsto \sigma(t)$ is called \emph{continuous} if and only if it is \emph{upper and lower semicontinuous} in the sense of Section~IV.3 of~\cite{KatoPerturbation80}, that is, for every $t_0 \in J$ and every $\eps >0$ there is neighborhood $J_{t_0}$ of $t_0$ in $J$ 
such that for every $t \in J_{t_0}$ 
\begin{align*}
\sigma(t) \subset B_{\eps}(\sigma(t_0)) \qquad \text{and} \qquad \sigma(t_0) \subset B_{\eps}(\sigma(t)).
\end{align*}

\subsection{Adiabatic evolutions and a trivial adiabatic theorem}  \label{sect: ad zeitentw} 

As has been explained in Section~\ref{sect: intro}, the principal 
goal of adiabatic theory is to establish the convergence~\eqref{eq: aussage des adsatzes} or, in other words, to show that the evolution systems $U_{\eps}$ for $\frac 1 \eps A$ are, in some sense, approximately adiabatic w.r.t.~$P$ as $\eps \searrow 0$. We say that an evolution system for a family $A$ of linear operators $A(t): D \subset X \to X$ is \emph{adiabatic w.r.t.~a family $P$ of bounded projections $P(t)$ in $X$} if and only if $U(t,s)$ for every $(s,t) \in \Delta$ exactly intertwines $P(s)$ with $P(t)$, that is,  
\begin{align} \label{eq: def adiab zeitentw} 
P(t) U(t,s) = U(t,s) P(s) 
\end{align}
for every $(s,t) \in \Delta$.
Since the pioneering work~\cite{Kato50} of Kato, the basic strategy in proving the convergence~\eqref{eq: aussage des adsatzes} has been to show that
\begin{align}
U_{\eps}(t)-V_{\eps}(t) \longrightarrow 0 \quad (\eps \searrow 0)
\end{align}
for every $t \in I$, where the $V_{\eps}$ are suitable comparison evolution systems that are adiabatic 
w.r.t.~the family $P$ of spectral projections $P(t)$ 
corresponding to $A(t)$ and $\sigma(t)$. 
A simple 
way of obtaining 
adiabatic evolutions w.r.t.~some given family $P$ 
(independently observed by Kato in~\cite{Kato50} and Daleckii--Krein in~\cite{DaleckiiKrein50}) is described in the following proposition.

\begin{prop}[Kato, Daleckii--Krein] \label{prop: intertwining relation}
Suppose $A(t): D \subset X \to X$ for every $t \in I$ is a densely defined closed linear operator and $P(t)$ a bounded projection in $X$ such that $P(t)A(t) \subset A(t)P(t)$ for every $t \in I$ and $t \mapsto P(t)$ is strongly continuously differentiable. If the evolution system $V_{\eps}$ for $\frac 1 \eps A + [P',P]$ exists on $D$ for every $\eps \in (0,\infty)$, then $V_{\eps}$ is adiabatic w.r.t.~$P$ for every $\eps \in (0,\infty)$.
\end{prop}

\begin{proof}
Choose an arbitrary $(s,t) \in \Delta$ with $s \ne t$. 
It then follows by the proof of Lemma~\ref{lm: zeitentw rechtsseit db} 
that, for every $x \in D$, the map 
\begin{align} \label{eq: prop: intertwining relation, 1}
[s,t] \ni \tau \mapsto V_{\eps}(t,\tau) P(\tau) V_{\eps}(\tau,s)x
\end{align}
is continuous and right differentiable. Since $P(\tau)$ commutes with $A(\tau)$ 
and 
\begin{align} \label{eq: PP'P=0}
P(\tau)P'(\tau)P(\tau) = 0 
\end{align}
for every $\tau \in I$ (apply $P$ from the left and the right to the identity $P' = (PP)' = P'P+PP'$), 
it further follows that the right derivative of~\eqref{eq: prop: intertwining relation, 1} is identically $0$. So, 
\begin{align*}
P(t)V_{\eps}(t,s)x - V_{\eps}(t,s)P(s)x = V_{\eps}(t,\tau) P(\tau) V_{\eps}(\tau,s)x \big|_{\tau=s}^{\tau=t} = 0
\end{align*}
by virtue of Corollary~2.1.2 of~\cite{Pazy}, as desired. 
\end{proof}

We now briefly discuss two situations where the conclusion of the adiabatic theorem is already trivially true.

\begin{prop} \label{prop: triv adsatz}
Suppose $A(t): D \subset X \to X$ for every $t \in I$ is a densely defined closed linear operator and $P(t)$ is a bounded projection in $X$ such that the evolution system $U_{\eps}$ for $\frac{1}{\eps} A$ exists on $D$ for every $\eps \in (0,\infty)$ and such that $P(t)A(t) \subset A(t)P(t)$ 
for every $t \in I$ and $t \mapsto P(t)$ is strongly continuously differentiable.
\begin{itemize}
\item[(i)] If $P' = 0$, then $U_{\eps}$ is adiabatic w.r.t.~$P$ for every $\eps \in (0,\infty)$ (in particular, the convergence~\eqref{eq: aussage des adsatzes} holds trivially), 
and the reverse implication is also true.

\item[(ii)] If there are $\gamma \in (0,\infty)$ and $M \in [1,\infty)$ such that for all $(s,t) \in \Delta$ and $\eps \in (0,\infty)$
\begin{align}  \label{eq: zweiter triv adsatz}
\norm{U_{\eps}(t,s)} \le M e^{-\frac \gamma \eps (t-s)},
\end{align}
then $\sup_{t \in I} \norm{U_{\eps}(t)-V_{\eps}(t)} = O(\eps)$ as $\eps \searrow 0$,
whenever the evolution system $V_{\eps}$ for $\frac 1 \eps A + [P',P]$ 
exists on $D$ for every $\eps \in (0,\infty)$.
\end{itemize}
\end{prop}

\begin{proof}
(i) See, for instance, Section~IV.3.2 of~\cite{Krein} for the reverse implication (differentiate the adiabaticity relation~\eqref{eq: def adiab zeitentw} with respect to the variable $s$) -- the other implication is obvious from Proposition~\ref{prop: intertwining relation}.
\smallskip

(ii) Since for $x \in D$ one has (by Corollary~2.12 of~\cite{Pazy} and by the proof of Lemma~\ref{lm: zeitentw rechtsseit db}) 
\begin{align*}
V_{\eps}(t)x - U_{\eps}(t)x = U_{\eps}(t,s)V_{\eps}(s)x \big|_{s=0}^{s=t} = \int_0^t U_{\eps}(t,s) [P'(s),P(s)] V_{\eps}(s)x \, ds
\end{align*}
for every $t \in I$ and $\eps \in (0,\infty)$, it follows 
by~\eqref{eq: zweiter triv adsatz} and a Gronwall argument that 
\begin{align} \label{eq: zweiter triv adsatz, 2}
\norm{V_{\eps}(s)} \le M e^{(-\gamma /\eps + Mc)s} \qquad (s \in I),
\end{align}
where $c := \sup_{s\in I} \norm{[P'(s),P(s)]}$. 
Combining now~\eqref{eq: zweiter triv adsatz} and~\eqref{eq: zweiter triv adsatz, 2}, we obtain
\begin{align*}
\norm{ U_{\eps}(t)-V_{\eps}(t) } \le M^2 c \, e^{Mc} \,\, t \, e^{-\frac \gamma \eps  t} 
\end{align*}
for all $t \in I$ and $\eps \in (0,\infty)$, and from this the desired conclusion is obvious.
\end{proof}

Combining Proposition~\ref{prop: triv adsatz}~(ii) with Example~\ref{ex: (M,0)-stabilitaet wesentl, mit sl} below, one sees that adiabatic theory is interesting only if the evolution systems for $\frac 1 \eps A$ are \emph{only just} bounded w.r.t.~$\eps \in (0,\infty)$: if even the evolution for $\frac 1 \eps (A+\gamma)$ is bounded in $\eps \in (0,\infty)$ for some $\gamma > 0$, then adiabatic theory is trivial for $A$ (by Proposition~\ref{prop: triv adsatz}~(ii)), and if only the evolution for $\frac 1 \eps (A-\gamma)$ is bounded in $\eps \in (0,\infty)$ for some $\gamma > 0$, then adiabatic theory is generally impossible for $A$ 
(by Example~\ref{ex: (M,0)-stabilitaet wesentl, mit sl}).


\subsection{Some facts about quantum dynamical semigroups}  \label{sect: vorber q.d.s.}

In this section we collect some basic facts about dephasing and weakly dephasing generators of quantum dynamical semigroups. A \emph{quantum dynamical semigroup} (on $S^1(\mathfrak{h})$) is, by definition, 
a strongly continuous semigroup $(\Phi_t)$ of bounded linear operators on $S^1(\frak{h})$ 
such that $\Phi_t$ for every $t \in [0,\infty)$ is trace-preserving and completely positive. 
Such semigroups naturally arise in the description of open quantum systems. See~\cite{Kraus71}, \cite{AttalJoyePillet}, \cite{AlickiFannes01}, \cite{AlickiLendi07}, for instance. 
In our application below, 
we will work, following~\cite{AvronGraf12}, with \emph{quantum dynamical semigroups on $S^p(\mathfrak{h})$} with $p \in (1,\infty)$, that is, strongly continuous semigroups $(\Phi_t)$ on $S^p(\mathfrak{h})$ such that $(\Phi_t|_{S^1(\mathfrak{h})})$ is a quantum dynamical semigroup on $S^1(\mathfrak{h})$.

\subsubsection{Weakly dephasing and dephasing generators of quantum dynamical semigroups}

A relatively large class of generators of quantum dynamical semigroups is provided 
by the following theorem. All generators of quantum dynamical semigroups considered in this paper will belong to that class.

\begin{thm} \label{thm: generation result q.d.s.}
Suppose $H: D(H) \subset \mathfrak{h} \to \mathfrak{h}$ is a self-adjoint operator and $B_j$ for every $j \in J$ ($J$ an arbitrary index set) is  a bounded operator in $\mathfrak{h}$ such that
\begin{align} \label{eq: weak dephasingness cond}
\sum_{j\in J} B_jB_j^* = \sum_{j\in J} B_j^* B_j < \infty.
\end{align}
Suppose further that $p \in [1,\infty)$ and that $Z_0$ is the generator of the (weakly and hence strongly continuous) semigroup on $X = S^p(\mathfrak{h})$ defined by $e^{Z_0 t}(\rho) := e^{-iHt} \rho e^{iHt}$. Then
\begin{itemize}
\item[(i)] $D(Z_0) = \{ \rho \in S^p(\mathfrak{h}): \rho D(H) \subset D(H) \text{ and } H\rho -\rho H \subset \sigma \text{ for a } \sigma \in S^p(\mathfrak{h}) \}$ with $Z_0(\rho)$ being the unique element $\sigma$ of $S^p(\mathfrak{h})$ satisfying   $H\rho -\rho H \subset \sigma$ and, moreover, the series 
\begin{align*}
\sum_{j\in J} B_j^*B_j \rho, \quad  \sum_{j\in J} \rho B_j^*B_j, \quad \sum_{j\in J} B_j  \rho B_j^*
\end{align*}
converge in the norm of $S^p(\mathfrak{h})$ for every $\rho \in S^p(\mathfrak{h})$ and define bounded linear operators from $S^p(\mathfrak{h})$ to $S^p(\mathfrak{h})$
\item[(ii)] the operator $A: D(Z_0) \subset X \to X$ defined by
\begin{align} \label{eq: weakly dephas generator}
A(\rho) := Z_0(\rho) + \sum_{j\in J} B_j  \rho B_j^* - 1/2 \, \{B_j^*B_j,\rho\} \qquad (\rho \in D(Z_0))
\end{align}
is the generator of a quantum dynamical semigroup on $X = S^p(\mathfrak{h})$.
\end{itemize}
\end{thm}

We call an operator $A$ of the form~\eqref{eq: weakly dephas generator} in $X = S^p(\mathfrak{h})$ ($p \in [1,\infty)$) 
with $H$ and $B_j$ as in the theorem above a \emph{weakly dephasing generator of a quantum dynamical semigroup} on $X = S^p(\mathfrak{h})$ and we refer to~\eqref{eq: weak dephasingness cond} as the \emph{weak dephasingness condition}. Its precise meaning is that there is a constant $M \in [0,\infty)$ such that 
\begin{align*}
\sum_{j\in F} B_j B_j^*, \quad \sum_{j\in F} B_j^* B_j \le M
\end{align*}
for every finite subset $F$ of $J$ and that the strong limits  $\sum_{j\in J} B_j B_j^*$ and $\sum_{j\in J} B_j^* B_j$ of the nets $(\sum_{j\in F} B_j B_j^*)$ and $(\sum_{j\in F} B_j^* B_j)$ (which exist by the theorem of Vigier) coincide.
In the case $p=1$, the equality in the weak dephasingness condition~\eqref{eq: weak dephasingness cond} is actually not needed. Indeed, for $p=1$ the conclusion of the above generation result already follows under the much more general condition that $H$ and $B_j$ for $j \in J$ are self-adjoint or bounded operators in $\mathfrak{h}$, respectively, such that 
\begin{align} \label{eq: vor lindblad, p=1}
\sum_{j\in J} B_j^* B_j < \infty.
\end{align} 
See, for instance, Corollary~3.6.2 of~\cite{diss} which easily follows by Lindblad's fundamental characterization~\cite{Lindblad76} of generators of norm-continuous quantum dynamical semigroups on $S^1(\mathfrak{h})$. 
See also~\cite{Davies77} and~\cite{ChebotarevFagnola98} for considerably more general sufficient conditions for an operator to generate a quantum dynamical semigroup on $S^1(\mathfrak{h})$. 
In the case $p \ne 1$, the weak dephasingness condition~\eqref{eq: weak dephasingness cond} is used to reduce 
the conclusion of the above generation result by way of complex interpolation theory (Calder\'{o}n--Lions) to the case $p=1$. See Lemma~4.2.9 of~\cite{diss} for a detailed proof of Theorem~\ref{thm: generation result q.d.s.} and~\cite{AvronGraf12} for a proof in the special case of so-called dephasing generators $A$ with bounded $H$.
A weakly dephasing generator $A$ in $X = S^p(\mathfrak{h})$ ($p \in [1,\infty)$) 
is called \emph{dephasing} if and only if $B_j$ and hence $B_j^*$ belongs to the double commutant $\{H\}''$ of $H$ for every $j \in J$ or, for short, 
\begin{gather}
B_j, B_j^* \in \{H\}'' = \mathcal{A}'' = \ol{ \mathcal{A} } \qquad (\text{closure w.r.t.~the strong operator topology}) \label{eq: dephasingness cond} \\
\mathcal{A} := \{ f(H): f \text{\, bounded measurable function \,} \sigma(H) \to \C \}. \notag
\end{gather}
(In the first equality of~\eqref{eq: dephasingness cond}, $\{H\}' = \mathcal{A}'$ is used and in the second equality, the double commutant theorem of von Neumann is used. In case $\mathfrak{h}$ is separable,  the strong closure in~\eqref{eq: dephasingness cond} is superfluous by the theorem of Riesz--Mimura, but for non-separable $\mathfrak{h}$ it is essential (Section~X.2 of~\cite{Sz-Nagy67})). 
Since the commutativity of the $^*$-algebra $\mathcal{A}$ carries over to its  strong closure $\ol{ \mathcal{A} }$ by the density theorem of Kaplansky (for instance), 
we see that for a dephasing generator $A$ the operators $B_j$ are all normal and so the equality in the weak dephasingness condition~\eqref{eq: weak dephasingness cond} is automatically satisfied.

\subsubsection{Some important properties of dephasing and weakly dephasing generators}

In the following proposition, we collect some important properties of dephasing and weakly dephasing genertors $A$ of quantum dynamical semigroups, especially concerning the relation of $\ker A$ and $\ker Z_0$.

\begin{prop} \label{prop: properties (weakly) dephas generators}
Suppose $A$ is a weakly dephasing generator of a quantum dynamical semigroup in $X = S^p(\mathfrak{h})$ with $p \in [1,2]$.
\begin{itemize}
\item[(i)] If $\rho \in \ker A$, then $\rho$ commutes with $B_j$, $B_j^*$ for all $j \in J$ and with $H$. In particular, $\ker A \subset \{H\}' \cap S^p(\mathfrak{h}) = \ker Z_0$.
\item[(ii)] If $A$ is dephasing, then $\ker A = \ker Z_0$. Conversely, if $\ker A = \ker Z_0$ and the spectrum of $H$ is pure point, then $A$ is dephasing.
\item[(iii)] If $p=1$, then $A$ is dephasing if and only if $\ker Z_0^* \subset \ker A^*$.
\end{itemize}
\end{prop}


\begin{proof}
We prove only part~(i) -- for the other parts we refer to Proposition~3.6.3~(ii) (the proof of which is easily seen to carry over from the case $p=1$ to $p \in [1,2]$) and to Proposition~3.6.3~(i) of~\cite{diss}.
We can write $A$ and the generator $A_{S^1}$ of the restricted semigroup in the form~\eqref{eq: weakly dephas generator} by assumption or, for brevity, in the form $A = Z_0 + W$ and $A_{S^1} = Z_{0 S^1} + W_{S^1}$, respectively. 
Choose now $\rho \in \ker A$ and fix it for the rest of the proof. Then $\rho, \rho^* \in D(Z_0) = D(A)$ with $A(\rho) = 0 = A(\rho)^* = A(\rho^*)$ and $\rho^* \rho \in D(Z_{0 S^1}) = D(A_{S^1})$ with $Z_{0 S^1}(\rho^* \rho) = Z_0(\rho^*) \rho + \rho^* Z_0(\rho)$. Indeed, 
\begin{align*}
t \mapsto e^{Z_{0 S^1}t}(\rho^* \rho) = e^{Z_0 t}(\rho^*) e^{Z_0 t}(\rho)
\end{align*}
is differentiable 
in the norm of $S^1(\mathfrak{h})$ because $S^p(\mathfrak{h})$ is continuously embedded in $S^2(\mathfrak{h})$ by virtue of $p \in [1,2]$. So,
\begin{align} \label{eq: Lambda(a*a)-Lambda(a*)a - a*Lambda(a)}
A_{S^1}(\rho^* \rho) &= A_{S^1}(\rho^* \rho) - A(\rho^*) \rho - \rho^* A(\rho)
= W_{S^1}(\rho^* \rho) - W(\rho^*) \rho - \rho^* W(\rho) \notag \\
&= \sum_{j\in J} [\rho,B_j^*]^* [\rho,B_j^*],
\end{align} 
where the last equality follows by straightforward computation using the weak dephasingness condition~\eqref{eq: weak dephasingness cond}. Since $(e^{A_{S^1}t})$ is trace-preserving and since the series in~\eqref{eq: Lambda(a*a)-Lambda(a*)a - a*Lambda(a)} converges in the norm of $S^1(\mathfrak{h})$ (Theorem~\ref{thm: generation result q.d.s.}~(i)), it follows from~\eqref{eq: Lambda(a*a)-Lambda(a*)a - a*Lambda(a)} that $0 =  \sum_{j\in J} \tr([\rho,B_j^*]^* [\rho,B_j^*])$. So we see that $\rho$ commutes with all $B_j^*$. 
Since with $\rho$ also $\rho^*$ belongs to $\ker A$, we see by the same arguments that also $\rho^*$ commutes with all $B_j^*$. 
Consequently, $W(\rho) = 0$ and thus $Z_0(\rho) = 0$ as well. 
\end{proof}

In the second implication of part~(ii) of the above proposition, the assumption that $H$ have pure point spectrum is essential. See the example below. 

\begin{lm} \label{lm: ker Z_0 endldim}
Suppose $H: D(H) \subset \frak{h} \to \frak{h}$ is self-adjoint and suppose $Z_0$ is the generator of the 
semigroup on $S^p(\frak{h})$ defined by $e^{Z_0 t}(\rho) := e^{-iH t} \rho \, e^{i Ht}$, where $p \in [1,2]$.
\begin{itemize}
\item[(i)] If $\sigma_p(H)$ is finite and each $\mu \in \sigma_p(H)$ has finite multiplicity, then $\ker Z_0 = \{H\}' \cap S^p(\frak{h})$ is finite-dimensional, 
more precisely
\begin{align*}
\ker Z_0 = \spn \big\{ \scprd{e_{\mu\,i}, \,.\,} e_{\mu\,j}: \mu \in \sigma_p(H) \text{\, and \,} i,j \in \{1, \dots, n_{\mu}\} \big\}, 
\end{align*}
where $\{ e_{\mu\,i}: i \in \{1, \dots, n_{\mu}\} \}$ is an orthonormal basis of $
\ker(H-\mu)$ for every $\mu \in \sigma_p(H)$. In particular, $\ker Z_0 = 0$ in case $\sigma_p(H) = \emptyset$.

\item[(ii)] If $\frak{h}$ is infinite-dimensional, then $\{ H \}'$ is infinite-dimensional. 
\end{itemize}
\end{lm}

\begin{proof}
(i) Clearly, $\ker Z_0 = \{H\}' \cap S^p(\frak{h})$. 
%
If $\rho \in \ker Z_0$, then 
\begin{align*}
\rho = 1/T \int_0^T e^{Z_0 t}(\rho) \,dt = 1/T \int_0^T e^{-iHt} \rho \, e^{iHt} \, dt \longrightarrow \sum_{\mu \in \sigma_p(H)} Q_{\{\mu\}} \rho \, Q_{\{\mu\}}
\end{align*}
w.r.t.~the strong operator topology as $T \to \infty$ (Theorem~5.8 of~\cite{Teschl09}).
Since $Q_{\{\mu\}} = \sum_{i=1}^{n_{\mu}} \scprd{e_{\mu\,i}, \,.\,} e_{\mu\,i}$, we 
thus see that for $\rho \in \ker Z_0$,
\begin{align*}
\rho = \sum_{\mu \in \sigma_p(H)} Q_{\{\mu\}} \rho \, Q_{\{\mu\}}
= \sum_{\mu \in \sigma_p(H)} \sum_{i, j = 1}^{n_{\mu}} \scprd{ e_{\mu\, j}, \rho e_{\mu\,i} } \scprd{ e_{\mu\,i}, \,.\,} e_{\mu\,j} 
\end{align*}
belongs to $\spn \{ \scprd{e_{\mu\,i}, \,.\,} e_{\mu\,j}: \mu \in \sigma_p(H) \text{\, and \,} i,j \in \{1, \dots, n_{\mu}\} \}$. We have thus proved the first of the asserted inclusions and the second inclusion is obvious. 
\smallskip

(ii) 
In the case where $\sigma_p(H)$ is infinite or 
some $\mu \in \sigma_p(H)$ has infinite multiplicity, there exists an infinite orthonormal system $\{ \phi_n: n \in \N \}$ consisting of eigenvalues of $H$ and therefore the infinite subset
\begin{align*}
\{ \rho_n: n \in \N \} \qquad (\rho_n := \scprd{\phi_n, \,.\,}\phi_n)
\end{align*}
of $\{H\}' \cap S^1(\frak{h}) \subset \{H\}'$ is linearly independent, which proves the assertion. 
In the case where $\sigma_p(H)$ is finite and every $\mu \in \sigma_p(H)$ has finite multiplicity, there exists an interval $J = [k,k+1]$ with $k \in \Z$ 
such that $\sigma(H) \cap J = \sigma(H) \cap [k,k+1]$ 
is infinite. (If this was not the case, then every spectral value $\mu \in \sigma(H)$ would be isolated in $\sigma(H)$ and would hence be an eigenvalue of $H$. Consequently, $\sigma(H) = \sigma_p(H)$ and therefore $1 = Q_{\sigma(H)} = Q_{\sigma_p(H)} = \sum_{\mu \in \sigma_p(H)} Q_{\{\mu\}}$ 
would have finite rank. Contradiction!)
We now show that the infinite subset
\begin{align*}
\{ H_{J}^n: n \in \N \} \qquad (H_{J} := H Q_{J}) 
\end{align*} 
of $\{H\}'$ is linearly independent, which proves the assertion. 
Indeed, if there was a (finite) linear combination
\begin{align*}
0 = \sum_{k=1}^n \alpha_k H_J^k = p(H_J) 
\qquad (p(\mu) := \sum_{k=1}^n \alpha_k \mu^k) 
\end{align*}
with $\alpha_1, \dots, \alpha_n \in \C$ not all equal to $0$, then the spectral mapping theorem would yield $p(\sigma(H_J)) = \sigma(p(H_J)) = \{ 0 \}$
so that $\sigma(H_J)$ and, a fortiori, $\sigma(H) \cap J$ would have to be finite. Contradiction!
\end{proof}

With this lemma at hand, we can now convince ourselves that there exist weakly dephasing generators $A$ with $\ker A = \ker Z_0$ that are non-dephasing. 

\begin{ex} \label{ex: H vert nicht mit H_0}
We choose a self-adjoint operator $H: D(H) \subset \frak{h} \to \frak{h}$ (with spectral measure denoted by $Q$) in an infinite-dimensional Hilbert space $\frak{h}$ such that $\sigma_p(H)$ is finite and every $\mu \in \sigma_p(H)$ has finite multiplicity. We also  choose
\begin{align*}
B := \sum_{\mu \in \sigma_p(H)} \beta_{\mu} Q_{\{\mu\}} + \beta \scprd{\psi, \,.\,} \psi
\end{align*} 
where $\beta_{\mu} \in \C$ and $\beta \in \C \setminus \{0\}$ and where $\psi = H\phi /\norm{H\phi}$ and $\phi \in M^{\perp} \setminus \{0\}$ with
\begin{align*}
M := Q_{\sigma_p(H)}\frak{h} = \bigoplus_{\mu \in \sigma_p(H)} \ker(H-\mu).
\end{align*}
It should be noticed that 
$H\phi \ne 0$ because otherwise $\phi$  would be an eigenvector of $H$ 
contradicting $\phi \in M^{\perp}\setminus \{0\}$. It should also be noticed that $B$ is a normal operator because $\psi = H\phi /\norm{H\phi} \in H M^{\perp} \subset M^{\perp}$. 
We now define
\begin{align*}
A(\rho) := Z_0(\rho) + B\rho B^* - 1/2 \{B^*B,\rho\} \qquad (\rho \in D(Z_0))
\end{align*}
on $X = S^p(\mathfrak{h})$ with $p \in [1,2]$, where $Z_0$ is the generator of the semigroup on $S^p(\mathfrak{h})$ defined by $e^{Z_0 t}(\rho) := e^{-iHt} \, \rho \, e^{iHt}$.  
It is then clear that $A$ is a weakly dephasing generator of a quantum dynamical semigroup on $S^p(\frak{h})$.
With the help of Proposition~\ref{prop: properties (weakly) dephas generators}~(i) and Lemma~\ref{lm: ker Z_0 endldim}~(i) it also follows that 
\begin{align*}
\ker A \subset \ker Z_0 \quad \text{and} \quad \ker Z_0 \subset \ker A.
\end{align*}
%
And finally, $H B \ne B H$,
whence $B \notin \{ H \}''$. So, $A$ is not dephasing. (In order to see 
that $H$ indeed does not commute with $B$, 
compute
\begin{align*}
H B\phi = \beta \scprd{\psi,\phi} H\psi \quad \text{and} \quad B H\phi = \beta \norm{H\phi} \psi.
\end{align*}
In case $\scprd{\psi,\phi} = 0$, it follows that $HB \phi - BH \phi = -\beta \norm{H\phi} \psi \ne 0$ because $\beta \ne 0$. 
In case $\scprd{\psi,\phi} \ne 0$, it follows that $HB \phi - BH \phi = \beta \scprd{\psi,\phi} \big( H\psi - (\norm{H\phi}/\scprd{\psi,\phi}) \, \psi \big) \ne 0$ because otherwise $\psi$ would be an eigenvector of $H$ with corresponding eigenvalue $\norm{H\phi}/\scprd{\psi,\phi}$ and would therefore 
belong to $M$. Contradiction!)
$\blacktriangleleft$
\end{ex}

\section{Adiabatic theorems with spectral gap condition} \label{sect: adsaetze mit sl}

After having provided the most important preliminaries, we now prove an adiabatic theorem with uniform spectral gap condition (Section~\ref{sect: adsatz mit glm sl}) and an adiabatic theorem with non-uniform spectral gap condition (Section~\ref{sect: adsatz mit nichtglm sl}) for general operators $A(t)$. 
In these theorems the considered spectral subsets $\sigma(t)$ are only assumed to be compact so that, even if they are singletons, they need not consist of eigenvalues: they are allowed to be singletons consisting of essential singularities of the resolvent. In~\cite{AbouSalem07}, \cite{AvronGraf12}, \cite{Joye07} the case of poles is 
treated 
and in~\cite{AbouSalem07}, \cite{AvronGraf12} they are of order $1$.

\subsection{An adiabatic theorem with uniform spectral gap condition}  \label{sect: adsatz mit glm sl}

We begin by proving an adiabatic theorem with uniform spectral gap condition by extending Abou Salem's 
proof from~\cite{AbouSalem07}, which rests upon solving a suitable commutator equation. 

\begin{thm} \label{thm: handl adsatz mit glm sl}
Suppose $A(t): D \subset X \to X$ for every $t \in I$ is a linear operator such that Condition~\ref{cond: reg 1}  is satisfied with $\omega = 0$. 
Suppose further that $\sigma(t)$ for every $t \in I$ is a compact 
subset of $\sigma(A(t))$, 
that $\sigma(\,.\,)$ at no point falls into $\sigma(A(\,.\,))\setminus \sigma(\,.\,)$, and that $t \mapsto \sigma(t)$ is continuous. And finally, for every $t \in I$, let $P(t)$ be the 
projection associated with $A(t)$ and $\sigma(t)$ and suppose that $I \ni t \mapsto P(t)$ is in $W^{2,1}_*(I,L(X))$. Then 
\begin{align*}
\sup_{t \in I} \norm{ U_{\eps}(t) - V_{\eps}(t) } = O(\eps) \quad (\eps \searrow 0),
\end{align*}
where $U_{\eps}$ and $V_{\eps}$ are the evolution systems for $\frac 1 \eps A$ and $\frac 1 \eps A + [P',P]$.
\end{thm}

\begin{proof}
Since $\sigma(\,.\,)$ is uniformly isolated in $\sigma(A(\,.\,)) \setminus \sigma(\,.\,)$ and $t \mapsto \sigma(t)$ is continuous, 
there is, for every $t_0 \in I$, a non-trivial closed interval $J_{t_0} \subset I$ containing $t_0$ and a cycle $\gamma_{t_0}$ in $\rho(A(t_0))$ such that $\ran \gamma_{t_0} \subset \rho(A(t))$ and 
\begin{align*}
\operatorname{n}(\gamma_{t_0}, \sigma(t)) = 1 \quad \text{and} \quad \operatorname{n}(\gamma_{t_0}, \sigma(A(t)) \setminus \sigma(t)) = 0 
\end{align*}
for all $t \in J_{t_0}$.
We can now define 
\begin{align*}
 B(t)x := \frac{1}{2 \pi i} \int_{\gamma_{t_0}} (z-A(t))^{-1} P'(t) (z-A(t))^{-1} x \, dz
\end{align*}
for all $t \in J_{t_0}$, $t_0 \in I$ and $x \in X$. Since $\rho(A(t)) \ni z \mapsto (z-A(t))^{-1} P'(t) (z-A(t))^{-1} x$ is a holomorphic $X$-valued map (for all $x \in X$) and since the cycles $\gamma_{t_0}$ and $\gamma_{t_0'}$ are 
homologous in $\rho(A(t))$ whenever $t$ lies both in $J_{t_0}$ and in $J_{t_0'}$, the path integral exists in $X$ 
and does not depend on the special choice of $t_0 \in I$ with the property that $t \in J_{t_0}$. In other words, $t \mapsto B(t)$ is 
well-defined on $I$.
\smallskip

As a first preparatory step, we easily infer from the closedness of $A(t)$ that $B(t)X \subset D(A(t)) = D = Y$ and that
\begin{align} \label{eq: commutator equation}
 B(t) A(t) - A(t) B(t) \subset [P'(t),P(t)]
\end{align}
for all $t \in I$, which commutator equation will be essential in the main part of the proof.
As a second preparatory step, we show that $t \mapsto B(t)$ is in $W^{1,1}_*(I,L(X,Y))$, which is not very surprising (albeit a bit technical). It suffices to show that $J_{t_0} \ni t \mapsto B(t)$ is in $W^{1,1}_*(J_{t_0},L(X,Y))$ for every $t_0 \in I$. We therefore fix $t_0 \in I$. 
Since $\rho(A(t)) \ni z \mapsto (z-A(t))^{-1}$ is continuous w.r.t.~the norm of $L(X,Y)$ for every $t \in J_{t_0}$, 
we see that $B(t)$ is in $L(X,Y)$ for every $t \in J_{t_0}$. 
We also see, by virtue of Lemma~\ref{lm: prod- und inversenregel}, that for every $z \in \ran \gamma_{t_0}$ the map $t \mapsto (z-A(t))^{-1} P'(t) (z-A(t))^{-1}$ is in $W^{1,1}_*(J_{t_0},L(X,Y))$ and $t \mapsto C(t,z) = C_1(t,z) + C_2(t,z) + C_3(t,z)$ is a $W^{1,1}_*$-derivative of it, 
where
\begin{align} \label{eq: W^{1,infty}-abl des integranden}
C_1(t,z) &= (z-A(t))^{-1} A'(t) (z-A(t))^{-1} P'(t) (z-A(t))^{-1}, \notag \\
&C_2(t,z) = (z-A(t))^{-1} P''(t) (z-A(t))^{-1}, \\
C_3(t,z) &= (z-A(t))^{-1} P'(t) (z-A(t))^{-1} A'(t) (z-A(t))^{-1}, \notag
\end{align}
and $A'$, $P''$ are arbitrary $W^{1,1}_*$-derivatives of $A$ and $P'$. 
Since $t \mapsto C(t,z)$ is strongly measurable for all $z \in \ran \gamma_{t_0}$, it follows that 
\begin{align*}
t \mapsto \frac{1}{2 \pi i} \int_{\gamma_{t_0}} C(t, z) \, dz
\end{align*}
is strongly measurable as well (as the strong limit of Riemann sums), and since $J_{t_0} \times \ran \gamma_{t_0} \ni (t,z) \mapsto (z-A(t))^{-1}$ is continuous w.r.t.~the norm of $L(X,Y)$ and hence bounded, it follows by~\eqref{eq: W^{1,infty}-abl des integranden} that
\begin{align*} 
t \mapsto \Big\| \frac{1}{2 \pi i} \int_{\gamma_{t_0}} C(t, z) \, dz \Big\|_{X, Y}
\end{align*}
has an integrable majorant. So $t \mapsto \frac{1}{2 \pi i} \int_{\gamma_{t_0}} C(t, z) \, dz$ is in $W^{0,1}_*(J_{t_0}, L(X,Y))$ 
and 
one easily concludes that
\begin{align*}
 B(t)x = B(t_0)x + \int_{t_0}^t \frac{1}{2 \pi i} \int_{\gamma_{t_0}} C(\tau, z)x \, dz \, d\tau
\end{align*}
for all $t \in J_{t_0}$ and $x \in X$, as desired.
\smallskip

After these preparations we can now turn to the main part of the proof. 
We fix $x \in D$ and let $V_{\eps}$ denote the evolution system for $\frac 1 \eps A + [P',P]$ (which really exists due to the well-posedness theorem recalled after Condition~\ref{cond: reg 1}). Then $s \mapsto U_{\eps}(t,s) V_{\eps}(s)x$ is continuously differentiable (Lemma~\ref{lm: zeitentw rechtsseit db}) 
and we get, exploiting the commutator equation~\eqref{eq: commutator equation} for $A$ and $B$, that
\begin{align} \label{eq:ad thm with sg, 1}
V_{\eps}(t)x - U_{\eps}(t)x &= U_{\eps}(t,s)V_{\eps}(s)x \big|_{s=0}^{s=t} = \int_0^t U_{\eps}(t,s) [P'(s),P(s)] V_{\eps}(s)x \, ds \notag \\
&= \int_0^t U_{\eps}(t,s) \bigl( B(s)A(s) - A(s)B(s) \bigr) V_{\eps}(s)x \, ds
\end{align}
for all $t \in I$. Since for every $t \in I$ the maps $s \mapsto V_{\eps}(s) \big|_{Y}$ and $s \mapsto U_{\eps}(t,s) \big|_{Y}$ are continuously differentiable on $[0,t]$ w.r.t.~the strong operator topology of $L(Y,X)$ (Lemma~\ref{lm: zeitentw rechtsseit db}) and hence belong to $W^{1,1}_*([0,t], L(Y,X))$, 
and since $s \mapsto B(s)$ belongs to $W^{1,1}_*([0,t], L(X,Y))$, 
we can further conclude 
that $s \mapsto U_{\eps}(t,s) B(s) V_{\eps}(s)x$ is in $W^{1,1}([0,t],X)$ by Lemma~\ref{lm: prod- und inversenregel}, so that by~\eqref{eq:ad thm with sg, 1}
\begin{align} \label{eq:ad thm with sg, 2}
&V_{\eps}(t)x - U_{\eps}(t)x = \eps \int_0^t U_{\eps}(t,s) \Bigl( - \, \frac 1 \eps A(s)B(s) + B(s) \frac 1 \eps A(s) \Bigr) V_{\eps}(s)x \, ds \\
&\quad = \eps \, U_{\eps}(t,s) B(s) V_{\eps}(s)x \big|_{s=0}^{s=t} 
- \eps \int_0^t U_{\eps}(t,s) \bigl( B'(s) + B(s) [P'(s),P(s)] \bigr) V_{\eps}(s)x \, ds  \notag
\end{align}
for all $t \in I$ and $\eps \in (0,\infty)$, where $B'$ denotes an arbitrary $W^{1,1}_*$-derivative of $B$. 
Since $U_{\eps}$ and $V_{\eps}$ are both bounded above by an $\eps$-independent constant (Condition~\ref{cond: reg 1} with $\omega = 0$ and Lemma~\ref{lm:stoersatz (M,omega)-stab}), the assertion of the theorem immediately follows from~\eqref{eq:ad thm with sg, 2}.
\end{proof}

%

\subsection{An adiabatic theorem with non-uniform spectral gap condition}   \label{sect: adsatz mit nichtglm sl}

We continue by proving an adiabatic theorem with non-uniform spectral gap condition where $\sigma(\,.\,)$ falls into $\sigma(A(\,.\,)) \setminus \sigma(\,.\,)$ at countably many points that, 
in turn, accumulate at only finitely many points. 
We do so by extending Kato's 
proof from~\cite{Kato50} where finitely many eigenvalue crossings for skew-adjoint $A(t)$ are treated.

\begin{thm} \label{thm: handl adsatz mit nichtglm sl}
Suppose $A(t): D \subset X \to X$ for every $t \in I$ is a linear operator such that Condition~\ref{cond: reg 1} is satisfied with $\omega = 0$. 
Suppose further that $\sigma(t)$ for every $t \in I$ is a compact 
subset of $\sigma(A(t))$, that $\sigma(\,.\,)$ at countably many points accumulating at only finitely many points 
falls into $\sigma(A(\,.\,))\setminus \sigma(\,.\,)$, and that $I \setminus N \ni t \mapsto \sigma(t)$ is continuous, where $N$ denotes the set of 
those points where $\sigma(\,.\,)$ falls into $\sigma(A(\,.\,))\setminus \sigma(\,.\,)$. And finally, for every $t \in I \setminus N$, let $P(t)$ be the 
projection associated with $A(t)$ and $\sigma(t)$ 
and suppose that $I \setminus N \ni t \mapsto P(t)$ extends to a map (again denoted by $P$) in $W^{2,1}_*(I,L(X))$. Then 
\begin{align*}
\sup_{t \in I} \norm{ U_{\eps}(t) - V_{\eps}(t) } \longrightarrow 0 \quad (\eps \searrow 0),
\end{align*}
where $U_{\eps}$ and $V_{\eps}$ are the evolution systems for $\frac 1 \eps A$ and $\frac 1 \eps A + [P',P]$.
\end{thm}

\begin{proof}
We first prove the assertion in the case where $\sigma(\,.\,)$ at only finitely many points $t_1, \dots, t_m$ (ordered in an increasing way) 
falls into $\sigma(A(\,.\,))\setminus \sigma(\,.\,)$. So let $\eta >0$. We partition the interval $I$ as follows: 
\begin{align*}  
 I = I_{0 \, \delta} \cup J_{1 \, \delta} \cup I_{1 \, \delta} \cup \dots \cup J_{m \, \delta} \cup I_{m \, \delta},
\end{align*}
where $J_{i \, \delta}$ for $i \in \{1,\dots, m\}$ is a relatively open subinterval of $I$ containing $t_i$ of length less than $\delta$ (which will be chosen in a minute) and where $I_{0 \, \delta}$, \dots, $I_{m \, \delta}$ are the closed subintervals of $I$ lying between the subintervals
$J_{1 \, \delta}$, \dots, $J_{m \, \delta}$. In the following, we set $t_{i \, \delta}^- := \inf I_{i \, \delta}$ and $t_{i \, \delta}^+ := \sup I_{i \, \delta}$ for $i \in \{0, \dots, m\}$, and we choose $c$ so large that $\norm{P(s)}$, $\norm{P'(s)}$ and
$\norm{ [P'(s),P(s)] } \le c$ for all $s \in I$.
Since 
\begin{align*}
 \norm{ V_{\eps}(t,t_{i-1 \, \delta}^+)x - U_{\eps}(t, t_{i-1 \, \delta}^+)x } &= \norm{  \int_{t_{i-1 \, \delta}^+}^t U_{\eps}(t, s) [P'(s), P(s)] V_{\eps}(s, t_{i-1 \, \delta}^+)x \, ds  } \\
 &\le M c M e^{M c} \, \delta \norm{x}
\end{align*}
for every $t \in J_{i \, \delta}$, $x \in D$ and $\eps \in (0,\infty)$, we can achieve -- by choosing $\delta$ small enough -- that
\begin{align}  \label{eq: adsatz mit nichtglm sl 1}
\sup_{t \in J_{i \, \delta}} \norm{ V_{\eps}(t,t_{i-1 \, \delta}^+) - U_{\eps}(t, t_{i-1 \, \delta}^+) } < \frac{\eta}{ \big( 4 M^2 e^{2 Mc} \big)^m }
\end{align}
for every $\eps \in (0,\infty)$ and $i \in \{1, \dots, m\}$.
And since $\sigma(\,.\,) \big|_{I_{i \, \delta}}$ at no point falls into $\big(  \sigma(A(\,.\,))\setminus \sigma(\,.\,)  \big) \big|_{I_{i \, \delta}}$, we conclude from the above adiabatic theorem with uniform spectral gap condition 
(applied to the restricted data $A|_{I_{i \, \delta}}$, $\sigma|_{I_{i \, \delta}}$, $P|_{I_{i \, \delta}}$) 
that there is an $\eps_{\delta} \in (0, \infty)$ such that
\begin{align}   \label{eq: adsatz mit nichtglm sl 2}
 \sup_{t \in I_{i \, \delta}}  \norm{ V_{\eps}(t,t_{i \, \delta}^-) - U_{\eps}(t, t_{i \, \delta}^-) } < \frac{\eta}{ \big( 4 M^2 e^{2 Mc} \big)^m }
\end{align}
for every $\eps \in (0, \eps_{\delta})$ and $i \in \{0, \dots, m\}$.
Combining the estimates~\eqref{eq: adsatz mit nichtglm sl 1} and~\eqref{eq: adsatz mit nichtglm sl 2} and using the product property from the definition of evolution systems, we readily conclude for every $i \in \{1, \dots, m\}$ that
\begin{align*}
 \norm{V_{\eps}(t)-U_{\eps}(t)} < \frac{\eta}{ \big( 4 M^2 e^{2 Mc} \big)^{m-i} } \le \eta
\end{align*}
for all $t \in I_{i-1 \, \delta} \cup J_{i \, \delta} \cup I_{i \, \delta}$ and $\eps \in  (0,\eps_{\delta})$, and the desired conclusion follows. 
\smallskip

We now prove the assertion in the case where $\sigma(\,.\,)$ at infinitely many points accumulating at only finitely many points $t_1, \dots, t_m$ (ordered in an increasing way) falls into $\sigma(A(\,.\,))\setminus \sigma(\,.\,)$. In order to do so, we partition $I$ and choose $\delta$ as we did above.
We then obtain the estimate~\eqref{eq: adsatz mit nichtglm sl 1} as above and the estimate~\eqref{eq: adsatz mit nichtglm sl 2} 
by realizing that  $\sigma(\,.\,) \big|_{I_{i \, \delta}}$ at only finitely many points falls into $\big(  \sigma(A(\,.\,))\setminus \sigma(\,.\,)  \big) \big|_{I_{i \, \delta}}$ (so that the case just proved 
can be applied). And from these estimates 
the conclusion follows as above.
\end{proof}

It should be noticed that, in the situation of the above theorem, one has $P(t) A(t) \subset A(t) P(t)$ for every $t \in I$ (although a priori this is clear only for $t \in I \setminus N$), which follows by a continuity argument. 
(Indeed, 
if $t_0 \in I$ then it can be approximated by a sequence $(t_n)$ in $I \setminus N$. Since $t \mapsto (A(t_0)-1)(A(t)-1)^{-1}$ is norm continuous (by the $W^{1,1}_*$-regularity of $t \mapsto A(t)$ and Lemma~\ref{lm: prod- und inversenregel}), we see that for any $x \in D$ 
\begin{align} \label{eq: P(t) vertauscht mit A(t) fuer alle t}
A(t_0)P(t_n)x &= (A(t_0)-1)(A(t_n)-1)^{-1} \, P(t_n) (A(t_n)-1)x + P(t_n)x \nonumber \\
&\longrightarrow P(t_0)A(t_0)x \quad (n \to \infty)
\end{align}
and therefore $P(t_0)A(t_0) \subset A(t_0)P(t_0)$ by the closedness of $A(t_0)$.)
In particular, the evolution $V_{\eps}$ appearing in the above theorem really is adiabatic w.r.t.~to $P$ by Proposition~\ref{prop: intertwining relation}, as it should be. 



\subsection{Some remarks and examples}  \label{sect: bsp, adsaetze mit sl}

We begin with three remarks concerning the adiabatic theorems with uniform and non-uniform spectral gap condition alike.
\smallskip

1. In the special situation where $\sigma(t) = \{ \lambda(t) \}$ and $\lambda(t)$ 
is a pole of the resolvent map $(\,.\,-A(t))^{-1}$ of order at most $m_0 \in \N$ for all $t \in I$, the operators $B(t)$ -- used in the proof of the adiabatic theorems with spectral gap condition above to solve the commutator equation~\eqref{eq: commutator equation} -- can be cast in a form, 
namely~\eqref{eq: wegweiser}, 
which points the way to the solution of 
an appropriate (approximate) commutator equation 
in 
the adiabatic theorems without spectral gap condition below. 
%
%
Since $PP'P$, $\ol{P}P'\ol{P} = 0$ by~\eqref{eq: PP'P=0} (where $\ol{P} := 1-P$) and 
\begin{align*}
(z-A(t))^{-1}P(t) = \frac{1}{z-\lambda(t)} \Big( 1- \frac{A(t)-\lambda(t)}{z-\lambda(t)} \Big)^{-1} P(t) 
= \sum_{k=0}^{m_0-1} \frac{ (A(t)-\lambda(t))^k P(t)}{ (z-\lambda(t))^{k+1} } 
\end{align*}
for every $z \in \rho(A(t))$ by Theorem~5.8-A of~\cite{Taylor58}, we see that  
\begin{gather*}
B(t) = \sum_{k = 0}^{m_0-1} \frac{1}{2 \pi i} \int_{\gamma_t} \frac{\ol{R}(t,z)}{(z-\lambda(t))^{k+1}}  \, dz \,\, P'(t) (A(t)-\lambda(t))^k P(t) \qquad \qquad \\
\qquad \qquad \qquad + \sum_{k = 0}^{m_0-1} (A(t)-\lambda(t))^k P(t) P'(t) \,\, \frac{1}{2 \pi i} \int_{\gamma_t} \frac{\ol{R}(t,z)}{(z-\lambda(t))^{k+1}}  \, dz,
\end{gather*}
and since the reduced resolvent map $z \mapsto \ol{R}(t,z) := (z-A(t)|_{\ol{P}(t)D(A(t))})^{-1} \ol{P}(t)$ is holomorphic on $\rho(A(t)) \cup \{ \lambda(t) \}$, we further see -- using Cauchy's theorem -- that
\begin{align} \label{eq: wegweiser}
&B(t) = \sum_{k = 0}^{m_0-1} \ol{R}(t,\lambda(t))^{k+1} P'(t) (\lambda(t)-A(t))^k P(t)  \notag \\
&\qquad \qquad \qquad \qquad \qquad \quad + \sum_{k = 0}^{m_0-1} (\lambda(t)-A(t))^k P(t) P'(t) \ol{R}(t,\lambda(t))^{k+1}.
\end{align}

2.  In the even more special situation where $\sigma(t) = \{ \lambda(t) \} \subset i \R$ and $\lambda(t)$ is a pole of the resolvent map $(\, . \, - A(t))^{-1}$, the hypotheses of the above adiabatic theorem with uniform spectral gap condition become essentially -- apart from regularity conditions -- equivalent to the hypotheses of the respective adiabatic theorem (Theorem~9) of~\cite{AvronGraf12}, 
and a similar equivalence holds true for the above adiabatic theorem with non-uniform spectral gap condition. 
Indeed, if $\sigma(t)$ for every $t \in I$ is a singleton consisting of a pole $\lambda(t)$ on the imaginary axis, then the order $m(t)$ of nilpotence of $A(t)|_{P(t)D}-\lambda(t)$ 
must be equal to~$1$, since otherwise the relation
\begin{align}  \label{eq: ascent = 1 fuer erz beschr halbgr}
\delta \big( \lambda(t)+\delta -A(t) \big)^{-1} P(t) 
= \sum_{k=0}^{m(t)-1} \frac{  (A(t)-\lambda(t))^k }{\delta^k} \,  P(t)
\end{align}
would yield the contradiction that the right hand side of~\eqref{eq: ascent = 1 fuer erz beschr halbgr} explodes as $\delta \searrow 0$ whereas the left hand side of~\eqref{eq: ascent = 1 fuer erz beschr halbgr} remains bounded as $\delta \searrow 0$ (by virtue of the $(M,0)$-stability of $A$ and by $\lambda(t) \in i \R$). And therefore (by Theorem~5.8-A of~\cite{Taylor58}) $P(t)X = \ker(A(t)-\lambda(t))$ and $(1-P(t))X = \ran(A(t)-\lambda(t))$ 
as in~\cite{AvronGraf12}.
\smallskip

3. We finally remark 
that the above adiabatic theorems -- along with the commutator equation method used in their proofs -- 
can be extended to several subsets $\sigma_1(t)$, \dots, $\sigma_r(t)$ of $\sigma(A(t))$. 
If $A$, $\sigma_j$, $P_j$ for every $j \in \{1, \dots, r\}$ satisfy the hypotheses of the above adiabatic theorem with uniform or non-uniform spectral gap and if $\sigma_j(\,.\,)$ and $\sigma_l(\,.\,)$ for all $j \ne l$ fall into each other at only countably many points accumulating at only finitely many points, then there exists an evolution system $V_{\eps}$, namely that for $\frac 1 \eps A + K$ with
\begin{align} \label{eq: K mehrere sigma_j}
K(t) := \frac 1 2 \sum_{j=1}^{r+1} [P_j'(t),P_j(t)] \quad \text{and} \quad 
P_{r+1}(t) := 1-P(t) := 1 - \sum_{j=1}^r P_j(t),
\end{align}
which on the one hand is simultaneously adiabatic w.r.t.~all the $P_j$ by~\cite{Kato50} and on the other hand well approximates the evolution system $U_{\eps}$ for $\frac 1 \eps A$ in the sense that
\begin{align*}
\sup_{t \in I} \norm{ U_{\eps}(t) - V_{\eps}(t) } \longrightarrow 0 \quad (\eps \searrow 0).
\end{align*}
In order to see this, one has only to observe 
that $B(t) := \frac 1 2 \sum_{j=1}^{r+1} B_j(t)$ with
\begin{gather}
B_j := \frac{1}{2 \pi i} \int_{\gamma_j} (z-A)^{-1} P_j' (z-A)^{-1} \, dz \quad (j \in \{1, \dots, r\}) \notag \\
B_{r+1} :=  \frac{1}{2 \pi i} \int_{\gamma} (z-A)^{-1} P' (z-A)^{-1} \, dz  \label{eq: B_{r+1}} \\ 
\gamma := \gamma_1 + \dotsb + \gamma_r  \text{\, $(\gamma_j = \gamma_{j \, t}$ as in the proofs above) and } P := P_1 + \dotsb + P_r \notag
\end{gather}
solves the commutator equation $B(t)A(t) - A(t)B(t) \subset K(t)$ for all points $t$ where no crossing 
takes place (because $[P_{r+1}', P_{r+1}] = [P',P]$) and then to proceed as in the proofs of the adiabatic theorems above. See also~\cite{Panati10}.
In the special case of skew-adjoint operators $A(t)$ one can further refine 
the statement above: it is then possible to show -- by further adapting the commutator equation method -- 
that even the 
evolution system $\ol{V}_{\eps}$ for $\frac 1 \eps A + \ol{K}$ with
\begin{align*}
\ol{K}(t) := \frac 1 2 \Big( [(P_{r+1}^-)'(t),P_{r+1}^-(t)] + \sum_{j=1}^r [P_j'(t),P_j(t)] + [(P_{r+1}^+)'(t),P_{r+1}^+(t)] \Big)
\end{align*}
well approximates the evolution system $U_{\eps}$ for $\frac 1 \eps A$ -- notice that $\ol{V}_{\eps}$ is 
is not only adiabatic w.r.t.~$P_{r+1} = P_{r+1}^- + P_{r+1}^+$ 
but also w.r.t.~$P_{r+1}^-$ and $P_{r+1}^+$ separately, where $P_{r+1}^{\pm}(t)$ are the spectral projections of $A(t)$ corresponding to the parts $\sigma^{\pm}(t)$ of the spectrum 
which on $i \R$ are located 
below respectively above all the compact parts $\sigma_1(t)$, \dots, $\sigma_r(t)$.
In order to see this, 
set
\begin{align*}
B_{r+1 \, n}^{\pm}(t) := \frac{1}{2 \pi i} \int_{\gamma_{n \, t}^{\pm}} (z-A(t))^{-1} (P_{r+1}^{\pm})'(t) (z-A(t))^{-1} \, dz
\end{align*}
where $\gamma_{n \, t}^{\pm}(\tau) := \pm \, \tau + c^{\pm}(t)$ for $\tau \in [-n,n]$ 
with points $c^{\pm}(t) \in i \R$ lying in the gap between $\sigma^{\pm}(t)$ and the rest of $\sigma(A(t))$ 
and depending continuously differentiably on $t$, 
and observe that (by the skew-adjointness of $A(t)$)
\begin{gather*}
P_{r+1 \, n}^{\pm}(t)x := \frac{1}{2 \pi i} \int_{\gamma_{n \, t}^{\pm}} (z-A(t))^{-1}x \, dz \longrightarrow P_{r+1}^{\pm}(t)x - \frac 1 2 x \quad (n \to \infty) \\
\text{and} \\
\big\| B_{r+1 \, n}(t) \big\|, \norm{B_{r+1 \, n}'(t)} \le \int_{-\infty}^{\infty} \frac{c}{\dist\big( \gamma_{n \, t}^{\pm}(\tau) ,\sigma(A(t)) \big)^2 } \, d \tau \le C < \infty \quad (n \in \N, t \in I).
\end{gather*}  
A slightly less general 
general statement was first proven in~\cite{Nenciu80} by a different method than the commutator equation technique indicated above. 
\smallskip

We close this section 
with a simple example 
showing that the conclusion of the above adiabatic theorems 
will, in general, fail if $A$ is not $(M,0)$-stable. 

\begin{ex} \label{ex: (M,0)-stabilitaet wesentl, mit sl}
Suppose $A$, $\sigma$, $P$ with $A(t) := R(t)^{-1} A_0(t) R(t)$, $\sigma(t) := \{\lambda(t)\}$ and $P(t) := R(t)^{-1} P_0 R(t)$ are given as follows in $X := \ell^2(I_2)$:
\begin{align*}
A_0(t) := \begin{pmatrix} \lambda(t) & 0 \\ 0 & 0 \end{pmatrix}, \quad  P_0 := \begin{pmatrix} 1 & 0 \\ 0 & 1 \end{pmatrix}, \quad  R(t) := e^{C t}  \text{ \, with \, } C := 2 \pi \begin{pmatrix} 0  & 1 \\ -1 & 0 \end{pmatrix},
\end{align*}
and $t \mapsto \lambda(t) \in [0,\infty)$ is absolutely continuous such that $\lambda(\,.\,)$ at only countably many points accumulating at only finitely many points falls into $0$. 
Then all the hypotheses of Theorem~\ref{thm: handl adsatz mit nichtglm sl} 
are fullfilled with the sole exception that $A$ is not $(M,0)$-stable (because $\sigma(A(t)) = \{ 0, \lambda(t) \}$ is contained in the closed left half-plane only for countably many $t \in I$) 
and, in fact, the conclusion of this theorems 
fails. Indeed, since
\begin{align*}
R(t) = e^{C t} = \begin{pmatrix} \cos(2 \pi t) & \sin(2\pi t) \\ -\sin(2\pi t) & \cos(2\pi t) \end{pmatrix},
\end{align*}
we see that
\begin{align*}
A(t) = R(t)^{-1} A_0(t) R(t) = \lambda(t) \begin{pmatrix} \cos^2(2\pi t) & \cos(2\pi t) \sin(2 \pi t) \\ \cos(2\pi t) \sin(2 \pi t)  & \sin^2(2\pi t) \end{pmatrix}
\end{align*}
is a positive linear operator (in the lattice sense) for all $t \in [0,t_0]$ with $t_0 := \frac{1}{4}$. And since $1-P(t_0) = P_0$, we see (by the series expansion for $U_{\eps}$) 
that
\begin{align*}
&\norm{(1-P(t_0))U_{\eps}(t_0)P(0) e_1} = \big| \scprd{ e_1, U_{\eps}(t_0) e_1} \big| = \scprd{ e_1, U_{\eps}(t_0) e_1} \\
&\qquad \qquad \ge 1 + \frac 1 \eps \int_0^{t_0} \scprd{e_1, A(\tau) e_1} \, d\tau = 1 + \frac 1 \eps \int_0^{t_0} \lambda(\tau) \cos^2(2\pi \tau) \, d\tau,
\end{align*} 
which right hand side does not converge to $0$ as $\eps \searrow 0$, as desired. 
$\blacktriangleleft$
\end{ex}

An example with non-diagonalizable $A(t)$ and $\sigma(A(t)) = \{0,i\}$ showing as well 
that the conclusion of the above adiabatic theorems 
will generally fail if the family $A$ is not $(M,0)$-stable 
can be found in Joye's paper~\cite{Joye07} at the end of Section~1.

\section{Adiabatic theorems without spectral gap condition} \label{sect: adsaetze ohne sl}

After having established 
general adiabatic theorems with spectral gap condition in Section~\ref{sect: adsaetze mit sl}, 
we can 
now prove 
an adiabatic theorem without spectral gap condition for general operators $A(t)$ with not necessarily weakly semisimple spectral values $\lambda(t)$. In Section~\ref{sect: qual adsatz ohne sl} it appears in a qualitative version 
and in Section~\ref{sect: quant adsatz ohne sl} in a quantitatively refined version, 
and both versions are applied to the special case of spectral operators.
%
We thereby generalize the recent adiabatic theorems without spectral gap condition of Avron, Fraas, Graf, Grech from~\cite{AvronGraf12} and of Schmid from~\cite{dipl}, which -- although independently obtained -- are essentially the same (except for some regularity subtleties). 

\subsection{A qualitative adiabatic theorem without spectral gap condition}  \label{sect: qual adsatz ohne sl}

We begin with a lemma that will be crucial in the proofs of the presented adiabatic theorems without spectral gap condition.
%

\begin{lm} \label{lm: lm 1 zum erw adsatz ohne sl}
Suppose that $A: D(A) \subset X \to X$ is a densely defined closed linear operator 
and that $\lambda \in \sigma(A)$ and $\delta_0 \in (0,\infty)$ and $\vartheta_0 \in \R$ such that $\lambda + \delta e^{i \vartheta_0} \in \rho(A)$ for all $\delta \in (0,\delta_0]$. Suppose finally that $P$ is a bounded projection in $X$ such that $P A \subset A P$ and
\begin{align*} 
(1-P)X \subset \ol{\ran}\,(A-\lambda)^{m_0}
\end{align*}
for some $m_0 \in \N$, and that there is $M_0 \in (0,\infty)$ such that
\begin{align} \label{eq: resolvent estimate, lm}
\norm{ \big( \lambda + \delta e^{i \vartheta_0} - A \big)^{-1} (1-P) } \le \frac{M_0}{\delta}
\end{align}
for all $\delta \in (0, \delta_0]$. 
Then $\delta \big( \lambda + \delta e^{i \vartheta_0} - A \big)^{-1} (1-P)x \longrightarrow 0$ as $\delta \searrow 0$ for all $x \in X$. 
\end{lm}

\begin{proof}
If $x \in \ran(A-\lambda)^{m_0}$, then $x = (\lambda - A)^{m_0}x_0$ for some $x_0 \in D(A^{m_0})$ and by~\eqref{eq: resolvent estimate, lm} 
\begin{align*} 
\delta \big( \lambda + \delta e^{i \vartheta_0} - A \big)^{-1} \ol{P}x 
=& \,\, \delta \big( \lambda + \delta e^{i \vartheta_0} - A \big)^{-1} \ol{P} \big(-\delta e^{i \vartheta_0} \big)^{m_0}x_0 \notag \\
&+ \delta \sum_{k=1}^{m_0} \binom{m_0}{k} \big( \lambda + \delta e^{i \vartheta_0} - A \big)^{k-1} \big(-\delta e^{i \vartheta_0} \big)^{m_0-k} \ol{P} x_0
\longrightarrow 0
\end{align*}
as $\delta \searrow 0$, where of course $\ol{P} := 1-P$.
And if $x \in X$, then $\ol{x} := \ol{P}x$ can be approximated arbitrarily well by elements $y$ of $\ran(A-\lambda)^{m_0}$ and therefore
\begin{align*}
\delta \big( \lambda + \delta e^{i \vartheta_0} - A \big)^{-1} \ol{P}x = \delta \big( \lambda + \delta e^{i \vartheta_0} - A \big)^{-1} \ol{P} (\ol{x}-y) +  \delta \big( \lambda + \delta e^{i \vartheta_0} - A \big)^{-1} \ol{P} y 
\end{align*}
can be made arbitrarily small for $\delta$ small enough by~\eqref{eq: resolvent estimate, lm} and the first step. 
\end{proof}

With this lemma at hand, we can now prove the announced 
general adiabatic theorem without spectral gap condition for not necessarily weakly semisimple eigenvalues. 
Similarly to the works~\cite{AvronElgart99} of Avron and Elgart and~\cite{Teufel01} of Teufel 
its proof rests upon solving a suitable approximate commutator equation. 
In this undertaking the insights from the special case of poles, 
especially formula~\eqref{eq: wegweiser}, will prove to be most helpful. 

\begin{thm} \label{thm: erw adsatz ohne sl}
Suppose $A(t): D \subset X \to X$ for every $t \in I$ is a linear operator such that Condition~\ref{cond: reg 1} is satisfied with $\omega = 0$. Suppose further that $\lambda(t)$ for every $t \in I$ is an eigenvalue of $A(t)$, and that there are numbers $\delta_0 \in (0,\infty)$ and $\vartheta(t) \in \R$ such that $\lambda(t) + \delta e^{i \vartheta(t)} \in \rho(A(t))$ for all $\delta \in (0,\delta_0]$ and $t \in I$ and such that $t \mapsto \lambda(t)$ and $t \mapsto e^{i \vartheta(t)}$ are absolutely continuous. 
Suppose finally that $P(t)$ for every $t \in I$ is a bounded projection in $X$ 
such that $P(t)$ for almost every $t \in I$ is weakly associated with $A(t)$ and $\lambda(t)$, suppose there is an $M_0 \in (0,\infty)$ such that 
\begin{align} \label{eq: resolvent estimate}
\norm{ \big( \lambda(t) + \delta e^{i \vartheta(t)} - A(t) \big)^{-1} (1-P(t)) } \le \frac{M_0}{\delta} 
\end{align}
for all $\delta \in (0, \delta_0]$ and $t \in I$, let $\rk P(0) < \infty$ and suppose that $t \mapsto P(t)$ is strongly continuously differentiable. 
%
\begin{itemize}
\item[(i)] If $X$ is arbitrary (not necessarily reflexive), then
\begin{align*} 
\sup_{t \in I} \norm{ \big( U_{\eps}(t) - V_{0 \, \eps}(t) \big) P(0) } \longrightarrow 0 \quad (\eps \searrow 0),
\end{align*}
where $U_{\eps}$ and $V_{0 \, \eps}$ are the evolution systems for $\frac 1 \eps A$ and $\frac 1 \eps  A P + [P',P]$ on $X$ for every $\eps \in (0,\infty)$. 
\item[(ii)] If $X$ is reflexive and $t \mapsto P(t)$ is norm continuously differentiable, then
\begin{align*}
\sup_{t \in I} \norm{ U_{\eps}(t) - V_{\eps}(t) } \longrightarrow 0 \quad (\eps \searrow 0),
\end{align*}
whenever the evolution system $V_{\eps}$ for $\frac 1 \eps A + [P',P]$ exists on $D$ for every $\eps \in (0, \infty)$.
\end{itemize}
\end{thm}

\begin{proof}
We begin with some preparations which will be 
used in the proof of both assertion~(i) and assertion~(ii).
As a first preparatory step, we show that $t \mapsto P(t)$ is in $W^{1,1}_*(I,L(X,Y))$ and 
conclude that $P(t)A(t) \subset A(t)P(t)$ for every $t \in I$ and that there is an $m_0 \in \N$ such that $P(t)X \subset \ker(A(t)-\lambda(t))^{m_0}$ for every $t \in I$.
%
%
Since $P(t)$ for almost every $t \in I$ is weakly associated with $A(t)$ and $\lambda(t)$ and since 
\begin{align*}
\dim P(t)X = \rk P(0)X < \infty
\end{align*}
for every $t \in I$ (which equality is due to the continuity of $t \mapsto P(t)$ and Lemma~VII.6.7 of~\cite{DunfordSchwartz}), there is a $t$-independent constant $m_0 \in \N$ -- for instance, $m_0 := \rk P(0)$ -- such that $P(t)$ is weakly associated of order $m_0$ 
with $A(t)$ and $\lambda(t)$ for almost every $t \in I$. 
In particular, it follows from Theorem~\ref{thm: typ mögl für PX und (1-P)X} that
\begin{align*}
P(t)X \subset \ker(A(t)-\lambda(t))^{m_0} \quad \text{and} \quad (1-P(t))X \subset \ol{\ran}\, (A(t)-\lambda(t))^{m_0} 
\end{align*} 
for almost every $t \in I$ (with exceptional set $N$).  
%
%
It now follows 
by the binomial formula that
\begin{align*}
&P(t) 
= S_{\delta}(t)^{m_0} \big( A(t)-\lambda(t)-\delta e^{i \vartheta(t)} \big)^{m_0} P(t)
= S_{\delta}(t)  \, \sum_{k=0}^{m_0-1} \binom{m_0}{k}  \big( -\delta e^{i \vartheta(t)} \big)^{m_0-k}  \cdot \notag \\ 
&\qquad \qquad \qquad \qquad \qquad \qquad \qquad \cdot S_{\delta}(t)^{m_0-1 - k} \, \big( 1 + \delta e^{i \vartheta(t)} S_{\delta}(t) \big)^k P(t) 
\end{align*}
for every $t \in I \setminus N$, where $S_{\delta}(t) := (A(t)-\lambda(t)-\delta e^{i \vartheta(t)})^{-1}$. Since both sides of this equation depend continuously on $t \in I$, the equation 
holds for every $t \in I$, and since the right-hand side belongs to $W^{1,1}_*(I,L(X,Y))$ by Lemma~\ref{lm: prod- und inversenregel}, 
we also have
\begin{align} \label{eq: gl -2, adsatz ohne sl}
(t \mapsto P(t)) \in W^{1,1}_*(I,L(X,Y)).
\end{align}
With this regularity property at hand, it is now easy to see that the inclusions 
\begin{align} 
P(t)A(t) \subset A(t)P(t) \quad \text{and} \quad 
P(t)X \subset \ker(A(t)-\lambda(t))^{m_0} \label{eq: gl 0, adsatz ohne sl}
\end{align}
also hold for $t \in N$ (while they clearly hold for $t \in I \setminus N$). 
%
In order to see that~(\ref{eq: gl 0, adsatz ohne sl}.a) holds also for $t \in N$, notice that every such $t$ is approximated by a sequence $(t_n)$ in $I \setminus N$ and hence 
\begin{gather*}
P(t_n)x \longrightarrow P(t)x, \\ 
A(t)P(t_n)x = (A(t)-A(t_n))P(t_n)x + P(t_n)A(t_n)x \longrightarrow P(t)A(t)x
\end{gather*}  
for every $x \in D(A(t)) = D$ by \eqref{eq: gl -2, adsatz ohne sl}. So, (\ref{eq: gl 0, adsatz ohne sl}.a) follows by the closedness of $A(t)$. (Alternatively, we could also have argued as in~\eqref{eq: P(t) vertauscht mit A(t) fuer alle t}.)
In order to see that~(\ref{eq: gl 0, adsatz ohne sl}.b) holds also for $t \in N$, notice that $\dim P(t)X = \rk P(0) < \infty$ and $P(t)A(t) \subset A(t)P(t)$ for every $t \in I$, so that $P(t)X = P(t)D(A(t))$ and $P(t)X \subset D(A(t)^{m_0})$ as well as 
\begin{align*}
(A(t)-\lambda(t))^{m_0} P(t) = \big( (A(t)-\lambda(t))P(t) \big)^{m_0}
\end{align*}
for every $t \in I$. So, (\ref{eq: gl 0, adsatz ohne sl}.b) follows by~\eqref{eq: gl -2, adsatz ohne sl}. 
\smallskip

As a second preparatory step, we solve -- in accordance 
with the proof of the adiabatic theorems with spectral gap condition -- 
a suitable (approximate) commutator equation. Inspired by~\eqref{eq: wegweiser}, we define the operators 
\begin{align}
&B_{n \, \bm{\delta}}(t) := \sum_{k=0}^{m_0-1} \Big(  \prod_{i=1}^{k+1} \ol{R}_{\delta_i}(t) \Big) Q_n(t) (\lambda(t)-A(t))^k P(t) \notag \\
&\qquad \qquad \qquad \qquad \qquad \quad + \sum_{k=0}^{m_0-1} (\lambda(t)-A(t))^k P(t) Q_n(t) \Big(  \prod_{i=1}^{k+1} \ol{R}_{\delta_i}(t) \Big)
\end{align}
for $n \in \N$, $\bm{\delta} := (\delta_1, \dots, \delta_{m_0}) \in (0,\delta_0]^{m_0}$ and $t \in I$, where
\begin{align*}
\ol{R}_{\delta}(t) := R_{\delta}(t) \ol{P}(t) \quad \text{with} \quad R_{\delta}(t) := \big( \lambda(t) + \delta e^{i \vartheta(t)} - A(t) \big)^{-1} \text{\, and \,\,} \ol{P}(t) := 1-P(t)
\end{align*}
for $\delta \in (0,\delta_0]$, and where 
\begin{align*}
Q_n(t) := \int_0^1 J_{1/n}(t-r) P'(r) \, dr
\end{align*}
with $(J_{1/n})$ being a standard mollifier in $C_c^{\infty}((0,1),\R)$. In other words, $Q_n$ is obtained from $P'$ by mollification, whence $t \mapsto Q_n(t)$ is strongly continuously differentiable and $Q_n(t) \longrightarrow P'(t)$ as $n \to \infty$ w.r.t.~the strong operator topology for 
$t \in (0,1)$ and
\begin{align*}
\sup \{ \norm{Q_n(t)}: t \in I, n \in \N \} \le \sup_{t \in I} \norm{P'(t)}.
\end{align*}
We now show that the operators $B_{n \, \bm{\delta}}(t)$ satisfy the approximate commutator equation
\begin{align} \label{eq: appr commutator equation}
B_{n \, \bm{\delta}}(t)A(t) - A(t)B_{n \, \bm{\delta}}(t) + C_{n \, \bm{\delta}}(t) \subset [Q_n(t),P(t)]
\end{align}
with remainder terms $C_{n \, \bm{\delta}}(t)$ that 
will have to be suitably controlled below. Since
\begin{align*}
(\lambda-A) \Big(  \prod_{i=1}^{k+1} \ol{R}_{\delta_i} \Big) =  \Big( \prod_{1 \le i \le k} \ol{R}_{\delta_i} \Big) - \delta_{k+1} e^{i \vartheta} \Big( \prod_{i=1}^{k+1} \ol{R}_{\delta_i} \Big) 
\supset \Big(  \prod_{i=1}^{k+1} \ol{R}_{\delta_i} \Big) (\lambda-A)
\end{align*}
(the $t$-dependence being suppressed here and in the following lines for the sake of convenience), it follows that
\begin{align*}
(\lambda-A) B_{n \, \bm{\delta}} &= \sum_{k=0}^{m_0-1} \Big( \prod_{1 \le i \le k} \ol{R}_{\delta_i} \Big) Q_n (\lambda-A)^k P   +  \sum_{k=0}^{m_0-1} (\lambda-A)^{k+1} P Q_n \Big(  \prod_{i=1}^{k+1} \ol{R}_{\delta_i} \Big)
- C_{n \, \bm{\delta}}^+ \\
B_{n \, \bm{\delta}} (\lambda-A) &\subset \sum_{k=0}^{m_0-1} \Big(  \prod_{i=1}^{k+1} \ol{R}_{\delta_i} \Big) Q_n (\lambda-A)^{k+1} P   +  \sum_{k=0}^{m_0-1} (\lambda-A)^k P Q_n  \Big( \prod_{1 \le i \le k} \ol{R}_{\delta_i} \Big)
- C_{n \, \bm{\delta}}^-
\end{align*}
where we used the abbreviations
\begin{align}
&C_{n \, \bm{\delta}}^+ := \sum_{k=0}^{m_0-1} \delta_{k+1} e^{i \vartheta} \Big( \prod_{i=1}^{k+1} \ol{R}_{\delta_i} \Big) Q_n (\lambda-A)^k P,  \notag \\
&\qquad \qquad \qquad \qquad \qquad C_{n \, \bm{\delta}}^- := \sum_{k=0}^{m_0-1} (\lambda-A)^k P Q_n \, \delta_{k+1} e^{i \vartheta} \Big( \prod_{i=1}^{k+1} \ol{R}_{\delta_i} \Big).
\end{align}
Subtracting $B_{n \, \bm{\delta}} (\lambda-A)$ from $(\lambda-A)B_{n \, \bm{\delta}}$ and noticing that, by doing so, of all the summands not 
belonging to $C_{n \, \bm{\delta}}^+$, $C_{n \, \bm{\delta}}^-$ only
\begin{align*}
Q_n P - \Big( \prod_{i=1}^{m_0} \ol{R}_{\delta_i} \Big) Q_n (\lambda-A)^{m_0} P    +   (\lambda-A)^{m_0} P Q_n \Big( \prod_{i=1}^{m_0} \ol{R}_{\delta_i} \Big) - P Q_n
= [Q_n,P]  
\end{align*}
remains (remember~\eqref{eq: gl 0, adsatz ohne sl}), 
we see that 
\begin{align*}
B_{n \, \bm{\delta}} A - A B_{n \, \bm{\delta}} \subset [Q_n,P] - C_{n \, \bm{\delta}}^+ + C_{n \, \bm{\delta}}^-
\end{align*}
which is nothing but~\eqref{eq: appr commutator equation} if one defines $C_{n \, \bm{\delta}} := C_{n \, \bm{\delta}}^+ - C_{n \, \bm{\delta}}^-$.
\smallskip

As a third 
preparatory step we observe that $t \mapsto B_{n \, \bm{\delta}}(t)$ belongs to $W^{1,1}_*(I,L(X,Y))$ and estimate $B_{n \, \bm{\delta}}$ as well as $B_{n \, \bm{\delta}}'$. Since 
\begin{align} \label{eq: (A(t)-lambda(t))^k P(t) regulaer}
t \mapsto (A(t)-\lambda(t))^k P(t) = \big( (A(t)-\lambda(t))P(t) \big)^k = P(t) \big( (A(t)-\lambda(t))P(t) \big)^k
\end{align}
is in $W^{1,1}_*(I,L(X,Y))$ by the first preparatory step 
the asserted $W^{1,1}_*(I,L(X,Y))$-regularity of $t \mapsto B_{n \, \bm{\delta}}(t)$ follows from 
Lemma~\ref{lm: prod- und inversenregel}.
Additionally, there is a constant $c$ such that
\begin{align} \label{eq: absch B_n eps}
\sup_{t \in I} \big\| B_{n \, \bm{\delta}}(t) \big\| \le \sum_{k=1}^{m_0} c \, \Big( \prod_{i=1}^{k} \delta_i \Big)^{-1}
\end{align}
for all $\bm{\delta} \in (0,\delta_0]^{m_0}$ by the assumed resolvent estimate and the continuity 
of~\eqref{eq: (A(t)-lambda(t))^k P(t) regulaer} just established. 
And since 
\begin{gather*}
\norm{R_{\delta}(t)}_{X, X} 
\le \sum_{k=0}^{m_0-1} \frac{1}{\delta^{k+1}} \norm{(A(t)-\lambda(t))^k P(t)}_{X, X}  + \,\, \norm{\ol{R}_{\delta}(t)}_{X, X}
\le \frac{c}{\delta^{m_0}} \\
\text{ as well as } \\
\norm{\ol{R}_{\delta}(t)}_{X, Y} \le \norm{(A(t)-1)^{-1}}_{X, Y} \norm{(A(t)-1) \ol{R}_{\delta}(t)}_{X, X} \le \frac{c}{\delta}
\end{gather*}
for all $t \in I$ and all $\delta \in (0,\delta_0]$ (with another constant $c$) by the assumed resolvent estimate and the continuity of of~\eqref{eq: (A(t)-lambda(t))^k P(t) regulaer} just established, 
it follows from Lemma~\ref{lm: prod- und inversenregel} that there is a $W^{1,1}_*$-derivative $\ol{R}_{\delta}'$ of $t \mapsto \ol{R}_{\delta}(t)$ such that
\begin{align} \label{eq: absch R_eps'}
\int_0^1 \big\| \ol{R}_{\delta}'(s) \big\| \,ds \le \frac{c}{\delta^{m_0+1}}
\end{align}
for all $\delta \in (0,\delta_0]$ (with yet another constant $c$) and, hence, that there is a $W^{1,1}_*$-derivative $B_{n \, \bm{\delta}}'$ of $t \mapsto B_{n \, \bm{\delta}}(t)$ such that
\begin{align}  \label{eq: absch B_n eps'}
\int_0^1 \norm{ B_{n \, \bm{\delta}}'(s) } \,ds \le \sum_{k=1}^{m_0} c_n \, \Big( \prod_{i=1}^{k} \delta_i \Big)^{-(m_0+1)}
\end{align}
for all $\bm{\delta} \in (0,\delta_0]^{m_0}$ and some constant $c_n \in (0,\infty)$ depending on the supremum norm $\sup_{t \in I} \norm{Q_n'(t)}$ of the strong derivative of $t \mapsto Q_n(t)$.
\smallskip

As a fourth and last preparatory step, we observe that for every $\eps \in (0,\infty)$ the evolution system $V_{0\,\eps}$ for $\frac 1 \eps AP + [P',P]$ exists on $X$ and is adiabatic w.r.t.~$P$ and satisfies the estimate
\begin{align} \label{eq: absch V_{0 eps}}
\norm{ V_{0\,\eps}(t,s) P(s) } \le Mc \, e^{Mc(t-s)} 
\end{align}
for all $(s,t) \in \Delta$, where $c$ is an upper bound of $t \mapsto \norm{P(t)}, \norm{ P'(t) }$. 
Indeed, $t \mapsto A(t)P(t)$ is strongly continuous (by the first preparatory step) and therefore 
the evolution system $V_{0\, \eps}$ for $\frac 1 \eps AP + [P',P]$ exists on $X$ (Theorem~5.1.1 of~\cite{Pazy}) and  
is adiabatic w.r.t.~$P$ for every $\eps \in (0,\infty)$ (by virtue of~(\ref{eq: gl 0, adsatz ohne sl}.a) and Proposition~\ref{prop: intertwining relation}) .  
It follows that for all $x \in X$ and $(s,t) \in \Delta$ the map $[s,t] \ni \tau \mapsto U_{\eps}(t,\tau)V_{0\,\eps}(\tau,s)P(s)x$ is continuously differentiable by Lemma~\ref{lm: zeitentw rechtsseit db} (use the adiabaticity of $V_{0\, \eps}$ w.r.t.~$P$ and (\ref{eq: gl 0, adsatz ohne sl}.b)) 
with derivative
\begin{gather*}
\tau \mapsto  \, \, U_{\eps}(t,\tau) \Big( \frac 1 \eps A(\tau)P(\tau) - \frac 1 \eps A(\tau) + [P'(\tau),P(\tau)] \Big) V_{0\,\eps}(\tau,s)P(s)x \\
= U_{\eps}(t,\tau) P'(\tau) V_{0\,\eps}(\tau,s)P(s)x, 
\end{gather*}
where in the last equation the adiabaticity of $V_{0\, \eps}$ w.r.t.~$P$ and~\eqref{eq: PP'P=0} are used.
So, 
\begin{align}  \label{eq: intgl, lm 2 zum erw adsatz ohne sl}
V_{0\,\eps}(t,s)P(s)x - U_{\eps}(t,s)P(s)x &= U_{\eps}(t,\tau) V_{0\,\eps}(\tau,s)P(s)x \big|_{\tau=s}^{\tau=t} \notag \\
&= \int_s^t U_{\eps}(t,\tau) P'(\tau) V_{0\,\eps}(\tau,s)P(s)x \, d\tau
\end{align} 
for all $(s,t) \in \Delta$ and $x \in X$, and this integral equation, by the Gronwall inequality, yields the desired estimate for $V_{0\,\eps}(t,s)P(s)$.
%
\smallskip

After these preparations we can now turn to the main part of the proof where the cases~(i) and~(ii) have to be treated 
separately. We first prove assertion~(i). 
%
%
As has already been shown in~\eqref{eq: intgl, lm 2 zum erw adsatz ohne sl}, 
\begin{align*}     
\big( V_{0\,\eps}(t)- U_{\eps}(t) \big) P(0)x &= U_{\eps}(t,s)V_{0\,\eps}(s)P(0)x \big|_{s=0}^{s=t} = \int_0^t U_{\eps}(t,s) \, P'(s) \, V_{0\,\eps}(s) P(0)x \,ds 
\end{align*}
so that, by rewriting the right hand side of this equation, we obtain 
\begin{align} \label{eq: gl 1, adsatz ohne sl}
\big( V_{0\,\eps}(t)- U_{\eps}(t) \big) P(0)x =& \int_0^t U_{\eps}(t,s)  \, (P'(s)-Q_n(s))P(s) \, V_{0\,\eps}(s)P(0)x \, ds \notag \\
&+ \int_0^t U_{\eps}(t,s)  \, [Q_n(s),P(s)] \, V_{0\,\eps}(s)P(0)x \, ds
\end{align}
for all $t \in I$, $\eps \in (0, \infty)$ and $x \in X$.
Since $Q_n(s)P(s) \longrightarrow P'(s)P(s)$ for every $s \in (0,1)$ by the strong convergence of $(Q_n(s))$ to $P'(s)$ for $s \in (0,1)$ and by $\rk P(s) = \rk P(0) < \infty$ for $s \in I$,
it follows by~\eqref{eq: absch V_{0 eps}} 
and by the dominated convergence theorem that
\begin{align} \label{eq: gl 2, adsatz ohne sl}
\sup_{\eps \in (0,\infty)} \sup_{t \in I} \norm{ \int_0^t U_{\eps}(t,s)  \, (P'(s)-Q_n(s))P(s) \, V_{0\,\eps}(s)P(0) \, ds  }  
\longrightarrow 0 
\end{align}
as $n \to \infty$.
In view of~\eqref{eq: gl 1, adsatz ohne sl} we therefore have to show that for each fixed $n \in \N$ 
\begin{align} \label{eq: zwbeh 1, adsatz ohne sl}
\sup_{t \in I} \norm{  \int_0^t U_{\eps}(t,s)  \, [Q_n(s),P(s)] \, V_{0\,\eps}(s)P(0) \, ds  } \longrightarrow 0 
\end{align}
as $\eps \searrow 0$.
So let $n \in \N$ be fixed for the rest of the proof.
%
%
Since $s \mapsto B_{n \, \bm{\delta}}(s)$ is in $W^{1,1}_*(I,L(X,Y))$ by the third preparatory step and since $[0,t] \ni s \mapsto U_{\eps}(t,s)|_Y \in L(Y,X)$ as well as $s \mapsto V_{0\,\eps}(s) \in L(X)$ are 
continuously differentiable w.r.t.~the respective strong operator topologies, Lemma~\ref{lm: prod- und inversenregel} yields that
\begin{align*}
[0,t] \ni s \mapsto U_{\eps}(t,s) B_{n \, \bm{\delta}}(s) V_{0\,\eps}(s) P(0)x 
\end{align*}
is the continuous representative of an element of $W^{1,1}([0,t],X)$ for every $x \in X$. With the help of the approximate commutator equation~\eqref{eq: appr commutator equation} of the second preparatory step, we therefore see that
\begin{gather}  
\int_0^t U_{\eps}(t,s)  \, [Q_n(s),P(s)] \, V_{0\,\eps}(s)P(0)x \, ds 
= \eps \, \int_0^t U_{\eps}(t,s)  \Big(\! -\frac 1 \eps A(s) B_{n \, \bm{\delta}}(s) \notag \\
+ \,\, B_{n \, \bm{\delta}}(s) \frac 1 \eps A(s) \Big) V_{0\,\eps}(s) P(0)x \, ds \,  + \int_0^t U_{\eps}(t,s) \, C_{n \, \bm{\delta}}^+(s) \, V_{0\,\eps}(s) P(0)x \, ds \notag \\
= \eps \, U_{\eps}(t,s) B_{n \, \bm{\delta}}(s) V_{0\,\eps}(s) P(0)x \Big|_{s=0}^{s=t} - \eps \, \int_0^t U_{\eps}(t,s)  \Big( B_{n \, \bm{\delta}}'(s) + B_{n \, \bm{\delta}}(s) [P'(s),P(s)] \Big) \notag \\
 V_{0\,\eps}(s) P(0)x \, ds + \int_0^t U_{\eps}(t,s) \, C_{n \, \bm{\delta}}^+(s) \, V_{0\,\eps}(s) P(0)x \, ds    \label{eq: gl 3, adsatz ohne sl}
\end{gather}
for all $t \in I$, $\eps \in (0,\infty)$, $x \in X$ and $\bm{\delta} \in (0,\delta_0]^{m_0}$. 
%
%
We now want to find functions $\eps \mapsto \delta_{1 \, \eps}, \dots, \delta_{m_0 \, \eps}$ defined on a small interval $(0,\delta_0']$ and converging to $0$ as $\eps \searrow 0$ in such a way that, if they are inserted in the right hand side of~\eqref{eq: gl 3, adsatz ohne sl}, the desired convergence~\eqref{eq: zwbeh 1, adsatz ohne sl} follows. 
In view of the estimates~\eqref{eq: absch B_n eps}, \eqref{eq: absch B_n eps'} and 
\begin{align}  \label{eq: absch C_n eps^+}
\int_0^1 \norm{ C_{n \, \bm{\delta}}^+(s) } \, ds \le \sum_{k=1}^{m_0} c \, \Big( \prod_{1 \le i < k} \delta_i \Big)^{-1} \, \int_0^1 \norm{ \delta_{k} \ol{R}_{\delta_{k}}(s) Q_n(s) P(s) } \, ds, 
\end{align} 
we would like the functions $\eps \mapsto \delta_{i \, \eps}$ to converge to $0$ so slowly that
\begin{gather}
\eps \, \Big( \prod_{i=1}^{k} \delta_{i \, \eps} \Big)^{-(m_0+1)} \longrightarrow 0 \quad (\eps \searrow 0)  \label{eq: wunsch 1 an eps_i} \\
\Big( \prod_{1 \le i < k} \delta_{i\,\eps} \Big)^{-1} \, \int_0^1 \norm{ \delta_{k\,\eps} \ol{R}_{\delta_{k\,\eps}}(s) Q_n(s) P(s) } \, ds \longrightarrow 0 \quad (\eps \searrow 0)   \label{eq: wunsch 2 an eps_i}
\end{gather}
for all $k \in \{ 1, \dots, m_0 \}$. 
Since 
\begin{align} \label{eq: wesentl grund fuer ex der eps_i}
\eta_{n}^+(\delta) := \int_0^1 \norm{\delta \ol{R}_{\delta}(s) Q_n(s) P(s) } \, ds \longrightarrow 0 \quad (\delta \searrow 0)
\end{align}
by Lemma~\ref{lm: lm 1 zum erw adsatz ohne sl}, by $\rk P(s) = \rk P(0) < \infty$ and by the dominated convergence theorem, 
such functions $\eps \mapsto \delta_{i\,\eps}$ really can be found. Indeed, define recursively 
\begin{gather*}
\delta_{m_0\,\eps} := \eps^{\frac{1}{(m_0+1)^2}}
\quad \text{and} \quad 
\delta_{m_0-l \,\eps} := \max \Big\{  \Big( \Big( \prod_{m_0-l+1 \le i < k} \delta_{i\,\eps} \Big)^{-1} \, \eta_{n}^+(\delta_{k\,\eps}) \Big)^{\frac{1}{2}}: \\
k \in \{ m_0-l+1, \dots, m_0 \}  \Big\} \cup \Big\{ \eps^{\frac{1}{(m_0+1)^2}} \Big\}
\end{gather*}
for $l \in \{ 1, \dots, m_0-1\}$. With the help of~\eqref{eq: wesentl grund fuer ex der eps_i} it then successively follows, by proceeding from larger to smaller indices $i$, that $\delta_{i\,\eps} \longrightarrow 0$ as $\eps \searrow 0$ for all $i \in \{1, \dots, m_0\}$ (so that, in particular, $\delta_{i\,\eps} \in (0,\delta_0]$ for small enough $\eps$ whence the expressions $\eta_{n}^+(\delta_{i\,\eps})$ used in the recursive definition 
make sense for small $\eps$ in the first place) and that~\eqref{eq: wunsch 1 an eps_i} and \eqref{eq: wunsch 2 an eps_i} are satisfied. Assertion~(i) 
now follows from~\eqref{eq: gl 1, adsatz ohne sl}, \eqref{eq: gl 2, adsatz ohne sl}, \eqref{eq: gl 3, adsatz ohne sl} by virtue of~\eqref{eq: absch B_n eps}, \eqref{eq: absch B_n eps'}, \eqref{eq: absch C_n eps^+} and~\eqref{eq: absch V_{0 eps}} 
\smallskip

We now prove assertion~(ii) and, for that purpose, additionally assume that $X$ is reflexive and $t \mapsto P(t)$ is norm continuously differentiable.
%
%
Analogously to~\eqref{eq: gl 1, adsatz ohne sl} we obtain 
\begin{align} \label{eq: gl 4, adsatz ohne sl}
\big( V_{\eps}(t)- U_{\eps}(t) \big) x =& \int_0^t U_{\eps}(t,s) \, [P'(s)-Q_n(s), P(s)] \, V_{\eps}(s)x \, ds \notag \\
&+ \int_0^t U_{\eps}(t,s) \, [Q_n(s),P(s)] \, V_{\eps}(s)x \, ds
\end{align}
for all $t \in I$, $\eps \in (0, \infty)$ and $x \in D(A(0)) = D$.
Since $Q_n(s) \longrightarrow P'(s)$ for every $s \in (0,1)$ by the additionally assumed norm continuous differentiability of $t \mapsto P(t)$,
it follows by Lemma~\ref{lm:stoersatz (M,omega)-stab} 
and by the dominated convergence theorem that
\begin{align} \label{eq: gl 5, adsatz ohne sl}
\sup_{\eps \in (0,\infty)}  \sup_{t \in I} \norm{ \int_0^t U_{\eps}(t,s) \, [P'(s)-Q_n(s), P(s)] \, V_{\eps}(s) \, ds  }  
\longrightarrow 0 
\end{align}
as $n \to \infty$.
In view of~\eqref{eq: gl 4, adsatz ohne sl} we therefore have to show that for each fixed $n \in \N$ 
\begin{align} \label{eq: zwbeh 2, adsatz ohne sl}
\sup_{t \in I} \norm{  \int_0^t U_{\eps}(t,s) \, [Q_n(s),P(s)] \, V_{\eps}(s) \, ds  } \longrightarrow 0 
\end{align}
as $\eps \searrow 0$.
So let $n \in \N$ be fixed for the rest of the proof.
%
%
Again completely analogously to the proof of~(i) it follows that
\begin{align*}
[0,t] \ni s \mapsto U_{\eps}(t,s) B_{n \, \bm{\delta}}(s) V_{\eps}(s)x 
\end{align*}
is the continuous representative of an element of $W^{1,1}([0,t],X)$ for every $x \in D(A(0)) = D$. With the help of the approximate commutator equation~\eqref{eq: appr commutator equation} of the second preparatory step, we therefore see that
\begin{gather}  
\int_0^t U_{\eps}(t,s) \, [Q_n(s),P(s)] \, V_{\eps}(s)x \, ds 
= \frac{1}{\eps} \, \int_0^t U_{\eps}(t,s)  \Big(\! -\frac 1 \eps A(s) B_{n \, \bm{\delta}}(s) \notag \\
+ \,\, B_{n \, \bm{\delta}}(s) \frac 1 \eps A(s) \Big) V_{\eps}(s) x \, ds \,  + \int_0^t U_{\eps}(t,s) \, C_{n \, \bm{\delta}}(s) \, V_{\eps}(s) x \, ds \notag \\
= \eps \, U_{\eps}(t,s) B_{n \, \bm{\delta}}(s) V_{\eps}(s) x \Big|_{s=0}^{s=t} - \eps \, \int_0^t U_{\eps}(t,s)  \Big( B_{n \, \bm{\delta}}'(s) + B_{n \, \bm{\delta}}(s) [P'(s),P(s)] \Big) \notag \\
 V_{\eps}(s) x \, ds + \int_0^t U_{\eps}(t,s) \, C_{n \, \bm{\delta}}(s) \, V_{\eps}(s) x \, ds    \label{eq: gl 6, adsatz ohne sl}
\end{gather}
for all $t \in I$, $\eps \in (0,\infty)$, $x \in D(A(0)) = D$ and $\bm{\delta} \in (0,\delta_0]^{m_0}$. 
%
%
In view of the estimates ~\eqref{eq: absch B_n eps}, \eqref{eq: absch B_n eps'}, \eqref{eq: absch C_n eps^+} and
\begin{align}  \label{eq: absch C_n eps^-}
\int_0^1 \norm{ C_{n \, \bm{\delta}}^-(s) } \, ds \le \sum_{k=1}^{m_0} c \, \Big( \prod_{1 \le i < k} \delta_i \Big)^{-1} \, \int_0^1 \norm{ P(s) Q_n(s) \delta_{k} \ol{R}_{\delta_{k}}(s) } \, ds,
\end{align}
we would now like to find functions $\eps \mapsto \delta_{1 \, \eps}, \dots, \delta_{m_0 \, \eps}$ defined on a small interval $(0,\delta_0']$ and converging to $0$ as $\eps \searrow 0$ so slowly that~\eqref{eq: wunsch 1 an eps_i}, \eqref{eq: wunsch 2 an eps_i} and
\begin{gather}
\Big( \prod_{1 \le i < k} \delta_{i\,\eps} \Big)^{-1} \, \int_0^1 \norm{ P(s) Q_n(s) \delta_{k\,\eps} \ol{R}_{\delta_{k\,\eps}}(s) } \, ds \longrightarrow 0 \quad (\eps \searrow 0)   \label{eq: wunsch 3 an eps_i}
\end{gather}
are satisfied for all $k \in \{ 1, \dots, m_0 \}$. 
Why is it possible to find such functions $\eps \mapsto \delta_{i\,\eps}$? In essence, this is 
because of~\eqref{eq: wesentl grund fuer ex der eps_i} and because
\begin{align} \label{eq: wesentl grund 2 fuer ex der eps_i}
\eta_{n}^-(\delta) := \int_0^1 \norm{P(s) Q_n(s) \delta \ol{R}_{\delta}(s) } \, ds \longrightarrow 0 \quad (\delta \searrow 0),
\end{align}
which last convergence can be seen as follows: by virtue of Proposition~\ref{prop: schwache assoziiertheit, dual}, which applies by the additionally assumed reflexivity of $X$, $P(s)^*$ is weakly associated of order $m_0$ with $A(s)^*$ and $\lambda(s)$ for almost every $s \in I$, and therefore Lemma~\ref{lm: lm 1 zum erw adsatz ohne sl} together with $\rk P(s)^* = \rk P(s) < \infty$ yields the convergence
\begin{align*}
\norm{ P(s) Q_n(s) \delta \ol{R}_{\delta}(s) } = \norm{ \delta \ol{R}_{\delta}(s)^* \, Q_n(s)^* P(s)^* }  \longrightarrow 0 \quad (\delta \searrow 0)
\end{align*}
for almost every $s \in I$, from which~\eqref{eq: wesentl grund 2 fuer ex der eps_i} follows by the dominated convergence theorem.
We now recursively define 
\begin{gather*}
\delta_{m_0\,\eps} := \eps^{\frac{1}{(m_0+1)^2}}
\quad \text{and} \quad 
\delta_{m_0-l \,\eps} := \max \Big\{  \Big( \Big( \prod_{m_0-l+1 \le i < k} \delta_{i\,\eps} \Big)^{-1} \, \eta_{n}^+(\delta_{k\,\eps}) \Big)^{\frac 1 2}, \\
\qquad \qquad \Big( \Big( \prod_{m_0-l+1 \le i < k} \delta_{i\,\eps} \Big)^{-1} \, \eta_{n}^-(\delta_{k\,\eps}) \Big)^{\frac 1 2}:  k \in \{ m_0-l+1, \dots, m_0 \}  \Big\} \cup \Big\{ \eps^{\frac{1}{(m_0+1)^2}} \Big\}
\end{gather*}
for $l \in \{ 1, \dots, m_0-1\}$. With the help of~\eqref{eq: wesentl grund fuer ex der eps_i} and~\eqref{eq: wesentl grund 2 fuer ex der eps_i} it then successively follows, by proceeding from larger to smaller indices $i$, that $\delta_{i\,\eps} \longrightarrow 0$ as $\eps \searrow 0$ for all $i \in \{1, \dots, m_0\}$ 
and that~\eqref{eq: wunsch 1 an eps_i}, \eqref{eq: wunsch 2 an eps_i} and~\eqref{eq: wunsch 3 an eps_i} are satisfied. Assertion~(ii) 
now follows from~\eqref{eq: gl 4, adsatz ohne sl}, \eqref{eq: gl 5, adsatz ohne sl}, \eqref{eq: gl 6, adsatz ohne sl} by virtue of~\eqref{eq: absch B_n eps}, \eqref{eq: absch B_n eps'}, \eqref{eq: absch C_n eps^+}, \eqref{eq: absch C_n eps^-} and Lemma~\ref{lm:stoersatz (M,omega)-stab}.
\end{proof}

Some remarks, which in particular clarify the relation of the above theorem with the adiabatic theorem without spectral gap condition from~\cite{AvronGraf12} and~\cite{dipl}, are in order.
\smallskip

1. Clearly, the adiabatic theorem above generalizes the adiabatic theorems without spectral gap condition from~\cite{AvronGraf12} (Theorem~11) and~\cite{dipl} (Theorem~6.4) 
which cover the case of general operators $A(t)$ and weakly semisimple eigenvalues $\lambda(t)$ under less general regularity conditions.  
%
%
%
In the special case where the eigenvalues $\lambda(t)$ from the above theorem lie on the imaginary axis $i \R$ for every $t \in I$, 
these eigenvalues are automatically weakly semisimple by the $(M,0)$-stability hypothesis of the theorem 
and by the weak associatedness hypothesis. 
(Argue as in the second remark at the beginning of Section~\ref{sect: bsp, adsaetze mit sl} to obtain that $P(t)$ is weakly associated of order $1$ with $A(t)$ and $\lambda(t)$ for almost every $t$.) 
And so, the above adiabatic theorem -- in the special case of purely imaginary eigenvalues -- essentially reduces to the adiabatic theorems without spectral gap condition from~\cite{AvronGraf12} and~\cite{dipl}. 
%
%
%
%
\smallskip

%

2. An inspection of the above proof shows that if the finite-rank hypothesis on $P(0)$ is the only one to be violated, then one still has the strong convergence 
\begin{align} \label{eq: starke konv, bem 2, adsatz ohne sl}
\sup_{t \in I} \norm{ \big( U_{\eps}(t) - V_{0\,\eps}(t) \big) P(0)x } \longrightarrow 0 \quad (\eps \searrow 0) \quad \text{for every } x \in X, 
\end{align}
provided that $\lambda(t)$ is even a weakly semisimple eigenvalue of $A(t)$ for almost every $t \in I$.
(In order to see this, notice that, under this extra 
condition, the inclusion $P(t)X \subset \ker(A(t)-\lambda(t))$ holds for every $t \in I$ 
by a closedness argument similar to the one in~\eqref{eq: P(t) vertauscht mit A(t) fuer alle t} and the $\eps$-dependence of $V_{0\,\eps}(s)P(0)$ is solely contained in a scalar factor, 
\begin{align*}
V_{0\,\eps}(s)P(0) = e^{\frac 1 \eps \int_0^s \lambda(\tau) \, d\tau} \, W(s)P(0) \quad (s \in I),
\end{align*}
where $W$ denotes the evolution system for $[P',P]$.) 
See~\cite{AvronGraf12} (Theorem~11).
\smallskip

3. As in the case with spectral gap, the adiabatic theorem without spectral gap condition above can 
be extended to several eigenvalues $\lambda_1(t)$, \dots, $\lambda_r(t)$. 
If $A$, $\lambda_j$, $P_j$ for all $j \in \{1, \dots, r\}$ satisfy the hypotheses of part~(ii) of the above adiabatic theorem and if for all $j \ne j'$ one has $\lambda_j \ne \lambda_{j'}$ almost everywhere, then the evolution system $V_{\eps}$ for $\frac 1 \eps A + K$ with $K$ as in~\eqref{eq: K mehrere sigma_j} is adiabatic w.r.t.~all the $P_j$ and well approximates the evolution system $U_{\eps}$ for $\frac 1 \eps A$ in the sense that
\begin{align} \label{eq: adsatz ohne sl fuer mehrere lambda_j}
\sup_{t \in I} \norm{ U_{\eps}(t) - V_{\eps}(t) } \longrightarrow 0 \quad (\eps \searrow 0),
\end{align}
provided $V_{\eps}$ exists on $D$. 
It seems that this version of the adiabatic theorem for several eigenvalues is new even in the special case of skew-adjoint operators $A(t)$. 
In order to prove this version of the theorem, set $B_{\bm{\delta}\,n}(t) := \frac 1 2 \sum_{j=1}^{r+1} B_{j \, \bm{\delta}\,n}(t)$ where
\begin{gather*}
B_{j \, \bm{\delta}\,n} := B_{j j \, \bm{\delta}\,n} \text{ for } j \in \{1,\dots,r\} \text{ and } B_{r+1 \, \bm{\delta}\,n} := \sum_{j,j'=1}^r B_{j j' \, \bm{\delta}\,n} \notag \\
B_{j j' \, \bm{\delta}\,n} := \sum_{k=0}^{m_j-1} \Big( \prod_{i=1}^{k+1} \ol{R}_{j \, \delta_i} \Big) Q_{j' n} (\lambda_j-A)^k P_j + \sum_{k=0}^{m_j-1} (\lambda_j-A)^k P_j Q_{j' n} \Big( \prod_{i=1}^{k+1} \ol{R}_{j \, \delta_i} \Big)  
\end{gather*}
with $\ol{R}_{j \, \delta}(t) := (\lambda_j(t) + \delta e^{i \vartheta_j(t)} - A(t))^{-1} (1-P_j(t))$ 
and $Q_{j' n}(t) := \int_0^1 J_{1/n}(t-r) P_{j'}'(r) \,dr$  and $m_j := \rk P_j(t)$.
It then follows as in the proof of the above theorem that the operators $B_{\bm{\delta}\,n}(t)$ satisfy the approximate commutator equation
\begin{align} \label{eq: comm eq, mehrere lambda_j}
B_{\bm{\delta}\,n}(t) A(t) - A(t) B_{\bm{\delta}\,n}(t) + C_{\bm{\delta}\,n}(t) \subset K_n(t)
\end{align}
for every $t \in I$, where the operators $K_n(t)$ on the right-hand side are given by 
\begin{align*}
K_n := \frac 1 2 \sum_{j=1}^r \ol{P}_jQ_{j n}P_j - P_jQ_{j n}\ol{P}_j + \frac 1 2 \sum_{j,j'=1}^r \ol{P}_jQ_{j' n}P_j - P_jQ_{j' n}\ol{P}_j
\end{align*}
and where the remainder terms $C_{\bm{\delta}\,n}(t)$ are given by $C_{\bm{\delta}\,n}(t) := \frac 1 2 \sum_{j=1}^{r+1} C_{j \, \bm{\delta}\,n}(t)$ with
\begin{gather*}
C_{j \, \bm{\delta}\,n} := C_{j j \, \bm{\delta}\,n} \text{ for } j \in \{1,\dots,r\} \text{ and } C_{r+1 \, \bm{\delta}\,n} := \sum_{j,j'=1}^r C_{j j' \, \bm{\delta}\,n}  \\
C_{j j' \, \bm{\delta}\,n} := \sum_{k=0}^{m_j-1} \delta_{k+1} e^{i\vartheta_j} \Big( \prod_{i=1}^{k+1} \ol{R}_{j \, \delta_i} \Big) Q_{j' n} P_j (\lambda_j-A)^k \\
- \sum_{k=0}^{m_j-1} (\lambda_j-A)^k P_j Q_{j' n} \delta_{k+1} e^{i\vartheta_j} \Big( \prod_{i=1}^{k+1} \ol{R}_{j \, \delta_i} \Big).  
\end{gather*}
It also follows that
\begin{align} \label{eq: K_n to K, mehrere lambda_j}
K_n(t) \longrightarrow K(t) \quad (n \to \infty) 
\end{align}
for all $t \in (0,1)$, because $P_j(t)P_{j'}'(t)P_j(t) = 0$ for $j, j' \in \{1,\dots,r\}$ and all $t \in I$ (for $j = j'$ recall~\eqref{eq: PP'P=0} and for $j \ne j'$ use $P_{j'}' = P_{j'}' P_{j'} + P_{j'} P_{j'}'$ and the third remark after Theorem~\ref{thm: typ mögl für PX und (1-P)X}) 
and because $[P_{r+1}', P_{r+1}] = [(1-P_{r+1})',1-P_{r+1}]$.
With~\eqref{eq: comm eq, mehrere lambda_j} and~\eqref{eq: K_n to K, mehrere lambda_j} at hand, the assertion~\eqref{eq: adsatz ohne sl fuer mehrere lambda_j} can be proved in the same way as part~(ii) of the above adiabatic theorem.
\smallskip

We close this section with a corollary tailored to the special situation of spectral operators. In this situation there are relatively simple and convenient criteria for the assumptions -- in particular, the reduced resolvent estimate -- of the above adiabatic theorem to be satisfied. 

\begin{cor} \label{cor: adsatz fuer spektrale A(t)}
Suppose $A(t): D \subset X \to X$ for every $t \in I$ is a spectral operator with spectral measure $P^{A(t)}$ such that Condition~\ref{cond: reg 1} is satisfied with $\omega = 0$ and such that $\sup_{t \in I} \sup_{E \in \mathcal{B}_{\C}} \norm{ P^{A(t)}(E) } < \infty$. 
Suppose further that $\lambda(t)$ for every $t \in I$ is an eigenvalue of $A(t)$ 
such that the open sector
\begin{align*}
\lambda(t) + \delta_0 \, S_{(\vartheta(t)-\vartheta_0, \vartheta(t) + \vartheta_0)} := \big\{ \lambda(t) + \delta e^{i \vartheta}: \delta \in (0,\delta_0), \vartheta \in (\vartheta(t)-\vartheta_0, \vartheta(t) + \vartheta_0) \big\}
\end{align*} 
of radius $\delta_0 \in (0,\infty)$ and angle $2 \vartheta_0 \in (0, \pi)$ for every $t \in I$ is contained in 
$\rho(A(t))$ 
and such that $\rk P^{A(t)}(\{\lambda(t)\}) < \infty$ for almost every $t \in I$ and $t \mapsto \lambda(t)$, $e^{i \vartheta(t)}$ are absolutely continuous. 
Suppose finally that 
$A(t)|_{P^{A(t)}(\sigma(t))D}$ 
for every $t \in I$ is of scalar type for some punctured neighborhood 
\begin{align*}
\sigma(t) := \sigma(A(t)) \cap \ol{B}_{r_0}(\lambda(t)) \setminus \{\lambda(t)\}
\end{align*}
of $\lambda(t)$ in $\sigma(A(t))$ of radius $r_0 \in (0,\infty) \cup \{\infty\}$ 
and that $t \mapsto P^{A(t)}(\{\lambda(t)\})$ coincides almost everywhere with a strongly continuously differentiable map $t \mapsto P(t)$ and 
$t \mapsto P^{A(t)}(\tau(t))$ is continuous, where $\tau(t) := \sigma(A(t)) \setminus (\sigma(t) \cup \{\lambda(t)\})$. 
Then the conclusions~(i) and~(ii) of the preceding adiabatic theorem hold true.
\end{cor}

\begin{proof}
We first observe that $P^{A(t)}(\{\lambda(t)\})$ is weakly associated with $A(t)$ and $\lambda(t)$ for every $t \in I$ where $\rk P^{A(t)}(\{\lambda(t)\}) < \infty$ by Proposition~\ref{prop: krit ex schw assoz proj, A spektral} and therefore $P(t)$ is weakly associated with $A(t)$ and $\lambda(t)$ for almost every $t \in I$. Also, $\rk P(0) = \rk P(t) = \rk P^{A(t)}(\{\lambda(t)\}) < \infty$ for almost every $t \in I$ by the continuity of $t \mapsto P(t)$.
%
We now verify the (reduced) resolvent estimate~\eqref{eq: resolvent estimate} from the theorem above by showing that 
\begin{align}
&\norm{ \big( \lambda(t) + \delta e^{i \vartheta(t)} - A(t)\big)^{-1} P^{A(t)}(\sigma(t)) } \le \frac{M_{0\,1}}{\delta}  \label{eq: absch 1, spektr} \\
&\qquad \qquad \qquad \norm{ \big( \lambda(t) + \delta e^{i \vartheta(t)} - A(t)\big)^{-1} P^{A(t)}(\tau(t)) } \le M_{0\,2}  \label{eq: absch 2, spektr}
\end{align}
for every $t \in I$ and $\delta \in (0,\delta_0']$.
Without loss of generality we may assume that $\lambda(t) \ne 0$ for all $t \in I$ (because otherwise we can choose $c \in i\R$ such that $\lambda(t) + c \ne 0$ for all $t$ and consider the shifted data $A_c(t) := A(t) + c$, $\lambda_c(t) := \lambda(t) + c$ and $P_c(t) := P(t)$).
%
In order to see~\eqref{eq: absch 1, spektr} notice that 
\begin{align*}
\big( \lambda(t) + \delta e^{i \vartheta(t)} - A(t)\big)^{-1} P^{A(t)}(\sigma(t)) = \big( \lambda(t) + \delta e^{i \vartheta(t)} - A_{\sigma}(t)\big)^{-1} P_{\sigma}(t) 
\end{align*}
where $A_{\sigma}(t) := A(t)|_{P^{A(t)}(\sigma(t))D}$ and $P_{\sigma}(t) := P^{A(t)}(\sigma(t))$, and that, by the scalar-type spectrality of $A_{\sigma}(t)$ and Theorem~XVIII.2.11 of~\cite{DunfordSchwartz},
\begin{align} \label{eq: zwbeh 2.1}
&\big| \langle x^*, \big( \lambda(t) + \delta e^{i \vartheta(t)} - A_{\sigma}(t)\big)^{-1} P_{\sigma}(t)x \rangle \big| 
\le \int_{\sigma(A_{\sigma}(t))} \frac{1}{ | \lambda(t) + \delta e^{i \vartheta(t)} - z | } \,\, d \big| P^{A_{\sigma}(t)}_{x^*,P_{\sigma}(t)x} \big|(z) \notag \\
&\qquad \qquad \qquad \le \frac{1}{\dist(\lambda(t)+\delta e^{i \vartheta(t)}, \sigma(A(t))) } \, \big| P^{A_{\sigma}(t)}_{x^*,P_{\sigma}(t)x} \big|(\C)
\end{align}
where 
$| P^{A_{\sigma}(t)}_{y^*,y} |$ denotes the total variation of the complex measure $E \mapsto P^{A_{\sigma}(t)}_{y^*,y}(E) := \scprd{y^*, P^{A_{\sigma}(t)}(E)y}$ for $y \in P_{\sigma}(t)X$, $y^* \in (P_{\sigma}(t)X)^*$. 
Since, by $P^{A_{\sigma}(t)}(E) = P^{A(t)}(E)|_{P_{\sigma}(t)X}$ and Lemma~III.1.5 of~\cite{DunfordSchwartz}, 
\begin{align*}
\big| P^{A_{\sigma}(t)}_{x^*,P_{\sigma}(t)x} \big|(\C) 
\le 4 \sup_{E \in \mathcal{B}_{\C}} \big| \langle x^*, P^{A(t)}(E \cap \sigma(t))x \rangle \big| \le 4 M' \norm{x^*} \norm{x}
\end{align*}
for every $t \in I$ (where $M' := \sup_{t \in I} \sup_{E \in \mathcal{B}_{\C}} \norm{ P^{A(t)}(E) } < \infty$) and since, by the sector condition, 
\begin{align*}
\dist\big( \lambda(t)+\delta e^{i \vartheta(t)}, \sigma(A(t)) \big) \ge (\sin \vartheta_0) \, \delta
\end{align*}
for every $t \in I$ and $\delta \in (0,\delta_0']$ (where $\delta_0'$ is chosen small enough),
the desired estimate~\eqref{eq: absch 1, spektr} follows from~\eqref{eq: zwbeh 2.1}.
%
In order to see~\eqref{eq: absch 2, spektr} notice that, by 
$\lambda(t) \ne 0$ for $t \in I$,
\begin{align*}
\lambda(t) + \delta e^{i \vartheta(t)} \notin \sigma(\tilde{A}_{\tau}(t)) \subset \ol{\tau(t)} \cup \{0\} \subset \C \setminus \ol{B}_{r_0}(\lambda(t)) \cup \{0\}
\end{align*}
for every $t \in I$ and $\delta \in [0,\delta_0']$ (where $\delta_0'$ is chosen small enough), 
and that 
\begin{align*}
\big( \lambda(t) + \delta e^{i \vartheta(t)} - A(t)\big)^{-1} P^{A(t)}(\tau(t)) = \big( \lambda(t) + \delta e^{i \vartheta(t)} - \tilde{A}_{\tau}(t)\big)^{-1} P_{\tau}(t), 
\end{align*}
where $\tilde{A}_{\tau}(t) := A(t) P^{A(t)}(\tau(t))$ and $P_{\tau}(t) := P^{A(t)}(\tau(t))$. (Also notice that in the case $r_0 = \infty$ there is nothing to show because then $\tau(t) = \emptyset$ for every $t \in I$.) 
We now show that $t \mapsto \tilde{A}_{\tau}(t)$ is continuous in the generalized sense. Since, for every fixed $z \in \C$ with $\Re z > 0$,
\begin{align} \label{eq: zwbeh 2.2.1}
(z-\tilde{A}_{\tau}(t))^{-1} &= \big( z-A(t)P_{\tau}(t) \big)^{-1} P_{\tau}(t) + \big( z-A(t)P_{\tau}(t) \big)^{-1} (1-P_{\tau}(t)) \notag \\
&= ( z-A(t) )^{-1} P_{\tau}(t) + \frac 1 z (1-P_{\tau}(t))
\end{align}
and since $(1-P_{\tau}(t))X = P^{A(t)}(\sigma(t) \cup \{\lambda(t)\})X \subset D(A(t)) = D$ by the boundedness of $\sigma(t) \cup \{\lambda(t)\} = \ol{B}_{r_0}(\lambda(t)) \cap \sigma(A(t))$, 
we obtain $(z-\tilde{A}_{\tau}(t_0))^{-1} X \subset D \subset D(\tilde{A}_{\tau}(t))$ and therefore
\begin{align} \label{eq: zwbeh 2.2.2}
(z-\tilde{A}_{\tau}(t))^{-1} - (z-\tilde{A}_{\tau}(t_0))^{-1} = (z-\tilde{A}_{\tau}(t))^{-1} \big( \tilde{A}_{\tau}(t) - \tilde{A}_{\tau}(t_0) \big) (z-\tilde{A}_{\tau}(t_0))^{-1}
\end{align}
for every $t, t_0 \in I$. Since 
\begin{align*}
\tilde{A}_{\tau}(t) (z-\tilde{A}_{\tau}(t_0))^{-1} = P_{\tau}(t) A(t)(z-\tilde{A}_{\tau}(t_0))^{-1} & \longrightarrow  P_{\tau}(t_0) A(t_0)(z-\tilde{A}_{\tau}(t_0))^{-1} \\
&= \tilde{A}_{\tau}(t_0) (z-\tilde{A}_{\tau}(t_0))^{-1} \quad (t \to t_0)
\end{align*}
by the assumed continuity of $t \mapsto P_{\tau}(t)$ and the $W^{1,1}_*$-regularity of $t \mapsto A(t)$, and since $\sup_{t \in I} \| (z-\tilde{A}_{\tau}(t))^{-1} \| < \infty$ 
by~\eqref{eq: zwbeh 2.2.1} and the $(M,0)$-stability of $A$, it follows from~\eqref{eq: zwbeh 2.2.2} that $t \mapsto (z-\tilde{A}_{\tau}(t))^{-1}$ is continuous 
and therefore $t \mapsto \tilde{A}_{\tau}(t)$ is continuous in the generalized sense (Theorem~IV.2.25 of~\cite{KatoPerturbation80}). In particular, $I \times [0,\delta_0'] \ni (t,\delta) \mapsto \big( \lambda(t) + \delta e^{i \vartheta(t)} - \tilde{A}_{\tau}(t)\big)^{-1}$ is continuous by Theorem~IV.3.15 of~\cite{KatoPerturbation80}, hence bounded, and the desired estimate~\eqref{eq: absch 2, spektr} follows.
Combining now~\eqref{eq: absch 1, spektr} and~\eqref{eq: absch 2, spektr} we obtain the desired resolvent estimate~\eqref{eq: resolvent estimate} 
because $1-P(t) = 1-P^{A(t)}(\{\lambda(t)\}) = P^{A(t)}(\sigma(t)) + P^{A(t)}(\tau(t))$ for almost every $t \in I$ and because the left-hand side of~\eqref{eq: resolvent estimate} 
is continuous in $t$.
\end{proof}

\subsection{A quantitative adiabatic theorem without spectral gap condition} \label{sect: quant adsatz ohne sl}

As a supplement to the qualitative adiabatic theorem above (Theorem~\ref{thm: erw adsatz ohne sl}), we note the following quantitative refinement. 
It implies 
that, if in the situation of the above theorem 
the maps $t \mapsto A(t), \lambda(t), e^{i \vartheta(t)}$ and $t \mapsto P(t)$ are even $W^{1,\infty}_*$- or $W^{2,\infty}_*$-regular respectively, then the rate of convergence  (Lemma~\ref{lm: lm 1 zum erw adsatz ohne sl}!) 
of the integrals
\begin{align} \label{eq: eta_0}
&\eta^+(\delta) := \int_0^1 \norm{\delta \big( \lambda(s)+\delta e^{i \vartheta(s)} - A(s) \big)^{-1} P'(s)P(s)} \, ds, \notag \\
&\qquad \qquad \qquad \eta^-(\delta) := \int_0^1 \norm{P(s)P'(s) \delta \big( \lambda(s)+\delta e^{i \vartheta(s)} - A(s) \big)^{-1} } \, ds 
\end{align}
yields a simple upper bound on the rate of convergence of $\sup_{t \in I} \norm{ U_{\eps}(t)-V_{\eps}(t) }$ which we are interested in here. 
%
See~\cite{Teufel01} for an analogous result in the case of skew-adjoint operators $A(t)$.

%

\begin{thm} \label{thm: erw adsatz ohne sl, quantitativ}
Suppose that 
$A(t)$, $\lambda(t)$, $P(t)$ are as in Theorem~\ref{thm: erw adsatz ohne sl} with $X$ not necessarily reflexive and that $t \mapsto A(t)$ is even in $W^{1,\infty}_*(I,L(Y,X))$, $t \mapsto \lambda(t), e^{i \vartheta(t)}$ are even Lipschitz and $t \mapsto P(t)$ is even in $W^{2,\infty}_*(I,L(X))$. Suppose further that $\eta: (0,\delta_0] \subset (0,1] \to (0,\infty)$ is a function such that $\eta(\delta) \longrightarrow 0$ as $\delta \searrow 0$ and 
\begin{align*}
\eta(\delta) \ge \delta \quad \text{as well as} \quad \eta^{\pm}(\delta) \le \eta(\delta)
\end{align*} 
for all $\delta \in (0,\delta_0]$ with $\eta^{\pm}$ as above. 
Then there is a constant $c$ such that
\begin{align*}
\sup_{t \in I} \norm{ U_{\eps}(t)-V_{\eps}(t) } \le c \, \tilde{\eta}^{m_0} \big( \eps^{2/(m_0(m_0+1))} \big) = c \, ( \tilde{\eta} \circ \dotsb \circ \tilde{\eta} ) \big( \eps^{2/(m_0(m_0+1))} \big)
\end{align*}
for $\eps$ sufficiently small, where $\tilde{\eta}(\delta) := \eta( \delta^{\frac{1}{2}} )$.
\end{thm}

\begin{proof}
We proceed as in the proof of the qualitative adiabatic theorem above, but now replace $Q_n$ and $Q_n'$ at any occurrence by $P'$ and $P''$. We can then conclude from~\eqref{eq: gl 4, adsatz ohne sl} and~\eqref{eq: gl 6, adsatz ohne sl} (with the replacements just mentioned) 
that there is a constant $c'$ such that
\begin{align}  \label{eq: absch interessierender ausdruck}
\sup_{t \in I} \norm{ U_{\eps}(t)-V_{\eps}(t) }  \le  \,\, c' \, \bigg(    \sum_{k=1}^{m_0} \eps \Big( \prod_{j=1}^{k} \delta_j \Big)^{-1} 
&+ \sum_{k=1}^{m_0} \eps \Big( \delta_k^{m_0+1} \, \prod_{j \ne k} \delta_j \Big)^{-1} \eta(\delta_k) \notag \\
&\qquad \qquad + \sum_{k=1}^{m_0} \Big( \prod_{1 \le j < k} \delta_j \Big)^{-1} \eta(\delta_k) \bigg)
%
\end{align}
for all $\delta_1, \dots, \delta_{m_0} \in (0,\delta_0]$ and $\eps \in (0,\infty)$. In this estimate the first, second, and third sum correspond to the $B_{\bm{\delta}}$-, $B_{\bm{\delta}}'$-, $C_{\bm{\delta}}$-terms in~\eqref{eq: gl 6, adsatz ohne sl}, respectively. See~\eqref{eq: absch B_n eps} and~\eqref{eq: absch C_n eps^+}, \eqref{eq: absch C_n eps^-} for the estimation of the $B_{\bm{\delta}}$-terms and $C_{\bm{\delta}}$-terms. In order obtain the upper bound for the $B_{\bm{\delta}}'$-terms, refine the estimate~\eqref{eq: absch B_n eps'} on $\int_0^1 \norm{ B_{\bm{\delta}}'(s)} \, ds$ 
from the proof of the previous theorem by using the fact 
that
\begin{align} \label{eq: 2, quant adsatz ohne sl}
\esssup_{s \in I} \norm{ \big( A'(s)-\lambda'(s) - \delta \, r'(s)\big) (A(s)-1)^{-1} } 
\le c < \infty,
\end{align} 
where the additional assumption that $t \mapsto A(t)$ and $t \mapsto \lambda(t), r(t):= e^{i \vartheta(t)}$ be even $W^{1,\infty}_*$-regular enters. 
It follows from this that the integral (from $0$ to $1$) 
of the critical terms in $B_{\bm{\delta}}'$, namely
\begin{gather}
\ol{R}_{\delta_1}(s) \dotsb R_{\delta_l}(s) \big( A'-\lambda'-\delta_l \, r' \big)(s) \ol{R}_{\delta_l}(s) \dotsb \ol{R}_{\delta_k}(s) P'(s) P(s), \label{eq: 3, quant adsatz ohne sl} \\
P(s) P'(s) \, \ol{R}_{\delta_1}(s) \dotsb \ol{R}_{\delta_l}(s) (A'-\lambda'-\delta_l \, r')(s) R_{\delta_l}(s) \dotsb \ol{R}_{\delta_k}(s), \label{eq: 4, quant adsatz ohne sl}
\end{gather}
can be estimated by $( \delta_l^{m_0+1} \, \prod_{j \ne l} \delta_j )^{-1} \eta(\delta_l)$ for all $l \in \{1, \dots, k\}$, as desired. 
(In order to see this, insert in both of the above products~\eqref{eq: 3, quant adsatz ohne sl} and~\eqref{eq: 4, quant adsatz ohne sl} the identity operators $(A(s)-1)^{-1}(A(s)-1)$ behind $(A'-\lambda'-\delta_l \, r')(s)$, commute $(A(s)-1) \ol{R}_{\delta_l}(s)$ directly in front of $P'(s)P(s)$ and $\ol{R}_{\delta_l}(s)$ directly behind $P(s) P'(s)$ respectively, and then use~\eqref{eq: 2, quant adsatz ohne sl} together with the fact that 
\begin{gather*}
\int_0^1 \norm{(A(s)-1) \ol{R}_{\delta}(s) P'(s)P(s)} \, ds, \,\, \int_0^1 \norm{P(s) P'(s) \ol{R}_{\delta}(s)} \, ds \le c \, \frac{\eta(\delta)}{\delta} 
\end{gather*}
and that $\sup_{s \in I} \norm{R_{\delta}(s)}, \sup_{s \in I} \norm{(A(s)-1)R_{\delta}(s)} \le \frac{c}{\delta^{m_0}}$ for all sufficiently small $\delta \in (0,\delta_0]$.)
We now recursively define
\begin{align*}
\delta_{m_0 \, \eps} := \eps^{\frac{1}{m_0(m_0+1)}} 
\quad \text{and} \quad 
\delta_{m_0-k \, \eps} := \big( \eta(\delta_{m_0-k+1 \, \eps}) \big)^{\frac{1}{2}}
\end{align*}
for $\eps$ so small that $\delta_{m_0-k+1 \, \eps}$ lies in $(0,\delta_0]$ and for $k \in \{1, \dots, m_0-1\}$. (It should be noticed 
that $\delta_{m_0-k+1 \, \eps} \longrightarrow 0$ as $\eps \searrow 0$ because $\eta(\delta) \longrightarrow 0$ and that $\delta_{m_0-k+1 \, \eps}$ therefore really lies in the domain $(0,\delta_0]$ of $\eta$ for sufficiently small $\eps$.) 
Since 
$\eta(\delta_{1 \, \eps}) = \tilde{\eta}^{m_0} ( \eps^{2/(m_0(m_0+1))} )$ and $\frac{1}{\delta_{k-1 \, \eps}} \eta(\delta_{k \, \eps}) = \delta_{k-1 \, \eps} \le \eta(\delta_{k-1 \, \eps})$ for $k \in \{2, \dots, m_0\}$,
it follows by induction that
\begin{align} \label{eq: absch 1}
\Big( \prod_{1 \le j < k} \delta_{j \, \eps} \Big)^{-1} \eta(\delta_{k \, \eps}) \le \tilde{\eta}^{m_0} \big( \eps^{2/(m_0(m_0+1))} \big)
\end{align}
and, in particular, $\eta(\delta_{k \, \eps}) \le \tilde{\eta}^{m_0} \big( \eps^{2/(m_0(m_0+1))} \big)$ for all $k \in \{1, \dots, m_0\}$ and sufficiently small $\eps$. Since 
$\delta_{m_0 \, \eps} \le \delta_{m_0-k+1 \, \eps} \le \delta_{m_0-k \, \eps}$ for $k \in \{1, \dots, m_0-1\}$ and small $\eps$, 
it further follows that
\begin{align} \label{eq: absch 2}
\eps \Big( \prod_{j=1}^{k} \delta_{j \, \eps} \Big)^{-1} \le \eps \Big( \delta_{k \, \eps}^{m_0+1}\prod_{j \ne k} \delta_{j \, \eps} \Big)^{-1} \eta(\delta_{k \, \eps}) 
&\le \eps \Big( \prod_{j=1}^{m_0} \delta_{m_0 \, \eps} \Big)^{-(m_0+1)} \tilde{\eta}^{m_0} \big( \eps^{2/(m_0(m_0+1))} \big) \notag \\
&= \tilde{\eta}^{m_0} \big( \eps^{2/(m_0(m_0+1))} \big)
\end{align}
for all $k \in \{1, \dots, m_0\}$ and sufficiently small $\eps$. 
Combining~\eqref{eq: absch interessierender ausdruck}, \eqref{eq: absch 1} and~\eqref{eq: absch 2} we finally obtain the assertion.
\end{proof}

We now specialize to the case of spectral operators $A(t)$ of scalar type 
and note the following quantitative adiabatic theorem tailored to scalar-type spectral operators $A(t)$ 
whose spectral measures $P^{A(t)}$ are Hölder continuous in $t$ around $\lambda(t)$ in some sense (which, in particular, means that in a punctured neighborhood of $\lambda(t)$ there is no more 
eigenvalue of $A(t)$). 
It generalizes a result for skew-adjoint $A(t)$ of Avron and Elgart (Corollary~1 in~\cite{AvronElgart99}) and a refinement of it due to Teufel (Remark~1 in~\cite{Teufel01}) and 
improves the rates of convergence given there.

\begin{cor} \label{prop: quant adsatz für hoelderstet spektrmass}
Suppose $A(t): D \subset X \to X$ for every $t \in I$ is a spectral operator of scalar type (with spectral measure $P^{A(t)}$) such that Condition~\ref{cond: reg 1} is satisfied with $\omega = 0$ and such that $\sup_{t \in I} \sup_{E \in \mathcal{B}_{\C}} \norm{ P^{A(t)}(E) } < \infty$. 
Suppose further that $\lambda(t)$ for every $t \in I$ is an eigenvalue of $A(t)$ 
such that the open sector
\begin{align*}
\lambda(t) + \delta_0 \, S_{(\vartheta(t)-\vartheta_0, \vartheta(t) + \vartheta_0)} := \big\{ \lambda(t) + \delta e^{i \vartheta}: \delta \in (0,\delta_0), \vartheta \in (\vartheta(t)-\vartheta_0, \vartheta(t) + \vartheta_0) \big\}
\end{align*} 
of radius $\delta_0 \in (0,\infty)$ and angle $2 \vartheta_0 \in (0, \pi)$ for every $t \in I$ is contained in 
$\rho(A(t))$ 
and such that $t \mapsto \lambda(t)$, $e^{i \vartheta(t)}$ are absolutely continuous. 
Suppose finally that $P(t)$ for every $t \in I$ is a bounded projection in $X$ such that $P(t) = P^{A(t)}(\{ \lambda(t) \})$ for almost every $t \in I$ 
and $t \mapsto P(t)$ is in $W^{2,1}_*(I,L(X))$, and suppose that 
$P^{A(t)}$ is H\"older continuous locally around $\lambda(t)$ with exponent $\alpha \in (0,1]$ uniformly in $t \in I$ in the following sense: 
there is a constant $c_0 \in (0, \infty)$ such that 
\begin{align*}
\big\| P^{A(t)}(E)x \big\| \le c_0 \, \lambda(E)^{\frac{\alpha}{2}} \, \norm{x}
\end{align*}
for all $x \in X$ and for all $t \in I$ and $E \in \mathcal{B}_{\C}$ that are contained in the punctured neighborhood $\dot{B}_{r_0}(\lambda(t)) := B_{r_0}(\lambda(t)) \setminus \{\lambda(t)\}$ of $\lambda(t)$ (with $r_0$ independent of $t$).
Then there exists a constant 
$c \in (0,\infty)$ such that
\begin{align*}
\sup_{t \in I} \norm{ U_{\eps}(t) - V_{\eps}(t) } \le c \; \eps^{ \frac{\alpha}{2(1+\alpha)} }
\end{align*}
for small enough $\eps \in (0, \infty)$, where $V_{\eps}$ denotes the evolution system for $\frac{1}{\eps} A + [P',P]$.
\end{cor}

\begin{proof}
We first show that there exists a function $\eta: (0, \delta_0'] \to (0,\infty)$ such that $\eta(\delta) \longrightarrow 0$ as $\delta \searrow 0$ and 
\begin{align} \label{eq: gl 1, quant für normal}
\eta(\delta) \ge \delta 
\quad \text{and} \quad 
\norm{\delta \ol{R}_{\delta}(t)} = \norm{ \delta \big( \lambda(t)+\delta e^{i \vartheta(t)} - A(t) \big)^{-1} (1-P(t)) } \le \eta(\delta)
\end{align}
for all $\delta \in (0,\delta_0']$ and $t \in I$ (with a suitable $\delta_0'$). In fact, it is sufficient to prove~\eqref{eq: gl 1, quant für normal} for all $t$ in the set $I \setminus N$ of those $t$ where $P(t) = P^{A(t)}(\{\lambda(t)\})$, because this set $I \setminus N$ is dense in $I$ by assumption and because the left-hand side of the second inequality in~\eqref{eq: gl 1, quant für normal} is continuous in $t$ by assumption.
We 
observe that for every $t \in I \setminus N$ 
\begin{align*}
\big| \scprd{x^*, \delta \ol{R}_{\delta}(t)x } \big|
\le \int_{\sigma(A(t)) \setminus \{\lambda(t)\} } \frac{\delta}{ | \lambda(t) + \delta e^{i \vartheta(t)} - z | } \,\, d \big| P^{A(t)}_{x^*,x} \big|(z),
\end{align*}
where $\big| P^{A(t)}_{x^*,x} \big|$ denotes the total variation of 
$E \mapsto P^{A(t)}_{x^*,x}(E) := \scprd{x^*, P^{A(t)}(E)x}$
(use the scalar-type spectrality of $A(t)$ and Theorem~XVIII.2.11 of~\cite{DunfordSchwartz}).
We then divide the punctured spectrum $\sigma(A(t)) \setminus \{\lambda(t)\}$ 
into the parts
\begin{align*}
\sigma_{1 \, r_{\delta}}(t) := \sigma(A(t)) \cap B_{r_{\delta}}(\lambda(t)) \setminus \{\lambda(t)\} \quad \text{and} \quad \sigma_{2 \, r_{\delta}}(t) := \sigma(A(t)) \cap \C \setminus B_{r_{\delta}}(\lambda(t))
\end{align*} 
of those spectral values that are close to $\lambda(t)$ respectively far from $\lambda(t)$,
where $r_{\delta} := \delta^{\gamma}$ and $\gamma \in (0,1)$ will be chosen in~\eqref{eq: gl 4, quant für normal} below. 
Since, by Lemma~III.1.5 of~\cite{DunfordSchwartz}, 
\begin{align*}
\big| P^{A(t)}_{x^*,x} \big|(E) \le 4 \sup_{F \in \mathcal{B}_E} \big| \langle x^*, P^{A(t)}(F) P^{A(t)}(E)x \rangle \big|  
\le 4 M' \norm{x^*} \big\| P^{A(t)}(E)x \big\|
\end{align*}
for every $t \in I$ and $E \in \mathcal{B}_{\C}$ (where $M' := \sup_{t \in I} \sup_{F \in \mathcal{B}_{\C}} \norm{ P^{A(t)}(F) } < \infty$) and since, by the assumed sector condition,
\begin{align*}
\operatorname{dist}\big( \lambda(t)+\delta e^{i \vartheta(t)}, \sigma(A(t)) \big) \ge (\sin \vartheta_0) \, \delta
\end{align*}
for every $t \in I$ and $\delta \in (0,\delta_0']$ (where $\delta_0'$ is chosen small enough), there are positive constants $c_1$, $c_2$ such that
\begin{gather*}
\int_{\sigma_{1 \, r_{\delta}}(t)} \frac{\delta}{ | \lambda(t) + \delta e^{i \vartheta(t)} - z | } \,\, d \big|P^{A(t)}_{x^*,x} \big|(z) 
\le \frac{1}{\sin \vartheta_0} \, \big| P^{A(t)}_{x^*,x} \big| \big( \dot{B}_{r_{\delta}}(\lambda(t))  \big)
\le c_1 \delta^{\alpha \, \gamma} \norm{x^*} \norm{x} \\ 
\text{as well as} \\
\int_{\sigma_{2 \, r_{\delta}}(t)} \frac{\delta}{ | \lambda(t) + \delta e^{i \vartheta(t)} - z | } \,\, d \big|P^{A(t)}_{x^*,x}\big| (z)
\le \frac{\delta}{ r_{\delta}-\delta } \, \big| P^{A(t)}_{x^*,x} \big|(\C)  
\le c_2 \delta^{1-\gamma} \norm{x^*} \norm{x}  
\end{gather*}
for every $x \in X$, $x^* \in X^*$, $\delta \in (0,\delta_0']$ and $t \in I$. 
Consequently,
\begin{align}
\norm{\delta \ol{R}_{\delta}(t) } 
\le c_1 \, \delta^{ \alpha \, \gamma} + c_2 \, \delta^{ 1-\gamma } 
\le \max \{ c_1, c_2 \} \, \delta^{ \min \{ \alpha \, \gamma ,  1-\gamma \} } = c_0' \, \delta^{ \beta(\gamma) } 
\end{align}
for every $t \in I \setminus N$ and $\delta \in (0, \delta_0']$ (notice that $\beta(\gamma) := \min \{ \alpha \, \gamma, 1-\gamma \}$, for given $\gamma$, is the best -- that is, biggest -- 
possible exponent in the second inequality above). 
And as $\gamma \mapsto \beta(\gamma)$ is maximal at $\gamma_0 := \frac{1}{1+\alpha}$, we choose 
\begin{align} \label{eq: gl 4, quant für normal}
\gamma := \gamma_0, \quad \beta := \beta(\gamma_0) = \frac{\alpha}{1+\alpha}, \quad \eta(\delta) := c_0' \, \delta^{\beta} = c_0' \, \delta^{ \frac{\alpha}{1+\alpha} },
\end{align}
thereby obtaining \eqref{eq: gl 1, quant für normal} (first for all $t \in I \setminus N$ and then for all $t \in I$). 
\smallskip

With~\eqref{eq: gl 1, quant für normal} at hand, we can now show the desired conclusion in essentially the same way as in the proof of the previous theorem (but for the convenience of the reader, we give a self-contained argument). 
Indeed, since $A(t)$ is a spectral operator of scalar type and $P(t) = P^{A(t)}(\{\lambda(t)\})$ for almost every $t \in I$, the projection $P(t)$ for almost every $t \in I$ is weakly associated of order $1$ with $A(t)$ and $\lambda(t)$ (Proposition~\ref{prop: krit ex schw assoz proj, A spektral}) 
and so 
\begin{align*}
P(t)A(t) \subset A(t)P(t) = \lambda(t)P(t)
\end{align*}
holds for every $t \in I$ by the closedness argument in~\eqref{eq: P(t) vertauscht mit A(t) fuer alle t}. 
We can therefore conclude from~\eqref{eq: gl 4, adsatz ohne sl} and~\eqref{eq: gl 6, adsatz ohne sl} (with $Q_n$ and $Q_n'$ replaced by $P'$ and $P''$ at any occurrence and with $m_0 = 1$)
and from~\eqref{eq: gl 1, quant für normal}
that there is a constant $c'$ such that
\begin{align} \label{eq: gl 5, quant für normal}
\sup_{t \in I} \norm{ U_{\eps}(t)-V_{\eps}(t) } \le c' \Big( \eps \, \frac 1 \delta + \eps \, \frac{1}{\delta^2} \, \eta(\delta) + \eta(\delta) \Big)
\end{align}
for all $\eps \in (0,\infty)$ and $\delta \in (0,\delta_0']$ with $\eta$ as in~\eqref{eq: gl 4, quant für normal} above. 
Choosing now $\delta_{\eps} := \eps^{\frac 1 2}$ we immediately get the desired conclusion from~\eqref{eq: gl 5, quant für normal} and~\eqref{eq: gl 4, quant für normal}. 
(In order to see~\eqref{eq: gl 5, quant für normal}, notice that the first and third term on the right-hand side of~\eqref{eq: gl 5, quant für normal} are upper bounds for the $B_{\bm{\delta}}$-terms and $C_{\bm{\delta}}$-terms in~\eqref{eq: gl 6, adsatz ohne sl} by virtue of~\eqref{eq: gl 1, quant für normal}. And to see that the middle term in~\eqref{eq: gl 5, quant für normal} is an upper bound for the $B_{\bm{\delta}}'$-terms in~\eqref{eq: gl 6, adsatz ohne sl}, argue as in the proof of the previous theorem, but notice that now it is sufficient to have instead of~\eqref{eq: 2, quant adsatz ohne sl} a $\delta$-independent bound on the integral of $s \mapsto (A'(s)-\lambda'(s)- \delta r'(s)) (A(s)-1)^{-1}$ because now we cannot only estimate the integral of $s \mapsto \ol{R}_{\delta}(s)$ but by~\eqref{eq: gl 1, quant für normal} even its supremum.)
\end{proof}

\subsection{Some examples} \label{sect: bsp, adsaetze ohne sl}

We begin with two examples of operators $A(t)$ with eigenvalues $\lambda(t)$ that are allowed to be non-isolated and non-weakly-semisimple for every $t \in I$. 
In the first example, the operators $A(t)$ are spectral.

\begin{ex} \label{ex: A_2(t) diagb}
Suppose $A$, $\lambda$, $P$ with $A(t) = R(t)^{-1} A_0(t) R(t)$, $P(t) = R(t)^{-1} P_0 R(t)$, and $R(t) = e^{C t}$ are given as follows in $X := \ell^p(I_d) \times \ell^p(I_{\infty})$ (where $p \in [1,\infty)$ and $d \in \N$): 
\begin{align*}
A_0(t) := 
\begin{pmatrix} \lambda(t) + \alpha(t) N  & 0 \\ 0 & \operatorname{diag}\big( (\lambda_n)_{n \in \N} \big) \end{pmatrix} 
\quad \text{and} \quad
P_0 := \begin{pmatrix} 1 & 0 \\ 0 & 0 \end{pmatrix},
\end{align*} 
where $\lambda(t) \in (-\infty, 0]$, $\alpha(t)$, $N$ are such that Condition~\ref{cond: baustein mit nicht-halbeinfachem ew} is satisfied and where $(\lambda_n)_{n \in \N}$ is an enumeration of $[-1,0] \cap \Q$ such that $\lambda(t) \notin \{ \lambda_n: n \in \N \}$ for almost every $t \in I$. 
Additionally, suppose $t \mapsto \lambda(t)$ and $t \mapsto \alpha(t)$ are absolutely continuous and $C$ is the right shift operator on $\ell^p(I_d) \times \ell^p(I_{\infty}) \cong \ell^p(I_{\infty})$:
\begin{align} \label{eq: def C, ex 1 without sg}
C(z_1, \dots, z_d, z_{d+1}, \dots ) := (0, z_1, \dots, z_{d-1}, z_d, \dots ).
\end{align} 
Then $t \mapsto A(t)$ is in $W^{1,\infty}_*(I,L(X))$ and $t \mapsto A_0(t)$ is $(M_0,0)$-stable (by Lemma~\ref{lm: char (M,0)-stab für einfaches A}), so that $A$ is $(M,0)$-stable for some $M \in [1,\infty)$ by Lemma~\ref{lm: (M,w)-stabilität und ähnl.trf.}. Since $A_0(t)|_{P_0 X} - \lambda(t)$ 
is nilpotent of order at most $m_0 := \operatorname{dim} \ell^p(I_d) = d$ for every $t \in I$ 
and since $A_0(t)|_{(1-P_0)X}-\lambda(t)$ is injective and has dense range in $(1-P_0)X$ 
(because $\lambda(t) \notin \{ \lambda_n: n \in \N \}$) 
for almost every $t \in I$, 
$P_0$ is weakly associated of order $m_0$ 
with $A_0(t)$ and $\lambda(t)$, whence the same is true for $A(t)$, $P$ instead of $A_0(t)$ and $P_0$.
And finally, the resolvent estimate~\eqref{eq: resolvent estimate} 
is clearly fulfilled if we choose $\vartheta(t) := \frac{\pi}{2}$ for all $t \in I$. 
All 
other hypotheses of 
Theorem~\ref{thm: erw adsatz ohne sl}~(i) are obvious. 
$\blacktriangleleft$
\end{ex}

In the second example, the operators $A(t)$ are not spectral (by Theorem~XV.3.10 and
~XV.8.7 of~\cite{DunfordSchwartz} 
and by the uncountability of $\sigma_r(S_+) = B_1(0)$ for the right shift operator $S_+$ on $X = \ell^p(I_{\infty})$ with $p \ne 1$). 

\begin{ex} \label{ex: A_2(t) nicht diagb}
Suppose $A$, $\lambda$, $P$ with $A(t) = R(t)^{-1} A_0(t) R(t)$, $P(t) = R(t)^{-1} P_0 R(t)$, and $R(t) = e^{C t}$ are given as follows in $X := \ell^p(I_d) \times \ell^p(I_{\infty})$ (where $p \in (1,\infty)$ and $d \in \N$): 
\begin{align*}
A_0(t) := 
\begin{pmatrix} \lambda(t) + \alpha(t) N  & 0 \\ 0 & S_+ - 1 \end{pmatrix} 
\quad \text{and} \quad
P_0 := \begin{pmatrix} 1 & 0 \\ 0 & 0 \end{pmatrix},
\end{align*} 
where $\lambda(t) \in \partial B_1(-1)$, $\alpha(t)$, $N$ are such that Condition~\ref{cond: baustein mit nicht-halbeinfachem ew} is satisfied. Additionally, $t \mapsto \lambda(t)$ and $t \mapsto \alpha(t)$ are absolutely continuous and $C$ is the bounded linear operator in $\ell^p(I_d) \times \ell^p(I_{\infty}) \cong \ell^p(I_{\infty})$ given by 
\begin{align} \label{eq: def C, ex 2 without sg}
C (z_1, \dots, z_d, z_{d+1}, \dots ) := (0, \dots, 0, z_{d+1}, -z_d, 0, \dots),
\end{align}
where in the vector on the right $z_{d+1}$, $-z_d$ appear in the $d$th and $(d+1)$th place. 
Since $\lambda(t) \in \partial B_1(-1) = \sigma_c(S_+ - 1)$ 
for every $t \in I$ (because $p \ne 1$), 
$P_0$ is weakly associated with $A_0(t)$ and $\lambda(t)$ and therefore the same goes for $A_0(t)$, $P_0$ replaced by $A(t)$ and $P(t)$.
Also, if for every $t \in I$ we choose $\vartheta(t)$ such that $\lambda(t) = -1 + e^{i \vartheta(t)}$, 
then the resolvent estimate~\eqref{eq: resolvent estimate} 
holds true 
because
\begin{align*}
\norm{  \big( \lambda(t) + \delta e^{i \vartheta(t)} - A_0(t) \big)^{-1} (1-P_0) } \le \norm{  \big( 1 + \delta - e^{-i \vartheta(t)} S_+ \big)^{-1}  } \le \frac{1}{\delta}
\end{align*}
for every $t \in I$ and $\delta \in (0,\infty)$ (use a Neumann series expansion!).
$\blacktriangleleft$
\end{ex}

We chose the operators $C$ in the particular way~\eqref{eq: def C, ex 1 without sg} and~\eqref{eq: def C, ex 2 without sg} above in order to make sure that the trivial adiabatic theorem from Section~\ref{sect: ad zeitentw} cannot be applied and that the examples cannot be reduced to finite-dimensional examples. See~\cite{diss} for detailed explanations.
%
%
In our last example we show that the conclusion of the adiabatic theorem without spectral gap condition may fail if the regularity assumption on $P$ 
is the only one to be violated. 

\begin{ex}  \label{ex: reg an P wesentl, ohne sl}
Set $A(t):= M_{f_t}$ in $X := L^p(\R)$ (for some $p \in [1,\infty)$), where 
\begin{align*}
f_t := f_0(\,.\, + t) \quad \text{with} \quad 0 \ne f_0 \in C_c^1(\R, i \R),
\end{align*}
$\lambda(t) := 0$ and $P(t) := M_{\chi_{E_t}}$ with $E_t := \{ f_t = 0 \}$. 
Then all the assumptions of the adiabatic theorem without spectral gap condition -- in the version for projections of infinite rank (second remark after Theorem~\ref{thm: erw adsatz ohne sl}) -- are satisfied with the sole exception that $t \mapsto P(t)$ is not strongly continuously differentiable (by Lemma~3.5.3 of~\cite{diss}). 
And indeed, the conclusion of the adiabatic theorem 
already fails: as the $A(t)$ are pairwise commuting and $t \mapsto f_t(x)$ is Riemann integrable for every $x \in \R$, one has
\begin{align*}
\big( U_{\eps}(t,s) g \big) (x) = \Big( e^{\frac 1 \eps \int_s^t A(\tau) \, d\tau} \, g \Big)(x) = e^{ \frac 1 \eps \int_s^t f_{\tau}(x) \, d\tau} \, g(x)
\end{align*}  
for almost every $x \in \R$ and therefore (by $f_0(\R) \subset i \R$)
\begin{align} \label{eq: ex 3 without sg}
\norm{ (1-P(t)) U_{\eps}(t) P(0) g }^p 
= \int \big| (1-\chi_{E_t}(x)) \chi_{E_0}(x) g(x) \big|^p \, dx
\end{align}
for every $t \in I$, $\eps \in (0,\infty)$ and $g \in X$. 
Since for every $t \in (0,1]$ there is a $g \in X$ such that the right-hand side of this equation 
the conclusion 
of the adiabatic theorem without spectral gap 
-- more precisely, the weaker assertion that $\sup_{t \in I} \norm{(1-P(t))U_{\eps}(t)P(0)g} \longrightarrow 0$ for all $g \in X$ --
fails.
$\blacktriangleleft$
\end{ex}


It should be pointed out that the failure of both the assumptions 
and the conclusion of the adiabatic theorems without spectral gap condition presented above is a quite typical phenomenon in the case where $A(t) = M_{f_t}$ in $X = L^p(X_0)$ for some $p \in [1,\infty)$ and some $\sigma$-finite measure space $(X_0, \mathcal{A}, \mu)$. 
Indeed, if $A(t) = M_{f_t}$ in $X = L^p(X_0)$ for measurable functions $f_t: X_0 \to \{ \Re z \le 0 \}$ such that $D(M_{f_t}) = D$ for all $t \in I$, 
if $\lambda(t)$ is an eigenvalue of $A(t)$, and if $P(t)$ for almost every $t \in I$ (with exceptional set $N$) is 
weakly associated with $A(t)$ and $\lambda(t)$, 
then 
\begin{align*}
P(t) = M_{\chi_{ \{ f_t = \lambda(t) \} }} = M_{\chi_{E_t}} \text{ for every } t \in I \setminus N 
\end{align*}
by Theorem~\ref{thm: typ mögl für PX und (1-P)X}, and therefore the following holds true.
As soon as $I \setminus N \ni t \mapsto P(t)$ is not constant, the assumptions of the adiabatic theorem without spectral gap (Theorem~\ref{thm: erw adsatz ohne sl}) must fail (because then $I \setminus N \ni t \mapsto P(t) = M_{\chi_{E_t}}$ cannot extend to a strongly continuously differentiable map 
by Lemma~3.5.3 of~\cite{diss}). 
And as soon as, in addition, the maps $f_t$ are $i \R$-valued and $t \mapsto f_t g \in X$ is continuous for all $g \in D$, the conclusion 
of Theorem~\ref{thm: erw adsatz ohne sl}, or more precisely, of its corollary
\begin{align*}
\sup_{t \in I} \norm{(1-P(t))U_{\eps}(t)P(0)} \longrightarrow 0 \quad \text{and} \quad \sup_{t \in I} \norm{P(t)U_{\eps}(t)(1-P(0))} \longrightarrow 0,
\end{align*}
must fail as well. (In order to see this, observe from~\cite{NickelSchnaubelt98} (Theorem~2.3) or~\cite{SchmidJEE} (Theorem~2.1) 
that the evolution system $U_{\eps}$ for $\frac 1 \eps A$ exists on $D$ and can be strongly approximated by finite products of operators of the form $e^{M_{f_{\tau}} \, \sigma}$, 
so that for arbitrary $g \in X$ 
\begin{align*}
\Big| (1-\chi_{E_t}(x)) \big( U_{\eps}(t) \chi_{E_0} g \big)(x) - \chi_{E_t}(x) \big( U_{\eps}(t) (1-\chi_{E_0}) g \big)(x) \Big| = 
\big| \chi_{E_t}(x) - \chi_{E_0}(x) \big| \big| g(x) \big|
\end{align*}
for almost every $x \in X_0$. 
Consequently, 
\begin{align}
\big\| (1-P(t))U_{\eps}(t)P(0)g - P(t)U_{\eps}(t)(1-P(0))g \big\| = \norm{P(t)g - P(0)g}
\end{align} 
for all $t \in I \setminus N$, $\eps \in (0,\infty)$ 
and since $I \setminus N \ni t \mapsto P(t)$ is not constant, 
there is a $t \in (0,1]$ and a $g \in X$ such that $(1-P(t))U_{\eps}(t)P(0)g$ and $P(t)U_{\eps}(t)(1-P(0))g$ do not both converge to $0$ as $\eps \searrow 0$.)

\subsection{An application to open quantum systems} \label{sect: anwendung q.d.s.} 

In this section we apply the adiabatic theorem without spectral gap condition from Section~\ref{sect: qual adsatz ohne sl} to weakly dephasing generators $A(t)$ of quantum dynamical semigroups in $X = S^p(\mathfrak{h})$ with $p \in (1,\infty)$ and with $\lambda(t) = 0$. So, 
\begin{align} \label{eq: Lambda(t) auf S^1}
A(t) \rho := Z_0(t)(\rho) + \sum_{j \in J} B_j(t) \rho B_j(t)^* - 1/2 \{ B_j(t)^*B_j(t), \rho \} 
\end{align}
for $\rho \in D(Z_0(t))$ with $Z_0(t)$ being the generator of the semigroup on $S^p(\mathfrak{h})$ defined by $e^{Z_0(t)\tau}(\rho) :=  e^{-i H(t) \tau} \rho \, e^{iH(t) \tau}$, where $H(t): D(H(t)) \subset \mathfrak{h} \to \mathfrak{h}$ is a self-adjoint operator and $B_j(t)$ for every $j \in J$ is a bounded opertor in $\mathfrak{h}$ such that the weak dephasingness condition
\begin{align} \label{eq: weak dephasingness cond, t-dependent}
\sum_{j\in J} B_j(t) B_j(t)^* = \sum_{j\in J} B_j(t)^* B_j(t) < \infty
\end{align}
is satisfied for every $t \in I$.
It would be desirable to apply the adiabatic theorem to the respective operators $A(t)$ on the natural space $X = S^1(\frak{h})$, but in this (non-reflexive) space, existence of projections $P(t)$ weakly associated with $A(t)$ and $\lambda(t) = 0$ goes wrong quite often. 
In fact, every operator $A$ of the form~\eqref{eq: weakly dephas generator} in $X = S^1(\mathfrak{h})$, where $\mathfrak{h}$ is chosen to be infinite-dimensional and where the operators $H$ and $B_j$ are chosen such that
\begin{itemize}
\item $H$ has finite point spectrum $\sigma_p(H)$ and each $\mu \in \sigma_p(H)$ has finite multiplicity, and
\item \eqref{eq: dephasingness cond} and \eqref{eq: vor lindblad, p=1} are satisfied,
\end{itemize}
is a dephasing generator of a quantum dynamical semigroup on $S^1(\mathfrak{h})$, but there exists no projection weakly associated with $A$ and $\lambda = 0$. 
(If such a projection existed, then
\begin{align*}
S^1(\frak{h}) = X = \ker A \oplus \ol{\ran} \, A = N \oplus R \qquad (N := \ker A \text{\, and \,} R := \ol{\ran} \, A)
\end{align*}
by the same argument as in the first remark after our adiabatic theorem without spectral gap condition (Theorem~\ref{thm: erw adsatz ohne sl}). 
So, on the one hand $X / R \cong N$ and hence $(X / R)^*$ would be finite-dimensional by virtue of $N = \ker Z_0$ (Proposition~\ref{prop: properties (weakly) dephas generators}~(ii)) and of Lemma~\ref{lm: ker Z_0 endldim}~(i), but on the other hand $(X / R)^* \cong R^{\perp}$ (Theorem~III.10.2 of~\cite{Conway:fana}) would be infinite-dimensional by virtue of $R^{\perp} = \ker A^* \supset \ker Z_0^*$ (Proposition~\ref{prop: properties (weakly) dephas generators}~(iii)) and of Lemma~\ref{lm: ker Z_0 endldim}~(ii). Contradiction!)
%
%
In the (reflexive) space $X = S^p(\mathfrak{h})$ with $p \ne 1$, by contrast, existence of weakly associated projections is often for granted. 

\begin{lm} \label{lm: ex of weakly ass proj for q.d.s.}
Suppose $A$ is a weakly dephasing generator of a quantum dynamical semigroup 
on $X=S^p(\mathfrak{h})$ with $p \in (1,\infty)$ and that $\lambda = 0 \in \sigma(A)$. If (i) $\ker A$ is finite-dimensional or if (ii) $p = 2$, then there exists a unique projection $P$ weakly associated with $A$ and $\lambda$.
\end{lm}

\begin{proof}
Suppose first that $\ker A$ is finite-dimensional. We then see that $\ker A + \ol{\ran} \, A$ is closed in $X$ (Proposition~III.4.3 of~\cite{Conway:fana}) and hence the conclusion follows by Proposition~\ref{prop: 2nd criterion ex of w ass proj}. 
Suppose now that $p = 2$. We show 
that $\ker A $ is orthogonal to $\ran A$ in $X=S^2(\frak{h})$. 
It then follows that $\ker A + \ol{\ran} \, A$ is closed in $X$ and the conclusion follows again by Proposition~\ref{prop: 2nd criterion ex of w ass proj}.
So, let $\rho \in \ker A$ and write $A = Z_0+W$ for brevity. Since $Z_0^* = -Z_0$, we see that $\rho \in D(A) = D(Z_0) = D(Z_0^*) = D(A^*)$ and that
\begin{align}
A^*(\rho) = Z_0^*(\rho) + W^*(\rho) = -Z_0(\rho) + \sum_{j\in J} B_j^* \rho B_j - 1/2\{B_j^* B_j, \rho\} = 0,
\end{align}
where for the last equality Proposition~\ref{prop: properties (weakly) dephas generators}~(i) was used. Consequently, $\ker A \subset \ker A^* = (\ran A)^{\perp}$, as desired.
\end{proof}

In the special case of dephasing generators $A$ with bounded $H$, criterion~(ii) of the above lemma is due to~\cite{AvronGraf12}.
If $p \in (1,2]$, then Proposition~\ref{prop: properties (weakly) dephas generators}~(i) and Lemma~\ref{lm: ker Z_0 endldim}~(i)  yield a simple sufficient condition for the finite-dimensionality criterion~(i) 
from the above lemma.
If $p=2$  and $\ker A = \ker Z_0$, then the projection $P$ weakly associated with $A$ and $\lambda = 0$ is orthogonal (by the orthogonality of the subspaces $\ker A$ and $\ol{\ran} \, A$ in $S^2(\frak{h})$ just proved in the lemma above) 
and hence, by $\ker Z_0 = \{H\}' \cap S^2(\frak{h})$, is given explicitly as
\begin{align} \label{eq: P explicit by RAGE}
P \rho = \sum_{\mu \in \sigma_p(H)} Q^H_{\{\mu\}} \rho \, Q^H_{\{\mu\}} \qquad (\rho \in S^2(\frak{h})),
\end{align} 
where $Q^H$ denotes the spectral measure of $H$. See Theorem~5.8 of~\cite{Teschl09} or the discussion at the very end of~\cite{AvronGraf12}.

\begin{thm}
Suppose that $A(t)$ for every $t \in I$ is a weakly dephasing generator 
on $X = S^p(\mathfrak{h})$ $(p \in (1,\infty))$ with time-independent domain $D(Z_0(t)) = D$ and that $t \mapsto A(t)$ is in $W^{1,1}_*(I,L(Y,X))$, where $Y$ is the space $D$ endowed with the graph norm of $A(0)$. 
Suppose further that $\lambda(t) = 0$ is an eigenvalue of $A(t)$ for every $t \in I$ and, finally, that 
either
\begin{align*}
(i) \,\, \ker A(t) \text{ is finite-dimensional for almost every } t \in I \text{ or } (ii) \,\, p = 2, 
\end{align*}
and that there is a null set in $I$ such that the projections $P(t)$ weakly associated with $A(t)$ and $\lambda(t)$ for $t$ outside that null set can be extended to a continuously differentiable map $t \mapsto P(t)$ on the whole of $I$. 
Then
\begin{align*}
\sup_{t \in I} \norm{ (U_{\eps}(t) - V_{0\, \eps}(t)) P(0) \rho } \longrightarrow 0 \quad (\eps \searrow 0)
\end{align*}
for every $\rho \in X$, where $V_{0\, \eps}$ 
is the evolution system for $\frac 1 \eps A P + [P',P] = [P',P]$ on $X$. 
\end{thm}

\begin{proof}
We have only to notice that $A(t)$ generates a contraction semigroup in $X$ for every $t \in I$ (Theorem~\ref{thm: generation result q.d.s.}), 
that the projections $P(t)$ weakly associated with $A(t)$ and $\lambda(t)$ really exist for almost every $t \in I$ (Lemma~\ref{lm: ex of weakly ass proj for q.d.s.}), and then to apply the second remark after Theorem~\ref{thm: erw adsatz ohne sl}. 
\end{proof}

Clearly, the above theorem is a generalization of 
the respective result (Theorem~22) from~\cite{AvronGraf12} for 
dephasing generators $A(t)$ of quantum dynamical semigroups on $X = S^2(\frak{h})$ with bounded $H(t)$. 
Incidentally, these types of 
generators are normal operators on $S^2(\frak{h})$, that is, $A(t)^*A(t) = A(t)A(t)^*$ (as can be verified by straightforward calculations using the fact that $B_i$ and $B_j$ commute for all $i,j \in J$ by~\eqref{eq: dephasingness cond}). 
We conclude with a simple example of generators $A(t)$ in $X = S^2(\frak{h})$ satisfying the assumptions of the above theorem without being dephasing (or normal). 

\begin{ex}
We choose the operators $H$ and $B$ as in Example~\ref{ex: H vert nicht mit H_0} 
and, in addition, we take $H$ to be bounded. 
We then define $A(t)$ for every $t \in I$ on $X := S^2(\frak{h})$ through 
\begin{gather*}
A(t)\rho := Z_0(t)(\rho) + B(t) \rho B(t)^* - 1/2 \{ B(t)^*B(t), \rho \}  \qquad (\rho \in S^2(\frak{h})) 
\end{gather*}
with $Z_0(t)(\rho) := -i [H(t),\rho]$, where $H(t) := R(t)^{-1} H R(t)$ and 
$B(t) := R(t)^{-1} B R(t)$ with $R(t) := e^{i C t}$ and $C$ a bounded self-adjoint operator on $\frak{h}$. 
Clearly, $A(t)$ for every $t \in I$ is a weakly dephasing generator and $D(A(t)) = X$ is time-independent while $t \mapsto A(t)$ is $W^{1,1}_*$-regular. 
It is also clear that $\ker A(t) = \ker Z_0(t)$ by Example~\ref{ex: H vert nicht mit H_0}.
So, by the remarks around~\eqref{eq: P explicit by RAGE}, the projection $P(t)$ weakly associated with $A(t)$ and $\lambda(t) = 0$ is explicitly given by
\begin{align*}
P(t) \rho = \sum_{ \mu \in \sigma_p(H(t)) } Q_{\{\mu\}}^{H(t)}  \rho \, Q_{\{\mu\}}^{H(t)} = \sum_{ \mu \in \sigma_p(H) } R(t)^{-1} Q_{\{\mu\}}^{H} R(t) \, \rho \, R(t)^{-1} Q_{\{\mu\}}^{H} R(t)
\end{align*}
for every $t \in I$, where $Q^{H(t)}$ and $Q^{H}$ denote the spectral measures of $H(t)$ and $H$. In particular, $t \mapsto P(t)$ is continuously differentiable.
So all the assumptions of the above theorem are satisfied, but $A(t)$ is non-dephasing for every $t$ 
because 
\begin{align*}
H(t)B(t) \ne B(t)H(t)
\end{align*}
by Example~\ref{ex: H vert nicht mit H_0}. 
Also, $A(t) = Z_0(t) + W(t)$ is non-normal on $X$ for every $t$ because $Z_0(t)$ is skew-adjoint and $W(t)$ is self-adjoint, but $Z_0(t)$ does not commute with $W(t)$ (as is verified in~\cite{diss} (Example~4.2.11)). 
$\blacktriangleleft$
\end{ex}

\subsection{An application to adiabatic switching}  \label{sec: ad switching}

In this section we apply the adiabatic theorem without spectral gap condition from Section~\ref{sect: qual adsatz ohne sl} -- in the version for several eigenvalues -- to adiabatic switching procedures.

\subsubsection{Setting and assumptions}

Adiabatic switching of (linear) perturbations has a long tradition in quantum physics. Since the famous work~\cite{Gell-MannLow51} of Gell-Mann and Low, it has been used, for instance, to relate 
-- by what is now known as the Gell-Mann and Low formula -- the eigenstates of a perturbed system, described by $A_0+V$, 
to the eigenstates of the unperturbed system, described by $A_0$. 
Adiabatic switching, in this context, means that $A_0 
= \ul{A}(0)$ is infinitely slowly deformed into $\ul{A}(1) 
= A_0+V$ in the following sense: one chooses a 
switching function $\kappa: (-\infty,0] \to [0,1]$ vanishing at $-\infty$ and taking the value $1$ at $0$ 
and then passes -- more and more slowly -- from $A_0 = \ul{A}(\kappa(-\infty))$ via
\begin{align*}
\{-\infty\} \cup (-\infty,0] \ni s \mapsto \ul{A}(\kappa(\eps s)) = A_0 + \kappa(\eps s)\, V
\end{align*}   
to $\ul{A}(\kappa(0)) = A_0+V$ by making the slowness parameter $\eps \in (0,\infty)$ smaller and smaller.
%
%
A rigorous -- and non-perturbative -- proof of the Gell-Mann and Low formula for non-degenerate and isolated eigenvalues $\ul{\lambda}(\kappa)$ of $\ul{A}(\kappa) = A_0 + \kappa V$ has been given by Nenciu and Rasche in~\cite{NenciuRasche89}. It is based on the adiabatic theorem with spectral gap condition. 
In a recent paper~\cite{Panati10} of Brouder, Panati, Stoltz, the Gell-Mann and Low theorem has been extended to the case of degenerate isolated eigenvalues -- again by using the adiabatic theorem with spectral gap condition.
In this section, we further extend the Gell-Mann and Low theorem to the case of 
non-isolated degenerate eigenvalues. 
We consider the following setting. 



\begin{cond} \label{cond: vor 1}
$\ul{A}(\kappa) := A_0 + \kappa V$ for $\kappa \in [0,1]$, where $A_0: D \subset H \to H$ is a skew-adjoint operator in the Hilbert space $H$ 
and where $V$ is a skew-symmetric operator in $H$ that is $A_0$-bounded with relative bound less than $1$.
$\ul{\lambda}_1(\kappa), \dots, \ul{\lambda}_r(\kappa) $ for every $\kappa \in [0,1]$ are eigenvalues of $\ul{A}(\kappa)$ such that $\kappa \mapsto \ul{\lambda}_j(\kappa)$ is continuously differentiable for every $j \in \{1, \dots, r\}$.
And finally, $\ul{P}_1(\kappa), \dots, \ul{P}_r(\kappa)$ for every $\kappa \in [0,1]$ are orthogonal projections in $H$ such that $\kappa \mapsto \ul{P}_j(\kappa)$ is twice strongly continuously differentiable, $0 \ne \rk \ul{P}_j(0) < \infty$, and $\ul{P}_j(\kappa)$ is the spectral projection of $\ul{A}(\kappa)$ corresponding to $\lambda_j(\kappa)$ for every $\kappa \in [0,1] \setminus N$ with some exceptional set $N$.    
\end{cond}

\begin{cond} \label{cond: vor 2}
$\kappa: (-\infty,0] \to [0,1]$ is a non-decreasing twice continuously differentiable (switching) function such that
\begin{itemize}
\item[(i)] $\kappa(t) \longrightarrow \kappa(-\infty) = 0$ as $t \to -\infty$ and $\kappa(0) = 1$
\item[(ii)] $\kappa$ and $\kappa'$ are integrable on $(-\infty,0]$.
\end{itemize}
\end{cond}
%
%
Suppose now that $\ul{A}$, $\ul{\lambda}_1, \dots, \ul{\lambda}_r$, $\ul{P}_1, \dots, \ul{P}_r$ satisfy Condition~\ref{cond: vor 1} and that $\kappa$ is as in Condition~\ref{cond: vor 2} 
and define
\begin{align}
A(t) := \ul{A}(\kappa(t)), \quad \lambda_j(t) := \ul{\lambda}_j(\kappa(t)), \quad P_j(t) := \ul{P}_j(\kappa(t))
\end{align} \label{eq: gell-mann--low, def A,lambda_j,P_j}
for $t \in (-\infty,0]$ and $j \in \{1, \dots, r\}$, 
along with
\begin{align} \label{eq: gell-mann--low, def K}
\ul{K}(\kappa) := \frac 1 2 \sum_{j=1}^{r+1} [P_j'(\kappa), P_j(\kappa)] \quad \text{and} \quad K(t) := \frac 1 2 \sum_{j=1}^{r+1} [P_j'(t), P_j(t)] = \kappa'(t) \ul{K}(\kappa(t)),
\end{align}
where $\ul{P}_{r+1}(\kappa) := 1-\ul{P}_1(\kappa) - \dotsb - \ul{P}_{r}(\kappa)$ for $\kappa \in [0,1]$ and $P_{r+1}(t) := \ul{P}_{r+1}(\kappa(t))$ for $t \in (-\infty,0]$.
It then follows by the standard well-posedness result of Kato (Theorem~6.1 from~\cite{Kato70}) mentioned in Section~\ref{sect: wohlg und zeitentw} that the evolution systems $U_{\eps}$, $V_{\eps}$ for the families $\frac 1 \eps A$ and $\frac 1 \eps A + K$ exist on $D$ 
and, by the skew-adjointness of $\frac 1 \eps A(t)$ and $K(t)$ for $t \in (-\infty,0]$, the evolution operators $U_{\eps}(t,s)$, $V_{\eps}(t,s)$ are unitary for all $(s,t) \in \Delta_{(-\infty,0]} := \{ (s,t) \in (-\infty,0]^2: s \le t\}$.
Instead of $U_{\eps}$, $V_{\eps}$, the Gell-Mann and Low formula and its proof below make use of the interaction picture counterparts 
$U_{\eps}^I$, $V_{\eps}^I$, defined by
\begin{align}
U_{\eps}^I(t,s) := e^{- A_0 t/\eps} U_{\eps}(t,s) e^{ A_0 s/\eps} \quad \text{and} \quad V_{\eps}^I(t,s) := e^{- A_0 t/\eps} V_{\eps}(t,s) e^{ A_0 s/\eps}  
\end{align}
for $(s,t) \in \Delta_{(-\infty,0]}$. It is easy to see that $U_{\eps}^I$, $V_{\eps}^I$ are the evolution systems for $\frac 1 \eps  A^I$ and $\frac 1 \eps  A^I + K^I$ on $D$, where 
\begin{align} \label{eq: gell-mann--low, def A^{I}}
A^I(t) := \kappa(t) \, e^{-A_0 t/\eps} \, V \, e^{A_0 t/\eps} \big|_D \quad \text{and} \quad K^I(t) := e^{-A_0 t/\eps} K(t) e^{A_0 t/\eps}.
\end{align}
(In order to see that 
the derivative of $t \mapsto U_{\eps}^I(t,s)x$ for $x \in D$ really is continuous -- as is required in the definition of evolution systems -- use that $t \mapsto U_{\eps}(t,s)|_Y$ is strongly continuous in $L(Y,Y)$ (Theorem~6.1~(f) of~\cite{Kato70}) and that $V|_Y$ is in $L(Y,H)$, where $Y$ denotes the space $D$ endowed with the graph norm of $A_0$.) 

\subsubsection{Adiabatic switching and a Gell-Mann and Low theorem without spectral gap condition}

We can now state and prove a Gell-Mann and Low theorem without spectral gap condition, 
where the eigenvalues $\lambda_1(t), \dots, \lambda_r(t)$ of $A(t) = \ul{A}(\kappa(t))$ are allowed to be non-isolated in $\sigma(A(t))$ for every $t \in (-\infty,0]$ 
-- as long as they stay isolated from each other except for a null set of crossing points. 


\begin{thm}  \label{thm: Gell-Mann--Low}
Suppose $\ul{A}$, $\ul{\lambda}_1, \dots, \ul{\lambda}_r$, $\ul{P}_1, \dots, \ul{P}_r$ are as in Condition~\ref{cond: vor 1} and that $\kappa$ is as in Condition~\ref{cond: vor 2} and define
$A(t)$, $\lambda_j(t)$, $P_j(t)$ for $t \in \{-\infty\} \cup (-\infty,0]$ and $j \in \{1, \dots, r\}$ as in~\eqref{eq: gell-mann--low, def A,lambda_j,P_j}. 
Suppose further that for all $j, j' \in \{1, \dots, r\}$ with $j \ne j'$ one has $\lambda_j \ne \lambda_{j'}$ almost everywhere, and that the exceptional set 
\begin{align*}
\big\{ t \in \{-\infty\} \cup (-\infty,0]: \kappa(t) \in N \big\}
\end{align*} 
where the $P_j$ are allowed to differ from the spectral projection of $A$ corresponding to $\lambda_j$, is a null set 
(remember Condition~\ref{cond: vor 1} for the definition of $N$). 
Then
\begin{align*}
\frac{ U_{\eps}^I(0,-\infty)x }{ \scprd{ x', U_{\eps}^I(0,-\infty)x } } \longrightarrow \frac{ W(0,-\infty)x }{ \scprd{ x', W(0,-\infty)x } } \in \ker(A(0)-\lambda_j(0)) \quad (\eps \searrow 0) 
\end{align*}
for all $x \in P_j(-\infty)H$ and $x' \in H$ such that $\scprd{ x', W(0,-\infty)x } \ne 0$. 
In the above relations $W$ denotes the evolution system for $K$, where $K(t)$ 
is defined as in~\eqref{eq: gell-mann--low, def K}.
\end{thm}

\begin{proof}
We proceed in three steps following the lines of proof of~\cite{Panati10}. 
As a first simple step observe that the limit 
\begin{align*} 
W(0,-\infty) := \lim_{t \to -\infty} W(0,t),
\end{align*}
employed in the very formulation of the theorem, exists w.r.t.~the norm operator topology of $H$ and that, likewise, the limits
\begin{align*} 
U_{\eps}^I(0,-\infty)x := \lim_{t \to -\infty} U_{\eps}^I(0,t)x \quad \text{and} \quad V_{\eps}^I(0,-\infty)x := \lim_{t \to -\infty} V_{\eps}^I(0,t)x
\end{align*}
exist for every $x \in H$. Indeed, by virtue of~\eqref{eq: gell-mann--low, def K},
\begin{align*}
\norm{W(0,t) - W(0,t')} =  \bigg \| \int_t^{t'} W(0,\tau) K(\tau) \, d\tau \bigg\| \le \Big| \int_t^{t'} c \, \kappa'(\tau) \, d\tau \Big| \longrightarrow 0 \quad (t,t' \to -\infty),
\end{align*}
and similarly, using the relative boundedness of $V$ w.r.t.~$A_0$ and the density of $D$ in $H$, one sees the existence of the two other limits.
\smallskip

As a second step we show that the assertion holds true at least for $V_{\eps}^I(0,-\infty)$ instead of $U_{\eps}^I(0,-\infty)$, more precisely, 
\begin{align} \label{eq: gell-mann--low, beh für V_{eps}^I statt U_{eps}^I}
\frac{ V_{\eps}^I(0,-\infty)x }{ \scprd{ x', V_{\eps}^I(0,-\infty)x } } = \frac{ W(0,-\infty)x }{ \scprd{ x', W(0,-\infty)x } } \in \ker(A(0)-\lambda_j(0)) 
\end{align} 
for every $x \in P_j(-\infty)H$ and every $x' \in H$ such that $\scprd{ x', W(0,-\infty)x } \ne 0$. 
So choose and fix vectors $x$ and $x'$ as above -- notice that such vectors always exist by $\rk P_j(0) \ne 0$ and by the unitarity of $W(0,-\infty)$. 
Since 
\begin{align} \label{eq: gell-mann--low, inklusion für alle t}
P_j(t)H \subset \ker(A(t)-\lambda_j(t))
\end{align} 
for every $t \in \{-\infty\} \cup (-\infty,0]$ (use a continuity argument to extend this inclusion from $\{-\infty\} \cup (-\infty, 0] \setminus \kappa^{-1}(N)$ to all of $\{-\infty\} \cup (-\infty,0]$) and since $V_{\eps}$ is adiabatic w.r.t.~$P_j$, 
it follows that
\begin{align*}
V_{\eps}(s,t)P_j(t) = e^{1/\eps \int_t^s \lambda_j(\tau) \, d\tau} \, W(s,t)P_j(t) 
\end{align*}
for all $(t,s) \in \Delta_{(-\infty,0]}$, in other words: the $\eps$-dependence of $V_{\eps}(s,t)P_j(t)$ is solely contained in a scalar factor. Consequently,
\begin{align*}
V_{\eps}^I(0,t)x = V_{\eps}(0,t) e^{1/\eps \lambda_j(-\infty) t} x 
&= e^{1/\eps \int_t^0 \lambda_j(\tau) - \lambda_j(-\infty) \, d\tau} \, W(0,t) P_j(t)x  \notag \\
&\qquad \qquad \quad +  e^{1/\eps \, \lambda_j(-\infty) t} \,  V_{\eps}(0,t) \big(P_j(-\infty) - P_j(t) \big)x,
\end{align*}
from which it follows 
with the help of
\begin{align*}
\big| \lambda_j(\tau)-\lambda_j(-\infty) \big| = \big| \ul{\lambda}_j(\kappa(\tau)) - \ul{\lambda}_j(0) \big| \le \norm{ \ul{\lambda}_j' }_{\infty} \, \kappa(\tau) \quad (\tau \in (-\infty,0]) 
\end{align*}
and the integrability of $\kappa$ 
that
\begin{align} \label{eq: gell-mann--low, formel für V_{eps}^I(0,-infty)}
V_{\eps}^I(0,-\infty)x = e^{1/\eps \int_{-\infty}^0 \lambda_j(\tau) - \lambda_j(-\infty) \, d\tau} \, W(0,-\infty) P_j(-\infty)x
\end{align}
for every $\eps \in (0,\infty)$. 
We now see that the equality in~\eqref{eq: gell-mann--low, beh für V_{eps}^I statt U_{eps}^I} holds true, and the element relation in~\eqref{eq: gell-mann--low, beh für V_{eps}^I statt U_{eps}^I} follows by the adiabaticity of $W$ w.r.t.~$P_j$ and by~\eqref{eq: gell-mann--low, inklusion für alle t}.
\smallskip


As a third -- core -- 
step resting upon the adiabatic theorem without spectral gap condition, we show that
\begin{align} \label{eq: gell-mann--low, schritt 3}
V_{\eps}^I(0,-\infty)x - U_{\eps}^I(0,-\infty)x \longrightarrow 0 \quad (\eps \searrow 0)
\end{align}
for every $x \in P_j(-\infty)H$, which then yields the convergence
\begin{align*}
\frac{ V_{\eps}^I(0,-\infty)x }{ \scprd{ x', V_{\eps}^I(0,-\infty)x } } - \frac{ U_{\eps}^I(0,-\infty)x }{ \scprd{ x', U_{\eps}^I(0,-\infty)x } } \longrightarrow 0 \quad (\eps \searrow 0)
\end{align*}
for every $x \in P_j(-\infty)H$ and every $x' \in H$ for which $\scprd{ x', W(0,-\infty)x } \ne 0$, 
and hence -- by virtue of~\eqref{eq: gell-mann--low, beh für V_{eps}^I statt U_{eps}^I} -- the desired conclusion.
%
%
So let $x \in P_j(-\infty)H$ be fixed. Since $U_{\eps}^{I}$ and $V_{\eps}^{I}$ are the evolution systems for $\frac 1 \eps A^{I}$ and $\frac 1 \eps A^{I} + K^{I}$ on $D$ with $A^{I}$ and $K^{I}$ as in~\eqref{eq: gell-mann--low, def A^{I}}, we see that
\begin{align*}
V_{\eps}^I(0,t)x - U_{\eps}^I(0,t)x 
&= V_{\eps}^I(0,t_0) \int_t^{t_0} U_{\eps}^I(t_0,\tau) K^I(\tau) V_{\eps}^I(\tau,t)x \,d\tau \\
& \qquad \qquad \qquad \qquad  + \big( V_{\eps}^I(0,t_0) - U_{\eps}^I(0,t_0) \big) U_{\eps}^I(t_0,t)x
\end{align*}
for every $t_0 \in (-\infty,0]$ and every $t \in (-\infty,t_0]$. So, by the unitarity of $V_{\eps}^I(t_0,t)$, $U_{\eps}^I(t_0,t)$, $e^{A_0 t/\eps}$ we get that
\begin{align*}
\norm{ V_{\eps}^{I}(0,-\infty)x - U_{\eps}^{I}(0,-\infty)x } \le \Big( C \int_{-\infty}^{t_0} \kappa'(\tau) \,d\tau + \big\| V_{\eps}(0,t_0) - U_{\eps}(0,t_0) \big\| \Big) \norm{x}
\end{align*}
for every $t_0 \in (-\infty,0]$ and $\eps \in (0,\infty)$, where $C := \sup_{\kappa \in [0,1]} \norm{ \ul{K}'(\kappa)} < \infty$. 
In view of the integrability of $\kappa'$ on $(-\infty,0]$, it remains to show that 
\begin{align*} 
\big\| V_{\eps}(0,t_0) - U_{\eps}(0,t_0) \big\| \longrightarrow 0 \quad (\eps \searrow 0)
\end{align*}
for every fixed $t_0 \in (-\infty,0]$. 
And this, in turn, is an immediate consequence of the adiabatic theorem without spectral gap for several eigenvalues (third remark after Theorem~\ref{thm: erw adsatz ohne sl}) with the interval $[0,1]$ replaced by $[t_0,0]$.
\end{proof}


If in the situation of the above theorem one additionally assumes $\norm{P_j(0)-P_j(-\infty)} < 1$, then the vectors $x, x'$ with $\scprd{x', W(0,-\infty)x} \ne 0$ can be chosen to both lie in $P_j(-\infty)H$.
Indeed, under this additional assumption $1+P_j(-\infty)-P_j(0)$ is invertible and, by the adiabaticity of $W$ w.r.t.~$P_j$ and the unitarity of $W(0,-\infty)$, we therefore see that for every $x \in P_j(-\infty)H \setminus \{0\}$
\begin{align*}
P_j(-\infty)W(0,-\infty)x = \big( 1+P_j(-\infty)-P_j(0) \big) W(0,-\infty)x \ne 0.
\end{align*} 

%
With the above theorem at hand, we can now also extend a formula for the energy shift from~\cite{NenciuRasche89}, \cite{GrossRunge} to the more general situation of not necessarily isolated eigenvalues. 
It expresses the energy shift $\lambda_j(0)-\lambda_j(-\infty)$ 
as a limit of logarithmic derivatives of certain transition functions.

\begin{cor} \label{cor: energy shift}
Suppose that the assumptions of Theorem~\ref{thm: Gell-Mann--Low} are satisfied. Then the energy shift $\lambda_j(0)-\lambda_j(-\infty)$ can be expressed as a limit of logarithmic derivatives of certain transition functions, 
more precisely,
\begin{align}  \label{eq: energy shift, allg kappa}
\lambda_j(0)-\lambda_j(-\infty) = \lim_{\eps \searrow 0} \eps \, \ddt{  \log \scprd{x', U_{\eps}^I(t,-\infty)x}  } \Big|_{t=0}
\end{align}
for all $x, x' \in P_j(-\infty)H$ with $\scprd{x', W(0,-\infty)x} \in \C \setminus (-\infty,0]$. In the above equation, $\log$ denotes the principal branch of the complex logarithm defined on $\C \setminus (-\infty,0]$.
\end{cor}

\begin{proof}
We fix $j \in \{1, \dots, r\}$ and assume $x, x' \in P_j(-\infty)H$ with $\scprd{x',W(0,-\infty)x} \in \C \setminus (-\infty,0]$. (It should be noticed that existence of such vectors $x, x'$ is not claimed in the statement of the corollary. If, however,  $\norm{P_j(0)-P_j(-\infty)} < 1$, then such vectors do exist by the remark after Theorem~\ref{thm: Gell-Mann--Low}. 
And if $\norm{P_j(0)-P_j(-\infty)} \ge 1$, then one can switch on the full perturbation $V$ in intermediate steps as in~\cite{Panati10} (Section~3.4) and then apply the formula for the energy shift in each intermediate step.) 
%
We also set
\begin{align} \label{eq: energy shift, def f_eps}
f_{\eps}(t) := \scprd{x', U_{\eps}^{I}(t,-\infty)x } 
\quad \text{and} \quad
g_{\eps}(t) := \scprd{x', V_{\eps}^{I}(t,-\infty)x } 
\end{align}
for $t \in [-1,0]$ and $\eps \in (0,\infty)$.
\smallskip

As a first step we show that the function $f_{\eps}: [-1,0] \to \C$ is differentiable with derivative at $0$ given by
\begin{align} \label{eq: energy shift, f_(eps)'(0)}
f_{\eps}'(0) = - \frac 1 \eps \scprd{V x', U_{\eps}^{I}(0,-\infty)x}
\end{align}
for every $\eps \in (0,\infty)$. 
In order to do so, we consider the pointwise approximants $f_{\eps \, n}: [-1,0] \to \C$ to $f_{\eps}$ defined by 
\begin{align} \label{eq: energy shift, f_(eps n)}
f_{\eps \, n}(t) := \scprd{x', U_{\eps}^{I}(t,-n)x } \qquad (n \in \N).
\end{align}
Since $U_{\eps}^{I}$ is the evolution system for $\frac 1 \eps A^{I}$ on $D$ with $A^{I}$ given by~\eqref{eq: gell-mann--low, def A^{I}}, 
the function $f_{\eps\, n}$ is differentiable for every $\eps \in (0,\infty)$ and every $n \in \N$ with 
\begin{align} \label{eq: energy shift, f_(eps n)'}
f_{\eps\, n}'(t) = - \frac{\kappa(t)}{\eps} \, \big\langle e^{-A_0 t/\eps} \, V e^{A_0 t/\eps} \, x',  U_{\eps}^{I}(t,-n)x \big\rangle
\end{align} 
for $t \in [-1,0]$, 
%
and, moreover,
\begin{align*}
&\sup_{t \in [-1,0]} \norm{ U_{\eps}^{I}(t,-n)x - U_{\eps}^{I}(t,-m)x } = 
\sup_{t \in [-1,0]}  \norm{ \int_{-n}^{-m}  U_{\eps}^{I}(t,\tau) \, \frac{\kappa(\tau)}{\eps} \, e^{-A_0 \tau/\eps} \, V e^{A_0 \tau/\eps}  x \, d\tau } \\
& \qquad \qquad \qquad \quad \le \frac 1 \eps \norm{V(A_0-1)^{-1}} \bigg| \int_{-n}^{-m} \kappa(\tau) \, d\tau \bigg| \norm{(A_0-1)x} \longrightarrow 0 
\end{align*}
as $m, n \to \infty$. 
So, $(f_{\eps\,n}')$ is uniformly convergent and, hence, the pointwise limit $f_{\eps}$ of the functions $f_{\eps\, n}$ is differentiable with derivative given by $f_{\eps}'(t) = \lim_{n \to \infty} f_{\eps\, n}'(t)$ for $t \in [-1,0]$. In particular, $f_{\eps}'(0)$ is given as in~\eqref{eq: energy shift, f_(eps)'(0)}.
\smallskip	

As a second step we show that $f_{\eps}(0) \ne 0$ for $\eps$ small enough and that
\begin{align} \label{eq: energy shift, schritt 2}
\eps \, f_{\eps}'(0)/f_{\eps}(0) \longrightarrow \lambda_j(0)-\lambda_j(-\infty) \quad (\eps \searrow 0),
\end{align}
from which~\eqref{eq: energy shift, allg kappa} readily follows. Since $|g_{\eps}(0)| = |\scprd{x',W(0,-\infty)x}| \ne 0$ for all $\eps \in (0,\infty)$ by virtue of~\eqref{eq: gell-mann--low, formel für V_{eps}^I(0,-infty)} and since $f_{\eps}(0) - g_{\eps}(0) \longrightarrow 0$ as $\eps \searrow 0$ by virtue of~\eqref{eq: gell-mann--low, schritt 3}, we see that indeed $f_{\eps}(0) \ne 0$ for $\eps$ small enough.
With the help of~\eqref{eq: energy shift, f_(eps)'(0)} and the previous theorem it then follows that
\begin{align} \label{eq: energy shift, schritt 2, 2}
\eps \, f_{\eps}'(0)/f_{\eps}(0) = 
-\frac{ \scprd{V x', U_{\eps}^{I}(0,-\infty)x } }{ \scprd{x',U_{\eps}^{I}(0,-\infty)x} }
\longrightarrow -\frac{ \scprd{V x', W(0,-\infty)x } }{ \scprd{x',W(0,-\infty)x} }  \qquad (\eps \searrow 0).
\end{align}  
Write now $V = A(0) - A(-\infty)$ and recall that $x' \in P_j(-\infty)H \subset \ker(A(-\infty)-\lambda_j(-\infty))$ and that $W(0,-\infty)x \in P_j(0)H \subset \ker(A(0)-\lambda_j(0))$ to obtain
\begin{align} \label{eq: energy shift, schritt 2, 3}
\scprd{V x', W(0,-\infty)x } 
= \big( \lambda_j(-\infty) - \lambda_j(0) \big) \scprd{x',W(0,-\infty)x}.
\end{align}
Combining~\eqref{eq: energy shift, schritt 2, 2} and~\eqref{eq: energy shift, schritt 2, 3} we then arrive at 
the asserted convergence~\eqref{eq: energy shift, schritt 2}. Clearly, 
\begin{align} \label{eq: energy shift, log abl}
f_{\eps}'(0)/f_{\eps}(0) = (\log \circ f_{\eps})'(0)
\end{align}
precisely for those $\eps \in (0,\infty)$ for which $f_{\eps}(0) \in \dom(\log) = \C \setminus (-\infty,0]$. So, \eqref{eq: energy shift, allg kappa} will follow from~\eqref{eq: energy shift, schritt 2} and~\eqref{eq: energy shift, log abl}, provided that $0$ is an accumulation point of the set
$E := \{ \eps \in (0,\infty): f_{\eps}(0) \in \C \setminus (-\infty,0] \}$.
%
Since 
\begin{gather*}
g_{\eps}(0) = \scprd{x', V_{\eps}^{I}(0,-\infty)x } = e^{i \phi_0/\eps} \, z_0, \\
i \phi_0 := \int_{-\infty}^0 \lambda_j(\tau) - \lambda_j(-\infty) \, d\tau \in i \R \quad \text{and} \quad z_0 := \scprd{x', W(0,-\infty)x} \in \C \setminus (-\infty,0]
\end{gather*}
and since $z_0 \in \C \setminus (-\infty,0]$, there exists a $\vartheta_0 \in (0,\pi/2)$ and a sequence $(\eps_n)$ such that $\eps_n \longrightarrow 0$ as $n \to \infty$ and such that $g_{\eps_n}(0)$ belongs to the sector $\{ z \in \C: |\arg(z) - \pi| > \vartheta_0 \}$ for all $n \in \N$. 
Since, moreover, $f_{\eps}(0)-g_{\eps}(0) \longrightarrow 0$ as $\eps \searrow 0$ by virtue of~\eqref{eq: gell-mann--low, schritt 3}, it follows that $f_{\eps_n}(0) \in \C \setminus (-\infty,0]$ for sufficiently large $n \in \N$. So, $0$ is indeed an accumulation point of $E$, and we are done.
\end{proof}

In physics, the switching function is almost always chosen as $\kappa(t) = e^t$ for $t \in (-\infty,0]$. 
And for that special choice of $\kappa$, 
an alternative formula for the energy shift can be deduced 
from the corollary above, namely
\begin{align} \label{eq: energy shift, kappa = exp}
\lambda_j(0)-\lambda_j(-\infty) = \lim_{\eps \searrow 0} \eps \, \frac{d}{d\mu} \Big(  \log \scprd{x', (U_{\eps}^{\mu})^I(0,-\infty)x}  \Big) \Big|_{\mu = 1},
\end{align}
where $U_{\eps}^{\mu}$ is the evolution system for $\frac 1 \eps A^{\mu}$ on $D$ with $A^{\mu}(t) := A_0 + \mu \, \kappa(t) V = A_0 + \mu \, e^t \, V$ for $t \in (-\infty,0]$ and $\mu \in (0,1]$ and where 
\begin{align*}
(U_{\eps}^{\mu})^I(t,s) 
:= e^{A_0 t/\eps} \, U_{\eps}^{\mu}(t,s)\, e^{A_0 s/\eps} \qquad ((s,t) \in \Delta_{(-\infty,0]}).
\end{align*}
It seems 
that~\eqref{eq: energy shift, kappa = exp} is 
used more often in the physics literature than~\eqref{eq: energy shift, allg kappa}. See, for instance, \cite{FetterWalecka}. 
In order to deduce~\eqref{eq: energy shift, kappa = exp} from the corollary above, one has only to notice that 
$A^{\mu}(t) 
= A(t + \log \mu)$ for all $t \in (-\infty,0]$ and $\mu \in (0,1]$. So, 
\begin{align}
U_{\eps}^{\mu}(t,s) = U_{\eps}(t+\log \mu, s + \log \mu) \qquad ((s,t) \in \Delta_{(-\infty,0]})
\end{align}
and therefore one sees for vectors $x, x' \in P_j(-\infty)H \subset \ker(A_0 - \lambda_j(-\infty))$ that 
\begin{align*}
\scprd{x', (U_{\eps}^{\mu})^I(0,-n)x} &= \scprd{ x',  e^{A_0 (\log\mu)/\eps} \, U_{\eps}^I(\log\mu,-n+\log\mu) \, e^{-A_0 (\log\mu)/\eps}x } \\
&= \scprd{ x',  U_{\eps}^I(\log\mu,-n+\log\mu) x }
\end{align*} 
for all $\mu \in (0,1]$ and $n \in \N$. 
Consequently, 
\begin{align}
\scprd{x', (U_{\eps}^{\mu})^I(0,-\infty)x} = \scprd{ x',  U_{\eps}^I(\log\mu,-\infty) x } = f_{\eps}(\log\mu)
\end{align}
for all $\mu \in (0,1]$ with $f_{\eps}$ defined as in~\eqref{eq: energy shift, def f_eps}, 
so that the corollary above and its proof yield the desired alternative formula~\eqref{eq: energy shift, kappa = exp} for the energy shift. 

\section*{Acknowledgement}

I would like to thank Marcel Griesemer for numerous discussions and for introducing me to adiabatic theory in the first place. I would also like to thank the German Research Foundation (DFG) for financial support through the research training group ``Spectral theory and dynamics of quantum systems'' (GRK 1838).

%
%
%





\begin{thebibliography}{}
\bibitem{AbouSalemFroehlich05} W. Abou Salem, J. Fr\"ohlich: \emph{Adiabatic theorems and reversible isothermal processes.} Lett. Math. Phys. \textbf{72} (2005), 153-163.
\bibitem{AbouSalem07} W. Abou Salem: \emph{On the quasi-static evolution of nonequilibrium steady states.} Ann. Henri Poincar\'{e} \textbf{8} (2007), 569-596.
\bibitem{AbouSalemFroehlich07} W. Abou Salem, J. Fr\"ohlich: \emph{Adiabatic theorems for quantum resonances.} Comm. Math. Phys. \textbf{237} (2007), 651-675.
\bibitem{AlickiFannes01} R. Alicki, M. Fannes: \emph{Quantum dynamical systems.} Oxford Univ. Press, 2001. 
\bibitem{AlickiLendi07} R. Alicki, K. Lendi: \emph{Quantum dynamical semigroups and applications.} 2nd edition, Springer, 2007.
\bibitem{ArendtBatty} W. Arendt, C. Batty, M. Hieber, F. Neubrander: \emph{Vector-valued Laplace transforms and Cauchy problems.} 2nd edition, Birkh\"auser, 2012.
\bibitem{AttalJoyePillet} S. Attal, A. Joye, C.-A. Pillet (editors): \emph{Open quantum systems I-III.} Lecture Notes in Mathematics \textbf{1880-1882}. Springer, 2006.
\bibitem{AvronElgart99} J. E. Avron, A. Elgart: \emph{Adiabatic theorem without a gap condition.} Commun. Math. Phys. \textbf{203} (1999), 445-463.
\bibitem{AvronGraf12} J. E. Avron, M. Fraas, G. M. Graf, P. Grech: \emph{Adiabatic theorems for generators of contracting evolutions.} Commun. Math. Phys. \textbf{314} (2012), 163-191.
\bibitem{AvronSeilerYaffe87} J. E. Avron, R. Seiler, L. G. Yaffe: \emph{Adiabatic theorems and applications to the quantum Hall effect.} Commun. Math. Phys. \textbf{110} (1987), 33-49. (In conjunction with the corresponding erratum of 1993.)
%
\bibitem{BornFock28} M. Born, V. Fock: \emph{Beweis des Adiabatensatzes.} Z. Phys. \textbf{51} (1928), 165-180.
\bibitem{Bornemann98} F. Bornemann: \emph{Homogenization in time of singularly perturbed mechanical systems.} Lecture Notes in Mathematics \textbf{1687}, Springer 1998.
\bibitem{Panati08} C. Brouder, G. Panati, G. Stoltz: \emph{Adiabatic approximation, Gell-Mann and Low theorem and degeneracies: a pedagogical example.} Phys. Rev. A \textbf{78} (2008).  
\bibitem{Panati10} C. Brouder, G. Panati, G. Stoltz: \emph{Gell-Mann and Low formula for degenerate unperturbed states.} Ann. Henri Poincar\'{e} \textbf{10} (2010), 1285-1309.
%
\bibitem{ChebotarevFagnola98} A. M. Chebotarev, F. Fagnola: \emph{Sufficient conditions for conservativity of minimal quantum dynamical semigroups.} J. Funct. Anal. \textbf{153} (1998), 382-404.
\bibitem{Conway:fana} J. B. Conway: \emph{A Course in functional analysis.} 2nd edition. Springer, 1990.
%
\bibitem{DaleckiiKrein50} J. L. Daleckii, S. G. Krein: \emph{On differential equations in Hilbert space.} Ukrain. Mat. Z. \textbf{2} (1950), 71-91.
\bibitem{Davies77} E. B. Davies: \emph{Quantum dynamical semigroups and the neutron diffusion equation.} Rep. Math. Phys. \textbf{11} (1977), 169-188. 
\bibitem{Dorroh75} J. R. Dorroh: \emph{A simplified proof of a theorem of Kato on linear evolution equations} J. Math. Soc. Japan \textbf{27} (1975), 474-478.
\bibitem{DunfordSchwartz} N. Dunford, J. T. Schwartz: \emph{Linear operators I-III.} Wiley, 1958, 1963, 1971.
%
\bibitem{ElgartHagedorn11} A. Elgart, G. A. Hagedorn: \emph{An adiabatic theorem for resonances.} Comm. Pure Appl. Math. \textbf{64} (2011), 1029-1058.
\bibitem{EngelNagel} K.-J. Engel, R. Nagel: \emph{One-parameter semigroups for linear evolution equations.} Springer, 2000.
%
\bibitem{FetterWalecka} A. L. Fetter, J. D. Walecka: \emph{Quantum theory of many-particle systems.} McGraw-Hill, 1971. 
\bibitem{FishmanSoffer16} S. Fishman, A. Soffer: \emph{Slowly changing potential problems in quantum mechanics: adiabatic theorems, ergodic theorems, and scattering.} J. Math. Phys. \textbf{57} (2016), 072101.  
\bibitem{FrankGang17} R. Frank, Z. Gang: \emph{A nonlinear adiabatic theorem for the one-dimensional Landau-Pekar system.} In preparation. (See the corresponding abstract in the Oberwolfach report no. 27/2017)
%
\bibitem{GangGrech17} Z. Gang, P. Grech: \emph{Adiabatic theorem for the Gross-Pitaevskii equation.} Comm. Partial Differential Equations \textbf{42} (2017), 731-756.
\bibitem{Garrido64} L. M. Garrido: \emph{Generalized adiabatic invariance.} J. Math. Phys. \textbf{5} (1964), 335-362. 
\bibitem{Gell-MannLow51} M. Gell-Mann, F. Low: \emph{Bound states in quantum field theory.} Phys. Rev. \textbf{84}(2) (1951), 350-354.
\bibitem{GesztesyTkachenko09} F. Gesztesy, V. Tkachenko: \emph{A criterion for Hill operators to be spectral operators of scalar type.} J. Anal. Math. \textbf{107} (2009), 287-353.
\bibitem{GohbergGoldbergKaashoek} I. Gohberg, S. Goldberg, M. A. Kaashoek: \emph{Classes of linear operators I-II.} Birkh\"auser, 1990, 1993.
\bibitem{GrossRunge} E. K. U. Gross, E. Runge, O. Heinonen: \emph{Many-particle systems.} Adam Hilger, 1991. 
%
\bibitem{HansonJoyePautratRaquepas17} E. Hanson, A. Joye, Y. Pautrat, R. Raqu\'{e}pas: \emph{Landauer's principle in repeated interaction systems.} Comm. Math. Phys. \textbf{349} (2017), 285-327.
\bibitem{HwangPechukas77} J.-T. Hwang, P. Pechukas: \emph{The adiabatic theorem in the complex plane and the semiclassical calculation of non-adiabatic transition amplitudes.} J. Chem. Phys. \textbf{67} (1977), 4640-4653.
%
%
\bibitem{JaksicPillet14} V. Jak\v{s}i\'{c}, C.-A. Pillet: \emph{A note on the Landauer principle in quantum statistical mechanics.} J. Math. Phys. \textbf{55} (2014), 075210.
\bibitem{JoyePfister91} A. Joye, C.-E. Pfister: \emph{Exponentially small adiabatic invariant for the Schr\"odinger equation.} Commun. Math. Phys. \textbf{140} (1991), 15-41.
\bibitem{JoyePfister93} A. Joye, C.-E. Pfister: \emph{Superadiabatic evolution and adiabatic transition probability between two non-degenerate levels isolated in the spectrum.} J. Math. Phys. \textbf{34} (1993), 454-479.
\bibitem{Joye07} A. Joye: \emph{General adiabatic evolution with a gap condition.} Commun. Math. Phys. \textbf{275} (2007), 139-162. 
%
\bibitem{Kato50} T. Kato: \emph{On the adiabatic theorem of quantum mechanics.} J. Phys. Soc. Japan \textbf{5} (1950), 435-439.
\bibitem{Kato53} T. Kato: \emph{Integration of the equation of evolution in a Banach space.} J. Math. Soc. Japan \textbf{5} (1953), 208-234.
\bibitem{Kato70} T. Kato: \emph{Linear evolution equations of ``hyperbolic'' type.} J. Fac. Sci. Univ. Tokyo \textbf{17} (1970), 241-258.
\bibitem{Kato73} T. Kato: \emph{Linear evolution equations of ``hyperbolic'' type II.} J. Math. Soc. Japan \textbf{25} (1973), 648-666.
\bibitem{KatoPerturbation80} T. Kato: \emph{Perturbation theory for linear operators.} 2nd edition. Springer, 1980.
\bibitem{Kato85} T. Kato: \emph{Abstract differential equations and nonlinear mixed problems.} Lezioni Fermiane, Accademia Nazionale dei Lincei, Scuola Normale  Superiore, Pisa (1985), 1-89.
\bibitem{KelerTeufel12} J. v. Keler, S. Teufel: \emph{Non-adiabatic transitions in a massless scalar field.} arXiv:1204.0344 (2012).
\bibitem{Kisynski63} J. Kisy\'{n}ski: \emph{Sur les op\'{e}rateurs de Green des probl\`{e}mes de Cauchy abstraits.} Stud. Math. \textbf{23} (1963), 285-328.
\bibitem{Kraus71} K. Kraus: \emph{General state changes in quantum theory.} Ann. Phys. \textbf{64}, 311-335. 
\bibitem{Krein} S. G. Krein: \emph{Linear differential equations in Banach space.} Transl. Math. Monographs, American Mathematical Society, 1971.
%
\bibitem{Lenard59} A. Lenard: \emph{Adiabatic invariance to all orders.} Ann. Phys. \textbf{6} (1959), 261-276.
\bibitem{Lindblad76} G. Lindblad: \emph{On the generators of quantum dynamical semigroups.} Commun. Math. Phys. \textbf{48} (1976), 119-130.
%
%
\bibitem{Nenciu80} G. Nenciu: \emph{On the adiabatic theorem of quantum mechanics.} J. Phys. A: Math. Gen. \textbf{13} (1980), 15-18.
\bibitem{NenciuRasche89} G. Nenciu, G. Rasche: \emph{Adiabatic theorem and Gell-Mann--Low formula.} Helv. Phys. Acta \textbf{62} (1989), 372-388.
\bibitem{NenciuRasche92} G. Nenciu, G. Rasche: \emph{On the adiabatic theorem for non-self-adjoint Hamiltonians.} J. Phys. A: Math. Gen. \textbf{25} (1992), 5741-5751. 
\bibitem{Nenciu93} G. Nenciu: \emph{Linear adiabatic theory. Exponential estimates.} Commun. Math. Phys. \textbf{152} (1993), 479-496.
\bibitem{NickelSchnaubelt98} G. Nickel, R. Schnaubelt: \emph{An extension of Kato's stability condition for non-autonomous Cauchy problems.} Taiw. J. Math. \textbf{2} (1998), 483-496.
\bibitem{Nickel00} G. Nickel: \emph{Evolution semigroups and product formulas for nonautonomous Cauchy problems.} Math Nachr. \textbf{212} (2000), 101-115.
%
\bibitem{Pazy} A. Pazy: \emph{Semigroups of linear operators and applications to partial differential equations.} Springer, 1983.
%
\bibitem{ReedSimon} M. Reed, B. Simon: \emph{Methods of modern mathematical physics I-IV.} Academic Press, 1980, 1975, 1979, 1978.
%
\bibitem{Sancho66} F. J. Sancho: \emph{$m$th order adiabatic invariance for quantum systems.} Proc. Phys. Soc. \textbf{89} (1966), 1-5.
\bibitem{dipl} J. Schmid: \emph{Adiabatens\"atze mit und ohne Spektrall\"uckenbedingung.} Master's thesis, Universit\"at Stuttgart. arXiv:1112.6338 (2011).
\bibitem{SchmidGriesemer14} J. Schmid, M. Griesemer: \emph{Kato's theorem on the integration of non-autonomous evolution equations.} Math. Phys. Anal. Geom. \textbf{17} (2014), 265-271.
\bibitem{Schmid15qmath} J. Schmid: \emph{Adiabatic theorems with and without spectral gap condition for non-semisimple spectral values.} Conf. Proc. QMath 12 (2014), 355-362. 
\bibitem{diss} J. Schmid: \emph{Adiabatic theorems for general linear operators and well-posedness of linear evolution equations.} PhD thesis, Universit\"at Stuttgart. http://dx.doi.org/10.18419/opus-5178   (2015).
\bibitem{SchmidJEE} J. Schmid: \emph{Well-posedness of non-autonomous linear evolution equations for generators whose commutators are scalar.}  J. Evol. Equ. \textbf{16} (2016), 21-50.
\bibitem{Sparber16} C. Sparber: \emph{Weakly nonlinear time-adiabatic theory.} Ann. Henri Poincar\'{e} \textbf{17} (2016), 913-936.
\bibitem{Sz-Nagy67} B. Sz.-Nagy: \emph{Spektraldarstellung linearer Transformationen des Hilbertschen Raumes.} Springer, 1967.
%
\bibitem{Taylor58} A. E. Taylor: \emph{Introduction to functional analysis.} Wiley, 1958.
\bibitem{TaylorLay80} A. E. Taylor, D. C. Lay: \emph{Introduction to functional analysis.} 2nd edition. Wiley, 1980.
\bibitem{Teschl09} G. Teschl: \emph{Mathematical methods in quantum mechanics: with applications to Schr\"odinger operators.} American Mathematical Society, 2009.
\bibitem{Teufel01} S. Teufel: \emph{A note on the adiabatic theorem without gap condition.} Lett. Math. Phys. \textbf{58} (2001), 261-266.
\bibitem{Teufel03} S. Teufel: \emph{Adiabatic perturbation theory in quantum dynamics.} Lecture Notes in Mathematics \textbf{1821}. Springer, 2003.
%
\bibitem{Wermer54} J. Wermer: \emph{Commuting spectral measures on Hilbert space.} Pacific J. Math. \textbf{4} (1954), 355-361.
%
\bibitem{Yosida} K. Yosida: \emph{Functional analysis.} 6th edition. Springer, 1980.
%

\end{thebibliography}
\end{document}